%% file: main.tex
\documentclass[a4paper,twocolumn,11pt, unpublished, noarxiv]{quantumarticle}
\pdfoutput=1

\input{macros}
\input{tikz_preamble}
\begin{document}
\title{An equivalence between time-symmetry and cyclic causality in quantum theory}

\author{Eliot Jean}
\affiliation{Institute for Theoretical Physics, ETH Z{\"u}rich, 8093 Z{\"u}rich, Switzerland}
\author{Ralph Silva}
\affiliation{Institute for Theoretical Physics, ETH Z{\"u}rich, 8093 Z{\"u}rich, Switzerland}
\author{V. Vilasini}
\affiliation{Universit\'{e} Grenoble Alpes, Inria, 38000 Grenoble, France}
\affiliation{Institute for Theoretical Physics, ETH Z{\"u}rich, 8093 Z{\"u}rich, Switzerland}

\maketitle

\onecolumn
\begin{center}
  \begin{minipage}{0.85\textwidth}
{\bf Abstract} Understanding the relationship between the time-symmetric nature of physical laws and the apparent directionality of causality is a central question in quantum foundations. The standard operational formulation, widely used in quantum information, imposes a definite, acyclic causal order on agents’ operations, contrasting with time-symmetric dynamics. Two prominent extensions of this framework are the multi-time state (MTS) formalism, which incorporates time symmetry via arbitrary pre- and post-selection, and the post-selected closed timelike curve (P-CTC) framework, which enables cyclic causal influences through post-selection on maximally entangled states. While prior work has noted structural connections between MTS and P-CTCs, it remained unclear whether an operational equivalence exists, or whether constructive mappings can be established between their most general objects. In this work, we address this gap by extending the P-CTC framework to define time-labelled P-CTC assisted combs, a more general class of P-CTC-assisted objects that support open processing slots and explicit temporal structure. We prove that for every (possibly mixed) MTS, there exists an operationally equivalent time-labelled P-CTC assisted comb, and vice versa. The equivalence is shown via explicit mappings, while discussing the number and dimensionality of the P-CTCs involved. We also explore a resource-theoretic view of MTS, defining a partial order under free transformations that do not use P-CTCs. We conclude by discussing future directions informed by the  operational equivalence between time symmetry and cyclic causality established here.

  \end{minipage}
\end{center}

\begingroup
\renewcommand\thefootnote{}
\footnotetext{\hspace{-2em}\emph{Author names appear in alphabetical order, a summary of author contributions can be found at the end of the manuscript.}}
\addtocounter{footnote}{-1}
\endgroup

\vspace{1em}
\twocolumn

\newpage
\tableofcontents
\newpage

\input{1-Intro/1-Intro}
\input{2-MTF/2-MTF}

\input{3-P-CTC/P-CTC_new}
\input{4-Connection/4-Connection}
\input{5-Discussion/5-Discussion}


\onecolumn
\appendix

\input{6-Appendix/6-0-MTF}
\input{6-Appendix/PCTC_probabilities}

\input{6-Appendix/PCTC-Combs}

\input{6-Appendix/6-1-Connection}

\input{6-Appendix/PartialOrder}

\end{document}

%% file: macros.tex
\usepackage[utf8]{inputenc}
\usepackage[english]{babel}
\usepackage[T1]{fontenc}

\usepackage{amsfonts}
\usepackage{bbold} 

\usepackage{subcaption}

\usepackage{amsmath}
\usepackage{dsfont}
\usepackage{braket}
\usepackage{mathtools}
\usepackage{amsthm}
\usepackage{thmtools, thm-restate} 

\usepackage[usenames,dvipsnames]{xcolor}
\usepackage{graphicx}
\usepackage[ colorlinks = true, 
             linkcolor = blue,
             urlcolor  = blue,
             citecolor = orange,
             anchorcolor = ForestGreen,
]{hyperref}
\usepackage[capitalise]{cleveref}

\usepackage{tikz}
\usepackage{pgf}
\usetikzlibrary{calc, math,positioning,shapes, decorations.pathmorphing, matrix,backgrounds,fit, intersections, arrows.meta, shapes.geometric, bending}
\usepackage{qcircuit2}
\usetikzlibrary{decorations.pathreplacing}

\usepackage{booktabs}
\usepackage{multirow}

\usepackage{comment}

\usepackage{algorithm2e}
\usepackage{listings}

\usepackage[numbers,sort&compress]{natbib}

\newcommand{\proj}[2]{\ket{#1}\bra{#2}}
\newcommand{\pure}[1]{\proj{#1}{#1}}

\newcommand{\hilbert}{\mathcal{H}}
\newcommand{\preH}[1]{\hilbert^{#1}}
\newcommand{\postH}[1]{\hilbert_{#1}}
\newcommand{\prepostH}[2]{\hilbert^{#1}_{#2}}
\newcommand{\preS}[2]{\ket{#1}^{#2}}
\newcommand{\postS}[2]{\bra{#1}_{#2}}

\newcommand{\id}{\mathbb{1}}
	
\newcommand{\tr}{\operatorname{Tr}}
\DeclareMathOperator{\Tr}{Tr}

\newcommand{\abs}[1]{\left| #1 \right|}

\newtheorem{theorem}{Theorem}[section]

\newtheorem{lemma}[theorem]{Lemma}
\newtheorem{proposition}[theorem]{Proposition}

\newtheorem{remark}[theorem]{Remark}
\newtheorem{definition}[theorem]{Definition}

\tikzmath{
    \WidthPrePost = 1cm;
    \WidthMeasurement = 0.65cm;
    \HeightMeasurement = 0.65cm;
}

\makeatletter
\newsavebox{\@brx}
\newcommand{\llangle}[1][]{\savebox{\@brx}{\(\m@th{#1\langle}\)}%
  \mathopen{\copy\@brx\kern-0.5\wd\@brx\usebox{\@brx}}}
\newcommand{\rrangle}[1][]{\savebox{\@brx}{\(\m@th{#1\rangle}\)}%
  \mathclose{\copy\@brx\kern-0.5\wd\@brx\usebox{\@brx}}}
\makeatother

%% file: 1-Intro/1-Intro.tex
\section{Introduction}
\label{sec:Intro}

Understanding the nature of time and causality lies at the core of several fundamental questions in physics, particularly in quantum theory. While we experience a distinct arrow of time, marking a clear distinction between past and future, the fundamental laws of quantum mechanics, such as the Schrödinger equation, are time-symmetric. However, the operational formulation of quantum theory, which accounts for measurements and interactions involving agents, introduces a notable asymmetry. In this framework, there exists a clear and acyclic causal ordering between operations, aligned with our experience of time. This asymmetry is often attributed to the agents' actions: they can freely prepare initial states at an earlier time, which then evolve into final states at a later time or causally influence measurement outcomes at a later time, but not the reverse.

The multi-time state (MTS) formalism \cite{Aharonov1964, Aharonov1991, Silva2014} and the post-selected closed timelike curves (P-CTC) framework \cite{Lloyd2011, Lloyd2011a} correspond to two significant extensions of this standard operational formulation of quantum theory, but in different ways. The MTS formalism is a time-symmetric yet operational approach that accommodates arbitrary pre- and post-selections, enabling states to evolve both forward and backward in time, while considering how this constrains intermediate measurements. On the other hand, the P-CTC framework allows for cyclic causal influences, where a system can affect its own past. Inspired by exotic solutions in general relativity that permit closed timelike curves, the P-CTC framework abstracts away spacetime and relativistic details, focusing instead on the information-theoretic mechanisms underlying causal loops, which are modeled through specific forms of pre- and post-selections. This raises natural questions: how are these two extensions of the standard quantum framework related? How are time symmetry and cyclic causality linked in quantum theory?

Interestingly, a third extension of the usual quantum formalism, corresponding to so-called indefinite causal order (ICO) processes \cite{Oreshkov2012, Chiribella2013}, has been shown to relate to both MTS \cite{Silva2017} and P-CTC \cite{Araujo2017a} frameworks. ICO processes allow abstract protocols where quantum operations are applied in a quantum superposition of orders on a target system. Previous work has demonstrated that ICO processes form a special linear subset of both the MTS and P-CTC frameworks (noting that the latter frameworks also allow for more general objects associated with non-linear measurement probabilities) \cite{Silva2017, Araujo2017}. Additionally, a range of studies exploring time symmetry, causal loops and retrocausality in quantum theory have suggested links between these concepts (see e.g., \cite{Coecke_2012a, Coecke_2012b, Oreshkov_2015, Leifer_2017, pinzani2019, selby2024, Barrett_2021}). Moreover, the computational power of quantum theory with arbitrary pre and post-selection (analogous to MTS) and of quantum theory assisted by P-CTCs have also been studied, suggesting that the associated complexity classes coincide \cite{aaronson2004, Brun2012, Araujo2017}.

Despite these known connections and similarities, the two frameworks exhibit key structural differences. While it is known that each P-CTC can be seen as an MTS with maximally entangled pre- and post-selections, the reverse direction is non-trivial as MTS allow arbitrary pre and post-selection and in arbitrary combinations, while P-CTCs involve pre and post selections on maximally entangled states that appear in pairs. Moreover, the MTS formalism features "slots" between pre- and post-selection stages, where external operations corresponding to arbitrary linear operators can be inserted, and the formalism explicitly incorporates time labels. In contrast, the original formulation of the P-CTC framework considers P-CTC assisted maps and circuits which lack such empty slots and does not inherently associate operations with time labels. Moreover, the formalism involves operators corresponding to unitary channels which are connected together via CTCs (pre and post selections on maximally entangled states), as opposed to generic linear operators.
These distinctions suggest that MTS might, at first glance, appear more general and expressive than P-CTCs.

To our knowledge, no counterexamples demonstrate an MTS scenario that cannot be captured within the P-CTC framework, nor has an explicit equivalence between their most general objects been established. Furthermore, it remains unclear whether a constructive protocol exists to systematically map any MTS into an operationally equivalent object within the P-CTC framework. We address such questions here. Notably, analogous structures to both MTS and P-CTC formalisms appear in relativistic models with spatio-temporal boundary conditions and in approaches to quantum gravity (e.g., \cite{Feynman1948, Oeckl_2003a, Oeckl_2003b}). Thus, establishing a rigorous and operationally grounded link between these frameworks may offer a conceptual and technical foundation for future explorations of causality and time at the quantum-relativity interface.

{\bf Summary of the paper and contributions} In \cref{sec:MTF} and \cref{sec:PCTC}, we begin by reviewing the MTS and P-CTC formalisms, while refining and extending them along the way. We clarify the role of pure versus mixed objects in the MTS formalism and introduce the concept of (time-labelled) P-CTC-assisted combs, representing P-CTC-assisted quantum circuits with ``open slots'' (full details given in \cref{appendix: PCTC_Combs}).
In \cref{sec:PCTC-MTF} and \cref{sec:Connection}, we establish the following theorem, which implies an operational equivalence between the MTF and P-CTC formalisms, or between time-symmetry and cyclic causality in quantum theory.

\begin{restatable}[]{theorem}{Main}
\label{thm: main}
For every (possibly mixed) multi-time object, there exists an operationally equivalent time-labelled P-CTC assisted comb, and vice versa. 
\end{restatable}
The P-CTC to MTS directed of the theorem follows easily from known results (summarised in \cref{sec:PCTC-MTF}). The reverse direction from MTS to P-CTCs addressed in \cref{sec:Connection} is the novel contribution of this paper, and is established via several intermediate results and explicit constructions in three stages, that might be of independent interest (see \cref{fig:summary} for an illustrative summary). At each stage, we analyze the number and dimensions of the P-CTCs required in the mapping and discuss the role of the time-labels. 
\begin{itemize}
\item {\bf Pure two-time operators to P-CTC assisted maps} First, we consider two time operators (2TO) which are a class of multi-time objects with a structure where post-selections occur at earlier times than pre-selections. This is in contrast to the class of two time states (2TS) where pre-selections occur at the earlier time, before the post-selections. We map an arbitrary pure 2TO to an operationally equivalent P-CTC-assisted unitary map. We show an explicit mapping which utilizes two P-CTCs, one being 2-dimensional and another of the same Hilbert dimension $d$ as the 2TO systems, and prove the existence of a mapping that uses only a single P-CTC of $d$-dimensions. 
(\cref{subsubsec:2TO-PCTC}) 
\item {\bf Pure multi-time states to two-time operators} Next, we connect pure 2TOs to general pure multi-time states (MTS) in \cref{subsubsec:2TO-MTS}, while distinguishing the cases where time-labels are preserved vs ignored. Together with the first step, this allows any pure MTS to be mapped to a P-CTC-assisted object (via a 2TO): a P-CTC assisted map when time-labels are ignored, and a time-labelled P-CTC assisted comb, when the time labels are preserved.
\item  {\bf Mixed multi-time states to P-CTC assisted objects} Finally, in \cref{subsec:Mixed}, we generalise the results to mixed MTS, i.e., probabilistic mixtures of pure MTS. We show that having a P-CTC construction for every pure MTS in the decomposition (as ensured by the previous steps) implies a P-CTC construction for the mixed MTS, and therefore \cref{thm: main}. 
\end{itemize}

In particular, these results inform the operational characterisation of the MTS formalism \cite{Silva2017} by showing that any multi-time instrument (a multi-time generalisation of quantum instruments) can be operationally prepared in quantum theory with pre and post selection (see \cref{sec:mts:operational}).

Motivated by these results, in \cref{sec: partialorder}, we consider transformations between multi-time states by means of free operations that do not utilise any P-CTCs. Such considerations can inform a resource-theoretic understanding of multi-time states, with P-CTCs taken as the resource. Specifically, we define a partial order between MTS that have the same coefficients but may differ in the temporal structure (\emph{isomorphic MTS}), where the order indicates whether an MTS can be transformed to another isomorphic MTS without using any P-CTCs. 2TS emerge as most useful and 2TOs as least useful in this partial order, with all other isomorphic MTS lying in between. Finally, in \cref{sec:Discussion}, we discuss future avenues for research.

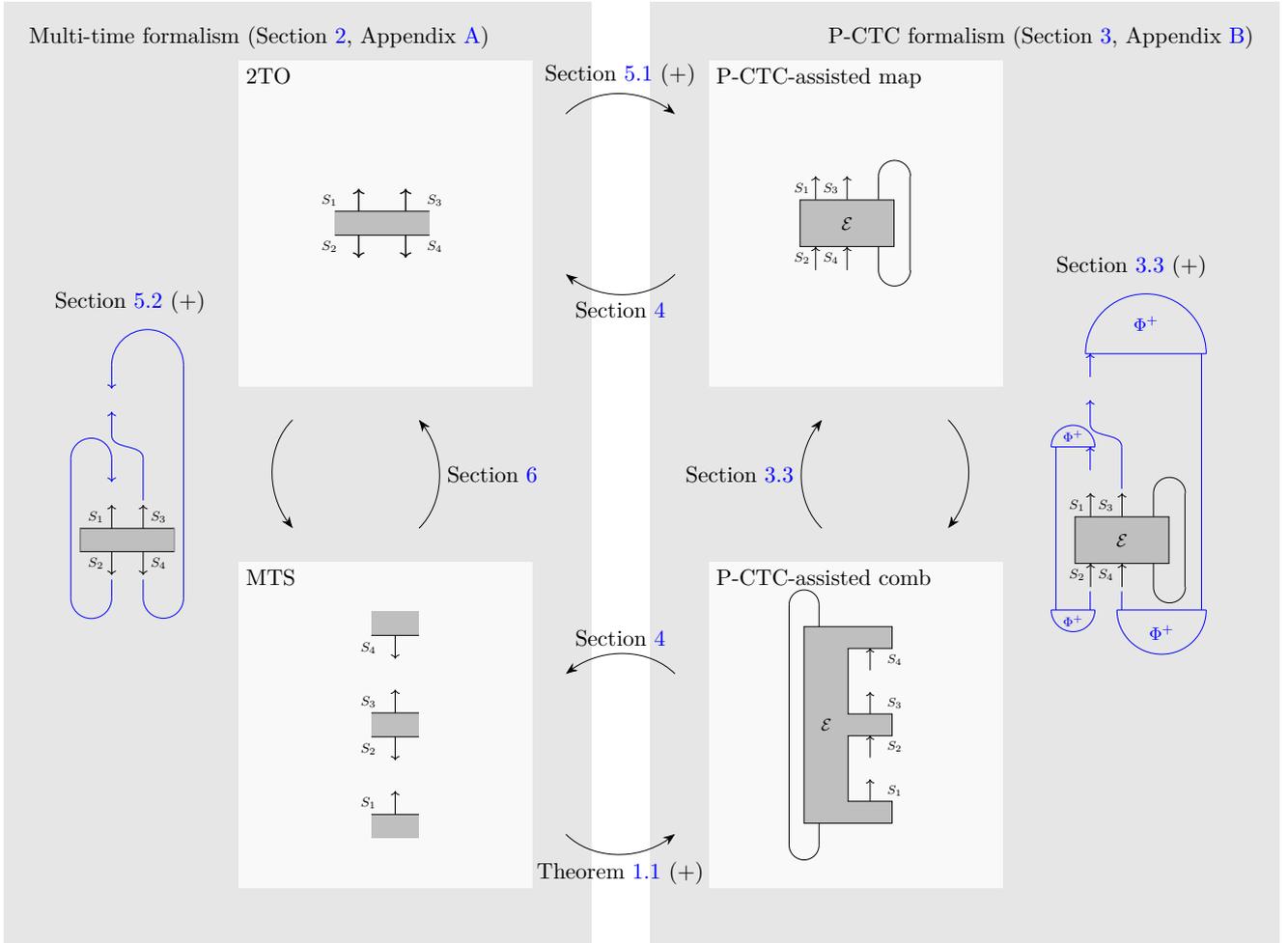
\begin{figure*}
    \centering
    \scalebox{0.8}{\input{1-Intro/Figs/Summary_Pics}}
    \caption{Summary of the main results and relationships between the objects of the multi-time formalism and the P-CTC framework. Each arrow represents a relationship and includes a reference to the corresponding section of the paper. The results from sections~\cref{sec:PCTC} and~\ref{sec:PCTC-MTF} appearing here are primarily based on the previous literature or immediate implications thereof. The novel contributions from this work are the connections established in \cref{sec:Connection} and~\cref{sec: partialorder}. The $+$ in parentheses indicates that P-CTCs are required to transition from the source object to the target object, while the absence of $+$ indicates that a construction without any P-CTCs exists.}
    \label{fig:summary}
\end{figure*}

%% file: 1-Intro/Figs/Summary_Pics.tex
\begin{tikzpicture}[
  node distance=2cm,
  myrect/.style={
    draw=none,
    fill=gray!5,
    minimum width=5cm,
    minimum height=5.6cm,
  },
  myarrow/.style={
    -{Stealth[bend, length=2mm]},
    shorten >=1mm,
    shorten <=1mm,
  },
  bigrect/.style={
    draw=none,
    fill=gray!20
  }
]

\node[myrect] (rect1) {%
  \scalebox{0.8}{\input{1-Intro/Figs/Fig_2to}}%
};
\node[below right] at (rect1.north west){2TO};
 
\node[myrect, right= 3cm of rect1] (rect2) {%
  \scalebox{0.8}{\input{1-Intro/Figs/fig_pctc_map}}%
};
\node[below right] at (rect2.north west){P-CTC-assisted map};

\node[myrect, below= 3cm of rect1] (rect3) {%
  \scalebox{0.8}{\input{1-Intro/Figs/fig_mts_test}}%
};
\node[below right] at (rect3.north west){MTS};

\node[myrect, right= 3cm of rect3] (rect4) {%
  \scalebox{0.75}{\input{1-Intro/Figs/fig_pctc_comb}}%
};
\node[below right] at (rect4.north west){P-CTC-assisted comb};

\begin{scope}[myarrow]
  \draw ([xshift=0.5cm, yshift=-1cm]rect1.north east) to[bend left=45] node[midway,above] {\cref{subsubsec:2TO-PCTC} (+)} ([xshift=-0.5cm, yshift=-1cm]rect2.north west);
  
  \draw ([xshift=1cm, yshift=-0.5cm]rect1.south west) to[bend right=45] node[midway,left] (arrow1) {} ([xshift=1cm, yshift=+0.5cm]rect3.north west);

  \draw ([xshift=0.5cm, yshift=1cm]rect3.south east) to[bend right=45] node[midway,below] {\cref{thm: main} (+)} ([xshift=-0.5cm, yshift=1cm]rect4.south west);
  
  \draw ([xshift=-1cm, yshift=-0.5cm]rect2.south east) to[bend left=45] node[midway,right] (arrow2) {} ([xshift=-1cm, yshift=0.5cm]rect4.north east);
  
  \draw ([xshift=-0.5cm, yshift=2cm]rect2.south west) to[bend left=45] node[midway,below] {\cref{sec:PCTC-MTF}} ([xshift=0.5cm, yshift=2cm]rect1.south east);
  
  \draw ([xshift=-2cm, yshift=0.5cm]rect3.north east) to[bend right=45] node[midway,right] {\shortstack{
      \cref{sec: partialorder}
    }} ([xshift=-2cm, yshift=-0.5cm]rect1.south east);
  
  \draw ([xshift=-0.5cm, yshift=-2cm]rect4.north west) to[bend right=45] node[midway,above] {\cref{sec:PCTC-MTF}} ([xshift=0.5cm, yshift=-2cm]rect3.north east);
  
  \draw ([xshift=2cm, yshift=0.5cm]rect4.north west) to[bend left=45] node[midway,left] {\cref{subsec:comb}} ([xshift=2cm, yshift=-0.5cm]rect2.south west);
\end{scope}

\node[left=1cm of arrow1] (trafo_2to){
    \scalebox{0.8}{\input{1-Intro/Figs/fig_trafo_2to}}%
  };

\node[right=1cm of arrow2] (trafo_pctc) {
    \scalebox{0.8}{\input{1-Intro/Figs/fig_trafo_pctc}}%
  };

\begin{pgfonlayer}{background}
  \node[bigrect, fit=(rect1)(rect3)(trafo_2to), inner sep=1cm, label={[anchor=north west, xshift=0.3cm, yshift=-0.3cm]north west:Multi-time formalism (\cref{sec:MTF}, \cref{sec:AppendixMTF})}] (groupA) {};
\end{pgfonlayer}

\begin{pgfonlayer}{background}
  \node[bigrect, fit=(rect2)(rect4)(trafo_pctc), inner sep=1cm, label={[anchor=north east, xshift=-0.3cm, yshift=-0.3cm]north east:P-CTC formalism (\cref{sec:PCTC}, \cref{sec:AppendixPCTC})}] (groupB) {};
\end{pgfonlayer}

\node[above=0cm of trafo_2to] () {\cref{subsubsec:2TO-MTS} (+)};

\node[above=0cm of trafo_pctc] () {\cref{subsec:comb} (+)};

\end{tikzpicture}

%% file: 1-Intro/Figs/Fig_2to.tex
\begin{tikzpicture}[
    simple/.style={rectangle, minimum width=\WidthPrePost,minimum height=0.5cm},
    Kop/.style={draw, minimum size=\WidthMeasurement},
    ghostOp/.style={minimum size=\WidthMeasurement},
    t/.style={font=\scriptsize}
]
    \pgfsetmatrixrowsep{0.5cm}
    \pgfsetmatrixcolumnsep{0cm}
    \pgfmatrix{rectangle}{center}{mymatrix}
    {\pgfusepath{}}{\pgfpointorigin}{\let\&=\pgfmatrixnextcell}
    {
    \node[ghostOp](00){};\&\node[ghostOp](01){}; \\
    \node[simple](10){};\&\node[simple](11){};\\
    \node[ghostOp](20){};\&\node[ghostOp](21){};  \\
    }

    \begin{scope}[on background layer]
        \fill[lightgray] (10.north west) rectangle (11.south east);
    \end{scope}
    
    \draw[-] (10.north west)--(11.north east);
    \draw[-] (10.south west)--(11.south east);
    \draw[->, thick] (10.north) -- (00.south) node [midway, inner sep=0cm, minimum size=0.1cm] (a1) {};
    \draw[->, thick] (11.north) -- (01.south) node [midway,inner sep=0 cm, minimum size=0.1cm] (a3) {};
    \draw[->, thick] (10.south) -- (20.north) node [midway,inner sep=0 cm, minimum size=0.1cm] (a2) {};
    \draw[->, thick] (11.south) -- (21.north) node [midway,inner sep=0 cm, minimum size=0.1cm] (a4) {};

    \node[t, left=0.5mm of a1, inner sep=0pt, minimum size=1cm](S1){$S_1$};
    \node[t, left =.5 mm of a2, inner sep=0pt, minimum size=1cm](S2){$S_2$};
    \node[t, right =.5 mm of a3, inner sep=0pt, minimum size=1cm](S3){$S_3$};
    \node[t, right =.5 mm of a4, inner sep=0pt, minimum size=1cm](S4){$S_4$};

  \end{tikzpicture}

%% file: 1-Intro/Figs/fig_pctc_map.tex
\begin{tikzpicture}
    
    \node[rectangle, fill=lightgray, draw, minimum width = 2cm, minimum height= 1cm] (1) at (0,0) {$\mathcal{E}$};
    
    \draw[->] ([xshift=-0.67cm]1.north) -- ([xshift=-0.67cm, yshift=0.5cm]1.north) node[midway, inner sep=0cm, minimum size=0.1mm] (a1){};
    \draw[->] (1.north) -- ([yshift=0.5cm]1.north) node[midway, inner sep=0cm, minimum size=0.1mm] (a2){};
    \draw[->] ([yshift=-0.5cm]1.south) -- (1.south) node[midway, inner sep=0cm, minimum size=0.1mm] (a3){};
    \draw[->] ([xshift=-0.67cm, yshift=-0.5cm]1.south) -- ([xshift=-0.67cm]1.south) node[midway, inner sep=0cm, minimum size=0.1mm] (a4){};

    \draw[-] ([xshift=0.67cm, yshift=-0.5cm]1.south) -- ([xshift=0.67cm]1.south) node[midway, inner sep=0cm, minimum size=0.1mm] (s1){};
    \draw[-] ([xshift=0.67cm]1.north) -- ([xshift=0.67cm, yshift=0.5cm]1.north) node[midway, inner sep=0cm, minimum size=0.1mm] (s2){};

    \draw[-] ([xshift=1.34cm, yshift=-0.5cm]1.south) -- ([xshift=1.34cm, yshift=0.5cm]1.north);

    \draw ([xshift=0.67cm, yshift=0.5cm]1.north) arc (180:0:0.34);
    \draw ([xshift=0.67cm, yshift=-0.5cm]1.south) arc (-180:0:0.34);

    \node[font=\scriptsize, inner sep=0cm, anchor=center, minimum size=0.1mm] (s1) at (0.05, 0.75) {$S_1$};
    \node[font=\scriptsize, inner sep=0pt, anchor = center, minimum size=0.1mm] (s3) at (0.665,0.75) {$S_3$};
    \node[font=\scriptsize, inner sep=0pt, anchor = center, minimum size=0.1mm] (s4) at (0.665, -0.75){$S_4$};
    \node[font=\scriptsize, inner sep=0pt, anchor = center, minimum size=0.1mm] (s3) at (0.05, -0.75) {$S_2$};
    
\end{tikzpicture}

%% file: 1-Intro/Figs/fig_mts_test.tex
\begin{tikzpicture}

    \node[rectangle, fill=lightgray, minimum width = \WidthPrePost, minimum height= 0.5cm] (1) at (0,0) {};
    \draw[-] (1.north west) -- (1.north east);
    \draw[-] (1.south west) -- (1.south east);

    \node[minimum size = \WidthMeasurement, below=0.5cm of 1] (2) {};
    \node[minimum size = \WidthMeasurement, above=0.5cm of 1] (3) {};

    \node[rectangle, fill=lightgray, minimum width = \WidthPrePost, minimum height= 0.5cm, above= 0.5cm of 3] (post)  {};
    \draw[-] (post.south west) -- (post.south east);

    \node[rectangle, fill=lightgray, minimum width = \WidthPrePost, minimum height= 0.5cm, below= 0.5cm of 2] (pre) {};
    \draw[-] (pre.north west) -- (pre.north east);

    \draw[->] (pre.north) -- (2.south) node[midway, inner sep=0cm, minimum size=0.1mm] (a1){};
    \draw[->] (1.south) -- (2.north) node[midway, inner sep=0cm, minimum size=0.1mm] (a2){};
    \draw[->] (1.north) -- (3.south) node[midway, inner sep=0cm, minimum size=0.1mm] (a3){};
    \draw[->] (post.south) -- (3.north) node[midway, inner sep=0cm, minimum size=0.1mm] (a4){};

    \node[font=\scriptsize, left=0.5mm of a1, inner sep=0pt, minimum size=1cm]{$S_1$};
    \node[font=\scriptsize, left=0.5mm of a2, inner sep=0pt, minimum size=1cm]{$S_2$};
    \node[font=\scriptsize, left=0.5mm of a3, inner sep=0pt, minimum size=1cm]{$S_3$};
    \node[font=\scriptsize, left=0.5mm of a4, inner sep=0pt, minimum size=1cm]{$S_4$};
    
\end{tikzpicture}

%% file: 1-Intro/Figs/fig_pctc_comb.tex
\begin{tikzpicture}
    \draw[fill=lightgray] (0,0) -- (2,0) -- (2,0.5) -- (1,0.5) -- (1,2) -- (2,2) -- (2,2.5) -- (1,2.5) -- (1,4) -- (2,4) -- (2,4.5) -- (0,4.5) -- (0,0) -- cycle;
    \node[minimum size=0.1cm] (e) at (0.25,2.25){$\mathcal{E}$}; 

    \draw[->] (1.5,0.5) -- (1.5,1) node[midway, inner sep =0cm, minimum size=0.1mm] (a1) {};
    \draw[->] (1.5,1.5) -- (1.5,2) node[midway, inner sep =0cm, minimum size=0.1mm] (a2) {};
    \draw[->] (1.5,2.5) -- (1.5,3) node[midway, inner sep =0cm, minimum size=0.1mm] (a3) {};
    \draw[->] (1.5,3.5) -- (1.5,4) node[midway, inner sep =0cm, minimum size=0.1mm] (a4) {};

    \draw[-] (0.34,-0.5) -- (0.34,0);
    \draw[-] (0.34, 4.5) -- (0.34,5);
    \draw[-] (-0.34,-0.5) -- (-0.34, 5);
    \draw (0.34,5) arc (0:180:0.34);
    \draw (0.34,-0.5) arc (0:-180:0.34);

    \node[font=\scriptsize, right=0.3mm of a1, inner sep=0pt, minimum size=1cm]{$S_1$};
    \node[font=\scriptsize, right=0.3mm of a2, inner sep=0pt, minimum size=1cm]{$S_2$};\node[font=\scriptsize, right=0.3mm of a3, inner sep=0pt, minimum size=1cm]{$S_3$};
    \node[font=\scriptsize, right=0.3mm of a4, inner sep=0pt, minimum size=1cm]{$S_4$};
    
\end{tikzpicture}

%% file: 1-Intro/Figs/fig_trafo_2to.tex
\begin{tikzpicture}

    \draw[fill=lightgray] (0,0) -- (2,0) -- (2,0.5) -- (0,0.5) -- cycle;
    \draw[-, color=lightgray] (0,0) -- (0,0.5);
    \draw[-, color=lightgray] (2,0) -- (2,0.5);

    \draw[->] (0.67, 0.5) -- (0.67, 1) node[midway, inner sep=0.mm, minimum size=0.1mm](a1){};
    \draw[->] (1.34, 0.5) -- (1.34, 1) node[midway, inner sep=0.mm, minimum size=0.1mm](a2){};
    \draw[->] (0.67, 0) -- (0.67, -0.5) node[midway, inner sep=0.mm, minimum size=0.1mm](a3){};
    \draw[->] (1.34, 0) -- (1.34, -0.5) node[midway, inner sep=0.mm, minimum size=0.1mm](a4){};

    \draw[->, color=blue] (0.67, 2) -- (0.67, 1.5) node[midway, inner sep=0.mm, minimum size=0.1mm](a5){};
    \draw[->, color=blue] (0.67, 2.5) -- (0.67, 3) node[midway, inner sep=0.mm, minimum size=0.1mm](a6){};
    \draw[->, color=blue] (0.67, 4) -- (0.67, 3.5) node[midway, inner sep=0.mm, minimum size=0.1mm](a7){};

    \draw[-, color=blue] (0.67, -0.6) -- (0.67, -1);
    \draw[-, color=blue] (1.34, -0.6) -- (1.34, -1);
    \draw[-, color=blue] (-0.2, -1) -- (-0.2, 2);
    \draw[-, color=blue] (2.2, -1) -- (2.2, 4);
    \draw[-, color=blue] (1.34, 1.1) -- (1.34, 2);
    
    \draw[color=blue] (0.67,-1) arc (0:-180:0.435);
    \draw[color=blue] (1.34, -1) arc (-180:0:0.43);
    \draw[color=blue] (2.2, 4) arc (0:180:0.765);
    \draw[color=blue] (0.67, 2) arc (0:180:0.435);

    \draw[color=blue] (1.34, 2) to[out=90,in=-90] (0.67, 2.5);

    \node[font=\scriptsize, inner sep=0cm, anchor=center] (s1) at (0.335, 0.75) {$S_1$};
    \node[font=\scriptsize, inner sep=0cm,anchor=center] (s2) at (0.335, -0.25) {$S_2$};
    \node[font=\scriptsize, inner sep=0cm,anchor=center] (s3) at (1.67, 0.75) {$S_3$};
    \node[font=\scriptsize, inner sep=0cm,anchor=center] (s4) at (1.67, -0.25) {$S_4$};
    
\end{tikzpicture}

%% file: 1-Intro/Figs/fig_trafo_pctc.tex
\begin{tikzpicture}
    \node[rectangle, fill=lightgray, draw, minimum width = 2cm, minimum height= 1cm] (1) at (0,0) {$\mathcal{E}$};
    
    \draw[->] ([xshift=-0.67cm]1.north) -- ([xshift=-0.67cm, yshift=0.5cm]1.north) node[midway, inner sep=0cm, minimum size=0.1mm] (a1){};
    \draw[->] (1.north) -- ([yshift=0.5cm]1.north) node[midway, inner sep=0cm, minimum size=0.1mm] (a2){};
    \draw[->] ([yshift=-0.5cm]1.south) -- (1.south) node[midway, inner sep=0cm, minimum size=0.1mm] (a3){};
    \draw[->] ([xshift=-0.67cm, yshift=-0.5cm]1.south) -- ([xshift=-0.67cm]1.south) node[midway, inner sep=0cm, minimum size=0.1mm] (a4){};

    \draw[-] ([xshift=0.67cm, yshift=-0.5cm]1.south) -- ([xshift=0.67cm]1.south) node[midway, inner sep=0cm, minimum size=0.1mm] (s1){};
    \draw[-] ([xshift=0.67cm]1.north) -- ([xshift=0.67cm, yshift=0.5cm]1.north) node[midway, inner sep=0cm, minimum size=0.1mm] (s2){};

    \draw[-] ([xshift=1.34cm, yshift=-0.5cm]1.south) -- ([xshift=1.34cm, yshift=0.5cm]1.north);

    \draw ([xshift=0.67cm, yshift=0.5cm]1.north) arc (180:0:0.34);
    \draw ([xshift=0.67cm, yshift=-0.5cm]1.south) arc (-180:0:0.34);

    \node[font=\scriptsize, inner sep=0cm, anchor=center] (s1) at (0.05, 0.75) {$S_1$};
    \node[font=\scriptsize, inner sep=0pt, anchor = center] (s3) at (0.665,0.75) {$S_3$};
    \node[font=\scriptsize, inner sep=0pt, anchor = center] (s4) at (0.665, -0.75){$S_4$};
    \node[font=\scriptsize, inner sep=0pt, anchor = center] (s3) at (0.05, -0.75) {$S_2$};

    \draw[-, color=blue] (0.33,-1.1) -- (0.33, -1.5);
    \draw[-, color=blue] (1,-1.1) -- (1, -1.5);
    \draw[-, color=blue] (2.7,-1.5) -- (2.7, 4);
    \draw[-, color=blue] (-0.4,-1.5) -- (-0.4, 2);
    \draw[-, color=blue] (1,1.1) -- (1, 2);

    \draw[->, color= blue] (0.33, 1.5) -- (0.33,2);
    \draw[->, color= blue] (0.33, 2.5) -- (0.33,3);
    \draw[->, color= blue] (0.33, 3.5) -- (0.33,4);

    \draw[-, color=blue] (0.9,-1.5) -- (2.8, -1.5);
    \draw[-, color=blue] (-0.5,-1.5) -- (0.43, -1.5);
    \draw[-, color=blue] (-0.5,2) -- (0.43, 2);
    \draw[-, color=blue] (0.23,4) -- (2.8, 4);

    \draw[color=blue] (1, 2) to[out=90,in=-90] (0.33, 2.5);

    \draw[color=blue] (0.9, -1.5) arc (-180:0:0.95);
    \draw[color=blue] (-0.5, -1.5) arc (-180:0:0.465);
    \draw[color=blue] (-0.5, 2) arc (180:0:0.465);
    \draw[color=blue] (0.23, 4) arc (180:0:1.285);

    \node[color=blue, anchor = center] at (1.85, -1.975) {$\Phi^{+}$};
    \node[color=blue, anchor = center] at (1.515, 4.6425) {$\Phi^{+}$};
    \node[color=blue, anchor = center, font=\scriptsize] at (-0.035, -1.7325) {$\Phi^{+}$};
    \node[color=blue, anchor = center, font=\scriptsize] at (-0.035, 2.2325) {$\Phi^{+}$};
    
\end{tikzpicture}

%% file: 2-MTF/2-MTF.tex
\section{Multi-time States: review of mathematical and operational aspects}\label{sec:MTF}

In this section, we review the multi-time formalism. We begin with the motivation and explain the simplest representative case of a multi-time object, both in Sec. \ref{sec:mts:motivation}. From the representative case we generalise in two ways: to a general multi-time mathematical framework in Sec. \ref{sec:mts:mathframework}, and a general operational scenario of states, instruments and post-selection in Sec. \ref{sec:mts:operational}. We conclude by arguing that these are in one-to-one correspondence.

\subsection{Motivation and representative example}\label{sec:mts:motivation}
The MT formalism is, in essence, based on an extension of standard quantum theory. In standard QM, one has states/preparations and measurements/effects, and the temporal relationship is one-way: states evolve forward into effects, preparations into measurements, a relationship depicted in Fig. \ref{fig:StandardQM}. The flow of information is in one direction; at any given $t$ our description of a quantum system --- such as a density matrix --- is a description based on the past of the system.

There are two reasons why one needs to go beyond this time asymmetry. The first is a motivation in common with the field of quantum causality: investigating the fundamental nature of causality requires - at least on an abstract level - consistent descriptions of quantum scenarios that permit information flow or causal order outside what is conventionally expected; indeed MT states \cite{Aharonov1964, Aharonov1991, Silva2014}, process matrices \cite{Oreshkov2012, Chiribella2013} and CTC's \cite{Deutsch1991, Lloyd2011, Lloyd2011a} all arise from this perspective. 

The second motivation is operational. While the usual notion of quantum states and measurements can describe multiple systems in space but at a single time and involves time asymmetry, the MT formalism provides a notion of quantum states and measurements across multiple times while preserving time-symmetry by accounting for pre and post-selections on an equal footing. There is a class of paradoxes arising in operational quantum mechanical scenarios such as Hardy's paradox \cite{hardy1992paradox}, the pigeon-hole paradox \cite{aharonov2014_pigeonhole,aharonov2016pigeonhole}, etc. and the MT formalism provides one manner to recover the correct reasoning for the measurement results in such scenarios.

The most simple example devised by Aharanov et al. \cite{Aharonov1964} serves to demonstrate the above aspects and illustrate the core of the MT formalism. Consider that a quantum system is known to be in the pure state $\ket{\psi} \in \mathcal{H}_S$ at $t_1$, and that at time $t_2$ a projective measurement is performed, the result of which corresponds to the state $\ket{\phi}$. If a measurement described by the normalized set of Kraus operators $\{A_k\}_k$ is performed at a time $t$ between $t_1$ and $t_2$, then what is the probability of obtaining the outcome $k$?

Using the quantum mechanical Born rule together with the Bayes' rule for conditional probabilities given the post-selection, one gets the ABL expression \cite{Aharonov1964} for the probability $P(k| \ket{\psi}, \ket{\phi}, \{A_k\}_k)$ of obtaining the outcome $k$ given the pre-selection, post-selection and the measurement. When the items in the conditional are evident from context, we simply denote this as $P(k)$.
\begin{align}
\label{eq: ABL}
    P(k) &= \frac{ \left| \braket{\phi|A_k|\psi} \right|^2 }{\sum_l \left| \braket{\phi|A_l|\psi} \right|^2 }.
\end{align}

On an abstract level one can understand $\ket{\phi}$ to be a description of a \textit{backward-evolving} state and denoted by $\bra{\phi}$: information about the system at time $t_2$ flows backward in time to influence the measurement outcome at time $t$.

One can promote this understanding to a technical level by taking the states $\ket{\psi}$ and $\bra{\phi}$ as belonging to distinct Hilbert spaces labelled $\mathcal{H}^{t_1}$ and $\mathcal{H}_{t_2}$ respectively. The superscript/subscript differentiates between Hilbert spaces that are understood to evolve forward/backward in time, this is also reflected in the ket/bra form of the state. The object that encodes the information one has is the \textit{2-time state}:
\begin{align}
    \Psi &= \bra{\phi}_{t_2} \otimes \ket{\psi}^{t_1} \in \mathcal{H}_{t_2} \otimes \mathcal{H}^{t_1},
\end{align}
this is depicted in Fig. \ref{fig:PrePost_Simple}.

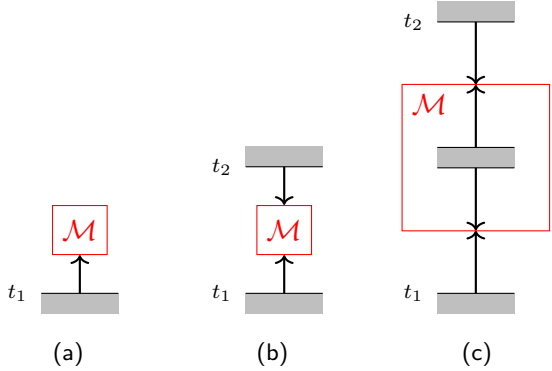
\begin{figure}[t]
    \centering
    \begin{subfigure}{.33\columnwidth}
        \centering
        \scalebox{1}{\input{2-MTF/Figs/StandardQM.tex}}
        \caption{}
        \label{fig:StandardQM}
    \end{subfigure}%
    \begin{subfigure}{.33\columnwidth}
        \centering
        \scalebox{1}{\input{2-MTF/Figs/PrePost_simple.tex}}
        \caption{}
        \label{fig:PrePost_Simple}
    \end{subfigure}%
    \begin{subfigure}{.33\columnwidth}
        \centering
        \scalebox{1}{\input{2-MTF/Figs/PrePost_full.tex}}
        \caption{}
        \label{fig:PrePost_full}
    \end{subfigure}
    \caption{(a) From standard quantum mechanics to the multi-time formalism. \emph{Single-time state}: At time $t_1$, a forward-evolving state, corresponding to a state in standard QM, is prepared. Later a measurement is performed on that state. (b) \emph{Two-time state}: As before, a state is prepared at time $t_1$, then a measurement is performed on that state. Finally, some quantum state is post-selected at time $t_2$, $t_2>t_1$. (c) \emph{Two-time operator}: The measurement itself can also be considered as a two-time object, which corresponds to a two-time operator (see also \cref{def: 2TO2TS} later).}
    \label{fig:CompStandardQM}
\end{figure}

The Kraus operator $A_k$ can also be understood in the same manner: given a basis in which to express it $A_k = \sum_{lm} a^k_{lm} \ket{l}\!\bra{m}$ we interpret $\bra{m}$ as the part that measures upon the forward evolving state $\ket{\psi}^{t_1} \in \mathcal{H}^{t_1}$, and $\ket{l}$ as the post-measurement state that $\bra{\phi}_{t_2} \in \mathcal{H}_{t_2}$ evolves backward into. This suggests that the Kraus operator can also be written as a 2-time object:
\begin{align}
    \textbf{A}_k &= \sum_{lm} a^k_{lm} \ket{l}^{t_2} \otimes \bra{m}_{t_1} \in \mathcal{H}^{t_2} \otimes \mathcal{H}_{t_1},\label{eq:2-timeKraus}
\end{align}
this is depicted in Fig. \ref{fig:PrePost_full}.

Within this example are the core ingredients of the general multi-time framework, namely:
\begin{enumerate}
    \item For each system and point of time at which one has information, there is a Hilbert space with a direction of evolution, chosen according to whether an object in this space meets another one at an earlier or a later time.
    \item Standard quantum operations are also understood as multi-time objects, with their measurement and post-measurement parts being in separate Hilbert spaces, rather than as linear operators upon the same Hilbert space.
    \item The probabilities of outcomes of measurements are conditional probabilities, normalised by the probability that some post-selection is successful.
\end{enumerate}

With these principles in mind, one can construct more complicated objects, notably: states that are entangled between the two times \cite{Aharonov1991, Aharonov2007}, mixtures of such pure states, coarse-grained measurements and POVMs \cite{Silva2014}, and generalise to more than just two times and multiple systems \cite{Aharonov1982, Aharonov1985, Vaidman1987_thesis,Popescu1991_thesis, Aharonov2009}. The entire result is the \textit{multi-time} (MT) formalism that we proceed to formally characterise as a mathematical framework.

\subsection{Mathematical characterisation of the multi-time formalism}\label{sec:mts:mathframework}

We begin by describing the mathematical formalism for pure multi-time objects, composition and probabilities therein before introducing the general mixed objects of the formalism. 

\paragraph{Hilbert spaces, pure states and fine-grained measurements}
\begin{definition}[MT Hilbert space]\label{def:MTspace}
A MT Hilbert space $\mathcal{H}$ is the tensor product of Hilbert spaces, where each space has three defining labels:
\begin{enumerate}
    \item which system it is associated with,
    \item the time at which the system interacts with another,
    \item the direction of evolution: forward if the system evolves forward into the interaction, backward for the opposite.
\end{enumerate}
For each unique trio of labels, there can only be a single Hilbert space present in the tensor product.
\end{definition}

As will be demonstrated in the results of this paper, the property of MT spaces that determines how they can be constructed is the relative time ordering between Hilbert spaces that are distinct from each other, either in dimension or in direction of evolution. As such it is useful to define an equivalence class of spaces as follows:
\begin{definition}[Time-order equivalence]
\label{def:timeorder_equiv}
Two MT Hilbert spaces are said to be time-order equivalent if they can be transformed into each other using any combination of two operations:
\begin{enumerate}
    \item shifting the time labels without changing their ordering,
    \item switching system labels between isomorphic Hilbert spaces evolving in the same temporal direction.
\end{enumerate}

\end{definition}

For the rest of this paper, any result pertaining to a MT Hilbert space can be treated as applicable to its equivalence class.

For notational simplicity, one can reduce the trio of labels into a single one w.l.o.g. by combining the system and time label into a single label, for example a system $S$ at time $t_1$ may be labeled as $S_1$. The direction of evolution is indicated by the placement of the label: if forward evolving then as a superscript $\mathcal{H}^S$, if backward then as a subscript $\mathcal{H}_S$. The general Hilbert space is thus of the form
\begin{align}\label{eq:MTHilbertspace}
    \mathcal{H} &= \bigotimes_i \mathcal{H}^{S_i} \bigotimes_j \mathcal{H}_{S_j},
\end{align}
where $S_i \neq S_j$ for any $i,j$.

Notice that, as each label contains all of the information pertaining to the system and time, the written order of the single Hilbert spaces is irrelevant. This is one example of a more general feature of the MT formalism that we will revisit: it has less hidden information due to having more explicitly labelled objects. The formalism can also be used without including the time information in the label but instead depicting this explicitly in the written order of the bras and kets in the expression for the MT object, and in its corresponding diagrammatic representation. Indeed, we will sometimes adopt this convention in later examples (see \cref{subsubsec:2TO-MTS}) as a simplification, while making it explicit that we are doing so. 

\begin{definition}[Reversed Hilbert space]
For any MT Hilbert space $\mathcal{H}$, the reversed Hilbert space $\bar{\mathcal{H}}$ is the one obtained by switching the direction of evolution of each Hilbert space within $\mathcal{H}$,
\begin{align}
    \text{if} \quad \mathcal{H} &= \bigotimes_i \mathcal{H}^{S_i} \bigotimes_j \mathcal{H}_{S_j}, \nonumber\\
    \text{then} \quad \bar{\mathcal{H}} &= \bigotimes_i \mathcal{H}_{S_i} \bigotimes_j \mathcal{H}^{S_j}.
\end{align}
\end{definition}

The reversed Hilbert space can be understood as being the one containing the `measurements' corresponding to the `states' in the original Hilbert space, for example the Hilbert spaces of the two-time state and measurement in Fig. \ref{fig:PrePost_full}. In general, there is no a priori reason to bias one or the other to be `states' rather than `measurements', only that they are concomitant in this manner. As in the case of states and measurement operators in standard QM here too vectors from these spaces form an inner product, as is defined shortly.

With the Hilbert spaces in place, one can now define the analogs of pure states and fine-grained quantum instruments \footnote{A fine-grained quantum instrument has a single Kraus operator corresponding to each classical outcome.},
\begin{definition}[Pure MT states and fine-grained instruments]
Given a multi-time Hilbert space $\mathcal{H}$, we define
\begin{enumerate}
    \item a pure MT state as any $\Psi \in \mathcal{H}$,
    \item a fine-grained MT instrument as an indexed set of vectors $\{A_k\}$, where $\{A_k\} \in \mathcal{H}$.
\end{enumerate}

\end{definition}

Thus both states and instruments correspond to vectors, see Table \ref{tab:SummaryMTNotation} for an example. Neither need to be normalized yet, we will discuss normalisation after the general MT state and instrument are defined.

\begin{table*}[ht!]
    \centering
    \begin{tabular}{p{0.25\linewidth}  p{0.3\linewidth} p{.3\linewidth}}
      \toprule
        Object & Hilbert space & Notation \\
      \midrule
      Prepared state & $\preH{S_1}$ & $\preS{\psi}{S_1}$ \\
      Post-selected state & $\postH{S_2}$ & $\postS{\phi}{S_2}$ \\
      2-time state state & $\prepostH{S_1}{S_2}:=\postH{S_2}\otimes\preH{S_1}$ & $\Psi=\postS{\phi}{S_2}\otimes\preS{\psi}{S_1}$ \\
      Kraus operator & $\prepostH{S_2}{S_1}:=\postH{S_1}\otimes\preH{S_2}$  & $E_a=\sum_{i,j} \beta_{a,i,j}\preS{i}{S_2}\otimes\postS{j}{S_1}$ \\
      \bottomrule
    \end{tabular}
    \caption{Basic objects in the MT notation. Pre- and post-selected states are denoted with superscripts and subscripts respectively. The Hilbert spaces of states and Kraus operators are dual to each other.}
    \label{tab:SummaryMTNotation}
\end{table*}

\paragraph{Composition and probabilities}

Finally, we define an operation that encompasses both the tensor product and the inner product, that we call composition and denote by the symbol $\bullet$. The technical definition may be found in Appendix \ref{app:MTcomposition}. It encodes the idea that if we combine two multi-time objects, each corresponding to a different part of a general quantum circuit, then the vectors from pairs of Hilbert spaces and their reverses should form an inner product, corresponding to the bringing together of a state and an operation upon that state. The rest of the Hilbert spaces correspond to disjoint operations, and should form the usual tensor product (Figure \ref{fig:venn}). Taken together this allows for the construction and manipulation of arbitrary circuit-like objects. 

\begin{figure}[ht!]
    \centering
    \scalebox{0.69}{\input{2-MTF/Figs/Venn}}
    \caption{The composed Hilbert space of two MT Hilbert spaces. The blue/red regions denote the forward/backward-evolving spaces. A forward and backward evolving space with the same label (one from $\mathcal{H}_1$ and the other from $\mathcal{H}_2$ corresponds to a point in the quantum circuit where a state from one space meets an effect from the other space. In the composition, the vectors from these overlapping spaces form an inner product in the composed MT object, with the corresponding Hilbert spaces removed from the composed MT space.}
    \label{fig:venn}
\end{figure}
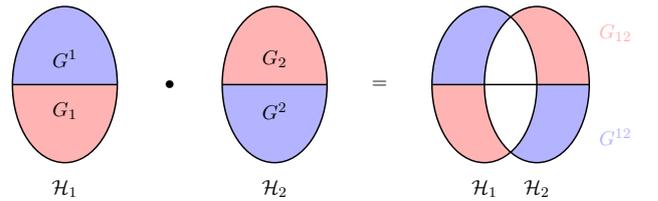

As discussed prior, the written order of Hilbert spaces is irrelevant in the MT formalism as the labels contain all pertinent information vis-\'a-vis temporal order. As a result the composition operation is also independent of the written order (thus automatically commutative for instance). In fact, as the composition operation includes the tensor product, one could in principle stop using the tensor product symbol $\otimes$ in its entirety as $\bullet$ would mean the same. It does however aid visual clarity to keep using $\otimes$ whenever the total Hilbert space does not contain pairs of Hilbert spaces and their reverse, so that it is clear that none of the spaces can contract.

Note that the composition of MT vectors reduces to a number in the case of two Hilbert spaces that are the reverse of one another, $\mathcal{H}_2 = \bar{\mathcal{H}_1}$; then the probabilities given by the ABL rule of \cref{eq: ABL} can be expressed as follows.
\begin{definition}[Probabilities of outcomes (pure states)]

Given a state $\Psi \in \mathcal{H}$ and a fine-grained instrument $\{A_k\}$ on the reversed space, $A_k \in \bar{\mathcal{H}}$, the probability of the outcome $k$ is defined as
\begin{align}\label{eq:2-timepureprob}
    P(k) &= \frac{ \abs{ \Psi \bullet A_k }^2 }{ \sum_l \abs{ \Psi \bullet A_l }^2 }.
\end{align}

\end{definition}

\paragraph{Mixed states and coarse-grained instruments}

To describe a mixed state in standard quantum theory, one promotes the pure state $\ket{\psi}_S \in \mathcal{H}_S$ to the density operator $\ket{\psi}_S\!\bra{\psi} \in \mathcal{L}(\mathcal{H}_S)$; the density operator of a mixed state then corresponds to a convex combination of the latter. This does not fit into the MT description for two reasons. Firstly, $\bra{\psi}_S \in \mathcal{H}_S$ already has a role --- as a vector in a different physical Hilbert space to that of $\ket{\psi}^S$.

We therefore choose a different --- albeit mathematically equivalent --- notation to denote mixed states in the MT formalism. For every MT Hilbert space $\mathcal{H}$ we construct a dual space that we call the \textit{daggered} space to $\mathcal{H}$, labelled $\mathcal{H}^\dagger$. For the forward-evolving space $\mathcal{H}^S$ the daggered space is denoted $\mathcal{H}_{S^\dagger}$, while for $\mathcal{H}_S$ it is denoted $\mathcal{H}^{S^\dagger}$\footnote{This is just a choice of convention to be consistent with existing works on the MT formalism, such as \cite{Silva2014,Silva2017}.}, with the states in these Hilbert spaces following the same convention for the label placements. Fig. \ref{fig:MTduality} depicts the relation between the four possible spaces with the same label.
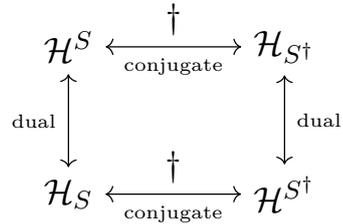
\begin{figure}[ht!]
    \centering
    \scalebox{1.2}{\input{2-MTF/Figs/MTduality}}
    \caption{Relations between the daggered and non-daggered spaces of bras and kets. The reversed space $\mathcal{H}_{S}$ is the true dual space to $\mathcal{H}^{S}$, while $\mathcal{H}_{S^\dagger}$ is only the complex conjugate space of $\mathcal{H}^{S}$.}
    \label{fig:MTduality}
\end{figure}

Here the pure multi-time state $\ket{\psi}^S$ is promoted to the \textit{pure density vector} $\ket{\psi}^S \otimes \bra{\psi}_{S^\dagger}$ rather than an operator, for a backward evolving state $\ket{\psi}_S$ it would be $\bra{\psi}_S \otimes \ket{\psi}^{S^\dagger}$. The only further clarification to add is for the composition rule --- this acts the same on daggered Hilbert spaces and vectors, but can only affect a reversed pair of Hilbert spaces or states that are both normal, or both daggered, never one of each.

The above definition of a density vector has an additional advantage. In standard quantum theory, both the density operator and Kraus operator belong to the space of linear operators $\mathcal{L}(\mathcal{H})$, even though they are conceptually distinct: one corresponds to the state as it is, the other to a transformation of the state. In the MT formalism the density vector of a forward-evolving state belongs to $\mathcal{H}_1 \otimes \mathcal{H}^{1\dagger}$, while a Kraus operator acting upon it belongs to $\mathcal{H}^1 \otimes \mathcal{H}_2$.

One may now define the analogs of mixed states and general quantum instruments for the MT formalism:

\begin{definition}[MT states and instruments]\label{def:MTobjects}
Given a multi-time Hilbert space $\mathcal{H}$, we define
\begin{enumerate}
    \item a mixed MT state from $\mathcal{H}$ as any \textit{positive} combination of pure density vectors from $\mathcal{H}$
    \begin{align}\label{eq:definitionMTstate}
        \eta &= \sum_r q_r \Psi_r \otimes \Psi_r^\dagger,
    \end{align}
    where $\Psi_r^\dagger \in \mathcal{H}^\dagger$ and $q_r \geq 0$,
    \item a MT instrument from $\mathcal{H}$ as an indexed set of MT Kraus density vectors $\{J_k\}$, each a sum of pure MT Kraus density vectors from $\mathcal{H}$
    \begin{align}
        J_k &= \sum_\chi A_{k,\chi} \otimes A_{k,\chi}^\dagger,
    \end{align}
    where $A_{k,\chi} \in \mathcal{H}$ and $A_{k,\chi}^\dagger \in \mathcal{H}^\dagger$. The MT channel corresponding to the MT instrument is thus the density vector which is the sum
      \begin{align}
      \label{eq: MTinstrument}
       \sum_k J_k &= \sum_{k,\chi} A_{k,\chi} \otimes A_{k,\chi}^\dagger.
    \end{align}
\end{enumerate}

\end{definition}
The above notion of MT Kraus density vectors $J_k$  generalises the definition of two-time Kraus density vectors introduced in \cite{Silva2014}.\footnote{Here, in line with terminology in \cite{Silva2014}: $A_{k,\chi} \in \mathcal{H}$ is a MT Kraus vector, $A_{k,\chi} \otimes A_{k,\chi}^\dagger$ is a pure MT Kraus density vector and $J_k$ is a generic MT Kraus density vector.} 

The above definition of an MT instrument is analogous to the definition of standard quantum instruments which are described by a collection of completely positive (CP) maps $\{\mathcal{E}_k\}_k$ where the action of the instrument is given by their sum $\sum_k \mathcal{E}_k$. By the Kraus representation theorem \cite{Kraus1983}, each CP map $\mathcal{E}_k$ can be represented in terms of a set of Kraus operators as $\mathcal{E}_k=\sum_\chi K_{k,\chi} (\rho) K_{k,\chi}^\dagger$.  Here the MT Kraus density vectors $J_k$ are analogous to the CP maps $\mathcal{E}_k$ and the MT Kraus vectors $A_{k,\chi}$ to the Kraus operators $K_{k,
\chi}$. However, while standard quantum instruments impose $\sum_k \mathcal{E}_k$ to be completely positive and trace preserving (CPTP) corresponding to a quantum channel, which would in turn imply a normalisation condition on the set $\{K_{k,\chi}\}$ of Kraus operators, MT instruments and Kraus (density) vectors do not have to be normalised. See \cref{sec:mts:operational} as well as Appendix D of \cite{Silva2014} for further details on the normalisation of these objects and links to standard operational quantum theory.

 In the case of having a state from $\mathcal{H}$ and an instrument from $\bar{\mathcal{H}}$, the probabilities of outcomes are defined analogously to that for pure states \eqref{eq:2-timepureprob}:
\begin{definition}[Probabilities of outcomes (mixed states)]

Given a mixed state $\eta$ from $\mathcal{H}$ and an instrument $\{J_k\}_k$ from $\bar{\mathcal{H}}$, the probability of the outcome $k$ is defined as
\begin{align}\label{eq:2-timemixedprob}
    P(k) &= \frac{ \eta \bullet J_k }{ \sum_l \eta \bullet J_l }.
\end{align}

\end{definition}

Note the consistency between probability expressions: the general formula \eqref{eq:2-timemixedprob} reduces to that for pure states and fine-grained measurements \eqref{eq:2-timepureprob} , i.e. in the case that $\eta = \Psi \otimes \Psi^\dagger$ and $J_k = A_k \otimes A_k^\dagger$.

Two multi-time objects $\Psi$ and $\Psi'$ differing by a constant $\Psi=k\Psi'$ would yield the same outcome probabilities for any instrument as the constant would appear in both the numerator and denominator and cancel out. This motivates the following definition. 
\begin{definition}
\label{def: op_equiv}
Two multi-time objects $\Psi$ and $\Psi'$ such that $\Psi=k\Psi'$ with $k\in\mathbb{C}$ are called operationally equivalent.
\end{definition}

A more detailed discussion on the normalisation of MT states may be found in Appendix \ref{app:MTnormalisation}

\paragraph{Partial and full trace}

To complete the MT formalism, we translate the operation of the trace from standard QM to MT notation. Given the standard partial trace is the following operation:
\begin{align}
    \Tr_A \left[ \rho_{AB} \right] &= \sum_i \bra{i}_A \rho_{AB} \ket{i}_A,
\end{align}
where $\{\ket{i}\}_A$ is an orthonormal basis for $\mathcal{H}_A$. The partial trace of a MT density vector $\eta$ is given by
\begin{align}
    \Tr^A \left[ \eta \right] &= \eta \bullet \left( \sum_i \bra{i}_A \otimes \ket{i}^{A^\dagger} \right) \\
    &= \eta \bullet \mathds{1}_A,
\end{align}
where we have defined the `identity operator' $\mathds{1}_A$ from $\mathcal{H}_A$.

\paragraph{Positivity}

An operator $A \in \mathcal{L}(\mathcal{H}_S)$ is said to be positive if $v^\dagger A v \geq 0$ for all $v \in \mathcal{H}_S$. For MT objects, this is mathematically equivalent to the following \cite{Silva2017}.
\begin{definition}[Positivity]\label{def:positivity}

An MT vector $J \in \mathcal{H} \otimes \mathcal{H}^\dagger$ is positive if and only if
\begin{align}
    J \bullet \left( v \otimes v^\dagger \right) \geq 0,
\end{align}
for all $v \in \bar{\mathcal{H}}$.

\end{definition}

We also have the spectral theorem for positive operators that follows from the usual spectral theorem, proven in Appendix \ref{app:compositionpositivity}.
\begin{lemma}

An MT vector $J \in \mathcal{H} \otimes \mathcal{H}^\dagger$ is positive if and only if it can be expressed as the positive sum of pure density vectors,
\begin{align}
    J &= \sum_r a_r u_r \otimes u_r^\dagger,
\end{align}
where $a_r > 0$ and $u_r \in \mathcal{H}$.
\label{lemma:spectral_thm}
\end{lemma}
This immediately implies the positivity of MT states and instruments. As one might expect, positivity is preserved by the operations of composition and trace.

\paragraph{Conclusion}

The crux of the mathematical MT formalism is in the following objects:
\begin{enumerate}
    \item the MT Hilbert space, a composition of forward and backward-evolving Hilbert spaces,
    \item MT states and instruments, which are composed of arbitrary positive density vectors on MT Hilbert spaces,
    \item the rule for probabilities.
\end{enumerate}

So far these are only mathematical objects. We proceed to argue that these are in one-to-one correspondence with standard quantum mechanics if post-selection is allowed.

\subsection{The operational MT scenario: quantum mechanics with post-selection}\label{sec:mts:operational}

In this section we characterize the MT formalism in an operational manner, and prove that this is equivalent to the mathematical formalism. The operational scenario is defined by the following three allowed processes:
\begin{enumerate}
    \item (preparation) the creation of any quantum state, pure or mixed,
    \item (instruments) an arbitrary quantum instrument, including measurements/channels,
    \item (post-selection) the discarding of results corresponding to some outcome/set of outcomes of any of the instruments.
\end{enumerate}

The first two of the above are part of standard QM, they include preparation, measurements and even the partial trace, which is one type of quantum instrument. The third is the general version of post-selection. More precisely, every quantum instrument is an indexed set of linear operators. At the end of the circuit, as long as the entire circuit is repeated so that data is collected upon the probability of various measurements, it is possible to choose a set of indexes from the various quantum instruments to discard.

Any scenario that can be created with the above three processes can be encoded into a MT object, for the following reasons. Firstly, states and instruments are clearly included in the MT formalism by construction, the only difference being the semantics and added explicit notation. Rather than density matrices, the states are density vectors, and rather than being a collection of Kraus operators, the instruments are collections of Kraus density vectors.

The third process is post-selection. A post-selection in the most general sense can be thought of as a generic quantum instrument, of whom a subset of the Kraus operators are ignored, corresponding to ignoring the statistics when the classical outcome corresponds to one of the ignored operators. More precisely, for every post-selection there exists a normalised set of Kraus operators $\{A_k\}$, where the index $k$ runs through some set $Q$, and the post-selection corresponds to a subset $Q^\prime \subset Q$.

This also falls under the MT formalism. A MT instrument need only be composed of Kraus density vectors, it does not need to be normalised. Thus by taking the set $\{A_k\}$ where $k \in Q^\prime$ and converting them into Kraus density vectors, we get a valid MT instrument. The only adjustment is statistical: all probabilities must be conditioned on the post-selection being successful, which is encoded into the definition of probabilities in the MT formalism, \eqref{eq:2-timemixedprob}.

Thus quantum mechanics with post-selection is automatically a subset of the MT formalism. The natural question is whether this correspondence is one-to-one, i.e. whether any MT state and instrument can be prepared by QM plus post-selection. The answer is yes. In \cite{Silva2014} it was argued that any MT state could be prepared using post-selection.

The main result of this paper gives an argument to show that any MT instrument can also be created in an operational manner, albeit using the language of P-CTC's. First notice that the action of an MT instrument is fully specified by a set of MT Kraus vectors $\{A_{k,\chi}\}_{k,\chi}$ as in \cref{eq: MTinstrument}, i.e., the object we wish to operationally realize is of the form $\sum_{k,\chi} A_{k,\chi} \otimes A_{k,\chi}^\dagger$ for $A_{k,\chi}$ in some MT Hilbert space $\mathcal{H}$ and $A_{k,\chi}^\dagger\in \mathcal{H}^\dagger$. It follows from the results of \cref{subsec:Mixed} that any MT object $\eta$ of the form $\eta=\sum_r p_r C_r \otimes C_r^\dagger$ can be realised through an operationally equivalent P-CTC assisted circuit, where $C_r$ is a pure MT object living in some MT Hilbert space $\mathcal{H}$ and where $p_r$ are probabilities, $\sum_r p_r =1$. In the case where all the probabilities $p_r$ are equal, $\eta$ is operationally equivalent to the MT object $\eta’=\sum_r C_r\otimes C_r^\dagger$, since overall constants do not affect the observable probabilities (\cref{def: op_equiv}). This is precisely of the form of an MT instrument of \cref{eq: MTinstrument} where the MT Kraus vectors $A_{k,\chi}$ take the role of $C_r$. The construction of such mixed MT objects is done by conditioning on an ancilla whose basis states correspond to $\{\ket{r}\}_r$ and implementing the corresponding  $C_r$ on the systems in $\mathcal{H}$ via a unitary P-CTC assisted circuit, and tracing out the ancilla. The result is understood in terms of a CPTP map assisted by P-CTCs (see \cref{fig:MixPCTC_2Figs} for an example).

Although P-CTCs provide an information-theoretic model for exotic causal structures associated to closed time-like curves, they also provide an operational procedure to simulate such exotic scenarios in a regular quantum experiment by means of experimental pre and post-selection on maximally entangled states. Therefore the above-mentioned construction gives an operational implementation of arbitrary MT instruments.

%% file: 2-MTF/Figs/StandardQM.tex
\begin{tikzpicture}[
    op/.style={shape= Op, minimum width = \WidthPrePost},
    pre/.style={shape = Pre, minimum width = \WidthPrePost},
    post/.style={shape = Post, minimum width = \WidthPrePost},
    Kop/.style={draw=red, minimum width=\WidthMeasurement, minimum height=\WidthMeasurement},
    t/.style={font=\scriptsize},
    ghost/.style={minimum width=\WidthMeasurement, minimum height=\WidthMeasurement},
    decoration={snake, segment length=4mm, amplitude=0.5mm}
]
\pgfsetmatrixrowsep{0.5cm}
\pgfsetmatrixcolumnsep{0.3cm}
\pgfmatrix{rectangle}{center}{mymatrix}
{\pgfusepath{}}{\pgfpointorigin}{\let\&=\pgfmatrixnextcell}
{
\node[Kop, text=red](11){$\mathcal{M}$}; \\
\node[pre](21){}; \\
}

\draw[->, thick] (21.north) -- (11.south) node [midway,inner sep=0 cm] (21to11) {};
\node[t,left=.5 mm of 21.north west](t1){$t_1$};

\end{tikzpicture}

%% file: 2-MTF/Figs/PrePost_simple.tex
\begin{tikzpicture}[
    op/.style={shape= Op, minimum width = \WidthPrePost},
    pre/.style={shape = Pre, minimum width = \WidthPrePost},
    post/.style={shape = Post, minimum width = \WidthPrePost},
    Kop/.style={draw=red, minimum width=\WidthMeasurement, minimum height=\WidthMeasurement},
    t/.style={font=\scriptsize},
    ghost/.style={minimum width=\WidthMeasurement, minimum height=\WidthMeasurement},
    decoration={snake, segment length=4mm, amplitude=0.5mm}
]
\pgfsetmatrixrowsep{0.5cm}
\pgfsetmatrixcolumnsep{0.3cm}
\pgfmatrix{rectangle}{center}{mymatrix}
{\pgfusepath{}}{\pgfpointorigin}{\let\&=\pgfmatrixnextcell}
{
\node[post](01){}; \\
\node[Kop, text=red](11){$\mathcal{M}$}; \\
\node[pre](21){}; \\
}

\draw[->, thick] (21.north) -- (11.south) node [midway,inner sep=0 cm] (21to11) {};
\draw[->, thick] (01.south) -- (11.north) node [midway,inner sep=0 cm] (01to11) {};

\node[t,left=.5 mm of 21.north west](t1){$t_1$};
\node[t,left=.5 mm of 01.south west](t2){$t_2$};

\end{tikzpicture}

%% file: 2-MTF/Figs/PrePost_full.tex
\begin{tikzpicture}[
    op/.style={shape= Op, minimum width = \WidthPrePost},
    pre/.style={shape = Pre, minimum width = \WidthPrePost},
    post/.style={shape = Post, minimum width = \WidthPrePost},
    Kop/.style={draw=red, minimum width=\WidthMeasurement, minimum height=\WidthMeasurement},
    t/.style={font=\scriptsize},
    ghost/.style={minimum width=\WidthMeasurement, minimum height=\WidthMeasurement},
    decoration={snake, segment length=4mm, amplitude=0.5mm}
]
\pgfsetmatrixrowsep{0.5cm}
\pgfsetmatrixcolumnsep{0.3cm}
\pgfmatrix{rectangle}{center}{mymatrix}
{\pgfusepath{}}{\pgfpointorigin}{\let\&=\pgfmatrixnextcell}
{
\node[post](11){}; \\
\node[ghost](21){}; \\
\node[op](31){}; \\
\node[ghost](41){}; \\
\node[pre](51){}; \\
}

\draw[->, thick] (11.south) -- (21.center);
\draw[->, thick] (31.north) -- (21.center) node [midway,inner sep=0 cm] (31to21) {};
\draw[->, thick] (31.south) -- (41.center) node [midway, inner sep=0 cm] (31to41){};
\draw[->, thick] (51.north) -- (41.center);

\node[circle, fit=(21.center) (41.center), inner sep=0cm](cir){};
\node[draw=red, fit=(cir), inner sep=0cm](fitSquare){};

\node[anchor=north west ,text=red] at (fitSquare.north west) (M) {$\mathcal{M}$};

\node[t,left=.5 mm of 51.north west](t1){$t_1$};
\node[t,left=.5 mm of 11.south west](t2){$t_2$};

\end{tikzpicture}

%% file: 2-MTF/Figs/Venn.tex
\begin{tikzpicture}
    \begin{scope}
        \clip (-1,0) -- (1,0) -- (1, 1.5) -- (-1, 1.5) -- cycle;
        \fill[blue!30] (0,0) ellipse (1 and 1.5);
    \end{scope}
    \begin{scope}
        \clip (-1,0) -- (1,0) -- (1, -1.5) -- (-1, -1.5) -- cycle;
        \fill[red!30] (0,0) ellipse (1 and 1.5);
    \end{scope}
    \draw[thick] (0,0) ellipse (1 and 1.5);
    \node at (0,-2) {$\mathcal{H}_1$};
    \draw[thick] (-1,0) -- (1,0);
    \node at (0,0.5) {$G^1$};
    \node at (0,-0.5) {$G_1$};

    \node at (2,0) {$\bullet$};
    
    \begin{scope}
        \clip (3,0) -- (5,0) -- (5, 1.5) -- (3, 1.5) -- cycle;
        \fill[red!30] (4,0) ellipse (1 and 1.5);
    \end{scope}
    \begin{scope}
        \clip (3,0) -- (5,0) -- (5, -1.5) -- (3, -1.5) -- cycle;
        \fill[blue!30] (4,0) ellipse (1 and 1.5);
    \end{scope}
    \draw[thick] (4,0) ellipse (1 and 1.5);
    \node at (4,-2) {$\mathcal{H}_2$};
    \draw[thick] (3,0) -- (5,0);
    \node at (4,0.5) {$G_2$};
    \node at (4,-0.5) {$G^2$};
    
    \node at (6,0) {$=$};
    
    \begin{scope}
        \clip (7,0) -- (9,0) -- (9, 1.5) -- (7, 1.5) -- cycle;
        \fill[blue!30] (8,0) ellipse (1 and 1.5);
        \fill[white] (9,0) ellipse (1 and 1.5);
    \end{scope}
    \begin{scope}
        \clip (9,0) ellipse (1 and 1.5);
        \clip (8,0) -- (10,0) -- (10, 1.5) -- (8, 1.5) -- cycle;
        \fill[red!30] (9,0) ellipse (1 and 1.5);
        \fill[white] (8,0) ellipse (1 and 1.5);
    \end{scope}
    \begin{scope}
        \clip (8,0) ellipse (1 and 1.5);
        \clip (7,0) -- (9,0) -- (9, -1.5) -- (7, -1.5) -- cycle;
        \fill[red!30] (8,0) ellipse (1 and 1.5);
        \fill[white] (9,0) ellipse (1 and 1.5);
    \end{scope}
    \begin{scope}
        \clip (9,0) ellipse (1 and 1.5);
        \clip (8,0) -- (10,0) -- (10, -1.5) -- (8, -1.5) -- cycle;
        \fill[blue!30] (9,0) ellipse (1 and 1.5);
        \fill[white] (8,0) ellipse (1 and 1.5);
    \end{scope}
    \draw[thick] (8,0) ellipse (1 and 1.5);
    \draw[thick] (9,0) ellipse (1 and 1.5);
    \draw[thick] (7,0) -- (10,0);
    \node at (8,-2) {$\mathcal{H}_1$};
    \node at (9,-2) {$\mathcal{H}_2$};
    \node[text=blue!30] at (10.5,-1) {$G^{12}$};
    \node[text=red!30] at (10.5,1) {$G_{12}$};
\end{tikzpicture}

%% file: 2-MTF/Figs/MTduality.tex
\begin{tikzpicture}
\pgfsetmatrixrowsep{1cm}
\pgfsetmatrixcolumnsep{1.5cm}
\pgfmatrix{rectangle}{center}{mymatrix}
{\pgfusepath{}}{\pgfpointorigin}{\let\&=\pgfmatrixnextcell}
{
\node[rectangle](1){$\mathcal{H}^{S}$}; \& \node[rectangle](2){$\mathcal{H}_{S^{\dagger}}$}; \\
\node[rectangle](3){$\mathcal{H}_{S}$}; \& \node[rectangle](4){$\mathcal{H}^{S^{\dagger}}$};\\
}

\draw[<->] (1.east) -- (2.west) node [midway,above] {$\dagger$} node [midway,below] {\tiny conjugate};
\draw[<->] (1.south) -- (3.north) node [midway,left] {\tiny dual};
\draw[<->] (3.east) -- (4.west) node [midway,above] {$\dagger$} node [midway,below] {\tiny conjugate};
\draw[<->] (2.south) -- (4.north) node [midway,right] {\tiny dual};

\end{tikzpicture}

%% file: 3-P-CTC/P-CTC_new.tex
\section{Post-selected Closed Timelike Curves: review and generalisation }\label{sec:PCTC}

Inspired by the possibility of spacetime geometries with closed timelike curves (CTCs) in Einstein's general relativity, Deutsch \cite{Deutsch1991} was among the first to consider CTCs in the context of quantum mechanics and quantum computation by considering a CTC-assisted unitary quantum interaction. The idea (see also \cref{fig:PCTC_map}) is to consider a chronology respecting system $S$ interacting with a chronology violating quantum system $A$. This interaction can be a unitary or a more general non-unitary quantum channel. After the interaction, the chronology-violating system is sent backwards in time through a CTC and the chronology-respecting system continues towards the future. The result is a map $\mathcal{C}_{CTC}$ acting on the chronology respecting system $S$ alone. In Deutsch's model, this map is defined through a consistency condition which preserves the quantum state of the system $A$ in the CTC. An alternative and distinct model of quantum CTCs, which will be the focus of this paper, is given by the formalism of postselected CTCs \cite{Lloyd2011,Lloyd2011a}. Here, the CTC-assisted quantum map is described by a post-selected teleportation protocol, a mechanism which has been found independently by different authors \cite{Bennett2005,Lloyd2011a}, also in the context of the Horowitz-Maldacena model for black-hole evaporation \cite{Horowitz2004}. 

In the following, we first review post-selected quantum teleportation and the resulting idea of P-CTC assisted maps modelling a one-time interaction \cite{Lloyd2011, Lloyd2011a}. We then extend this formalism to include more general P-CTC assisted objects which we call \emph{time-labelled P-CTC assisted combs} that can describe multi-time quantum protocols assisted by P-CTCs. We introduce these objects briefly and more intuitively here, while further technical details are deferred to \cref{appendix: PCTC_Combs}.

\subsection{Post-selected teleportation}

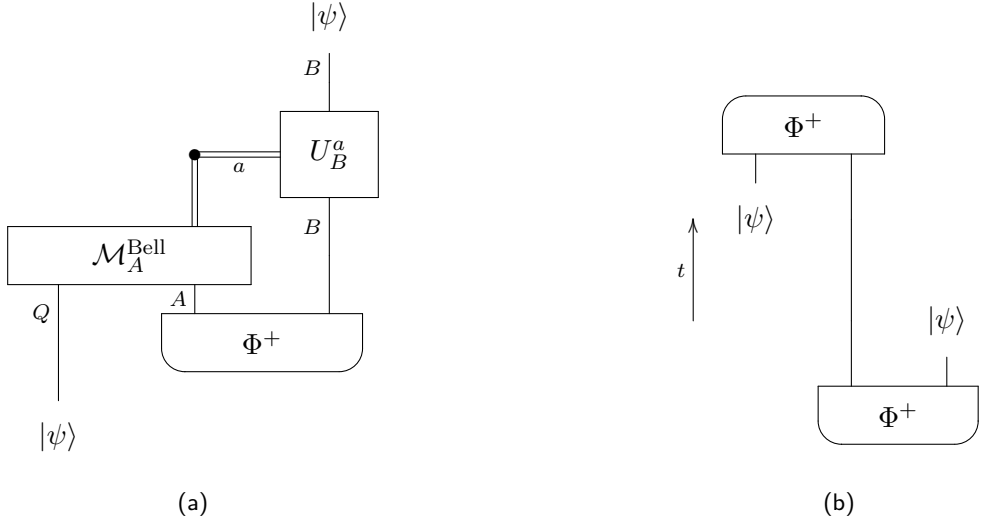
\begin{figure*}[ht!]
    \centering
    \begin{subfigure}{.5\textwidth}
        \centering
        \scalebox{1}{
            \input{3-P-CTC/Figs/StandardTeleportation.tex}
        }
        \caption{}
        \label{fig:StandardTeleportation}
    \end{subfigure}%
    \begin{subfigure}{.5\textwidth}
        \centering
        \scalebox{1}{
            \input{3-P-CTC/Figs/PostSelectedTeleportation.tex}
        }
        \caption{}
        \label{fig:PostSelectedTeleportation}
    \end{subfigure}
    \caption{(a) Standard quantum teleportation protocol, as described in the main text. The doubled lines denote classical communication from Alice to Bob.  
  (b)  Post-selected quantum teleportation. It corresponds to the standard protocol along with a post-selection on the outcome of Alice's Bell measurement being that of the maximally entangled state $\Phi^+$. It thus consists of pre- and post-selection on a maximally entangled state, and simulates a backward in time identity channel as described in the main text. This figure is inspired by \cite{Lloyd2011a}.}
    \label{fig:Teleportation}
\end{figure*}

Consider a teleportation protocol where Alice and Bob share a maximally entangled state of two qubits $\ket{\Phi^+}_{AB}$ while Alice has an additional qubit $Q$ in a state $\ket{\psi}_Q$ which she wishes to teleport to Bob. Alice jointly measures her qubits $A$ and $Q$ in the Bell basis and communicates the classical measurement outcome $a\in \{0,1,2,3\}$ to Bob. Bob can then recover the state $\ket{\psi}_B$ on his system $B$ by performing a unitary $U^a_B$ which depends on Alice's outcome $a$, on his half of the maximally entangled state. Notice that in a run of the experiment where Alice obtains the outcome corresponding to the Bell state $\ket{\Phi^+}$ (say $a=0$), Bob's operation is trivial $U_B^0=\id_B$. In other words, post-selecting on such a run of the experiment simulates a backwards in time identity channel whereby Bob already had the state $\ket{\psi}$ in his lab at a time earlier than when Alice attempts to teleport this state to him.\footnote{Thus experimental postselection on maximally entangled states allows to simulate closed timelike curves. In practice however, there is no causality violation, as the $a=0$ outcome only occurs probabilistically, and Bob cannot gain any information about the initial state on $Q$ until he receives physical classical communication from Alice about the outcome $a$. }
The protocol readily generalises to qudits ($d$-dimensional quantum systems for finite $d$), by taking $\ket{\Phi^+}_{AB}=\frac{1}{\sqrt{d}}\sum_{i=0}^{d-1}\ket{ii}_{AB}$ and the outcome $a=0$ to correspond to the same qudit maximally entangled state. The standard and post-selected teleportation protocols are  illustrated in \cref{fig:StandardTeleportation} and \cref{fig:PostSelectedTeleportation} respectively.

\subsection{P-CTC-assisted maps}
\label{ssec:ctcmap}

Consider a completely positive (CP) map $\mathcal{E}$ acting on a system $S$ (chronology-respecting system) and an ancilla system $A$ (chronology-violating system). This map can be assisted by a P-CTC which employs the post-selected teleportation protocol to ``teleport'' the chronology-violating output system backwards in time to the corresponding chronology violating input, as shown in figure~\ref{fig:PCTC_map}. In a hypothetical universe that allows such CTCs (e.g., by ensuring the post-selection fundamentally through appropriate boundary conditions), this would look like a backwards-in-time causal influence, if we regard the output as being later in time than the input of the interaction. On the other hand, an experimental simulation of the P-CTC would involve pre and post-selection on maximally entangled states, and discarding the experimental rounds where the post-selection does not succeed. At a mathematical level, we have the following definition.

\begin{figure}
    \centering
\begin{tikzpicture}[phiprep/.pic={\node at (1.3,-0.65) {\large{$\Phi^+$}};\draw (0,0) arc (180:360:1.3);
\draw (0,0)--(2.6,0); },phipost/.pic={\node at (1.3,0.65) {\large{$\Phi^+$}};\draw (0,0) arc (180:0:1.3);
\draw (0,0)--(2.6,0);},scale=0.5, transform shape]

\draw[->] (-2,-1)--node[anchor=east]{\large{$t$}}(-2,3); 
 
\draw (0,0) rectangle node[align=center]{\large{$\mathcal{C}_{CTC}$}} (2,2);

 \draw (1,-2)--node[anchor=east]{$S$}(1,0);  \draw (1,2)--node[anchor=east]{$S$}(1,4); 

   \node[align=center, black] at (3.5,1) {\Huge{$\mathbf{:=}$}};
   
   \draw (5,0) rectangle node[align=center]{\large{$\mathcal{E}$}} (8,2);

\draw (5.5,2)--node[anchor=east]{$S$}(5.5,4); \draw (7.5,2)--node[anchor=east]{$A$}(7.5,4);   \draw (7.5,4) arc (180:0:1);  \draw (7.5,-2) arc (180:360:1);

\draw (5.5,-2)--node[anchor=east]{$S$}(5.5,0);
\draw (7.5,-2)--node[anchor=east]{$A$}(7.5,0);
\draw (9.5,-2)--(9.5,4);

\begin{scope}[shift={(-2,-10)}]
    \draw (1,-2)--node[anchor=east]{$A$}(1,0);  \draw (1,2)--node[anchor=east]{$A$}(1,4); \draw (3,-2)--(3,4); \draw (1,4) arc (180:0:1);  \draw (1,-2) arc (180:360:1); \node[align=center, black] at (5.5,1) {\Huge{$\mathbf{:=}$}};

      \draw (8,-2)--node[anchor=east]{$A$}(8,0);  \draw (8,2)--node[anchor=east]{$A$}(8,4); \draw (10,-2)--node[anchor=east]{$B$}(10,4); \draw (7.7,-2) pic {phiprep}; \draw (7.7,4) pic {phipost};
      
\end{scope}

\end{tikzpicture}
    \caption{
    Top: Illustration of the idea of a CTC-assisted interaction $\mathcal{E}$ (e.g., a quantum channel or unitary). Bottom: In the P-CTC formalism, the CTC is obtained through pre and post-selection on a maximally entangled state which effectively teleports the system $A$ from future to past.}
    \label{fig:PCTC_map}
\end{figure}
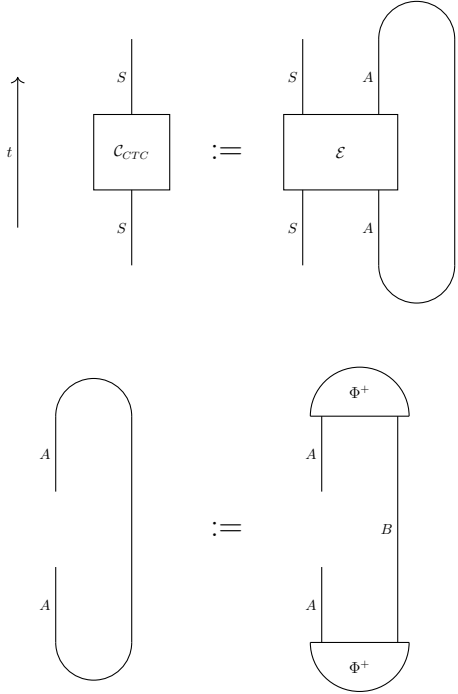

\begin{definition}[P-CTC-assisted map]
A PCTC-assisted map $\mathcal{C}_{CTC}$ is a linear CP map $\mathcal{C}_{CTC}: \mathcal{L}(\mathcal{H}_S)\mapsto \mathcal{L}(\mathcal{H}_S)$ which is obtained by assisting a linear CP map $\mathcal{E}: \mathcal{L}(\mathcal{H}_S)\otimes \mathcal{L}(\mathcal{H}_A)\mapsto \mathcal{L}(\mathcal{H}_S)\otimes \mathcal{L}(\mathcal{H}_A)$ through a P-CTC on the $A$ system. Explicitly, the action of this map is given as follows where we use the shorthand $\Phi^+_{AA'}:=\ket{\Phi^+}\bra{\Phi^+}_{AA'}$ and $A'$ is a system isomorphic to $A$, $\mathcal{H}_A\cong \mathcal{H}_{A'}$.
\begin{align}
    \begin{split}
          \label{eq: PCTC_map}
    &\mathcal{C}_{CTC} (\rho_S)\\=&\tr_{AA'}\Bigg[\Phi^+_{AA'}\Big(\mathcal{E} \otimes \id_{A'} (\rho_S \otimes \Phi^+_{AA'})\Big)\Bigg]\\
    =&\bra{\Phi^+}_{AA'} \mathcal{E} \otimes \id_{A'} (\rho_S \otimes \Phi^+_{AA'})\ket{\Phi^+}_{AA'}.
    \end{split}
\end{align}

\end{definition}

In the above, we have allowed $\mathcal{E}$ to be any CP map for generality, often this will take the form of a completely positive and trace preserving (CPTP) map which models physical quantum channels that includes unitaries. A convenient form of the action of a P-CTC assisted map is given by the following lemma, which generalises a result shown for the case of unitary interactions in \cite{Brun2012} to CP maps.

\begin{restatable}[Action of P-CTC-assisted maps]{lemma}{PCTCAction}
\label{lemma: PCTC_map_action}
Consider a P-CTC-assisted map $\mathcal{C}_{CTC}: \mathcal{L}(\mathcal{H}_S)\mapsto \mathcal{L}(\mathcal{H}_S)$ obtained through an interaction given by a CP map $\mathcal{E}: \mathcal{L}(\mathcal{H}_S)\otimes \mathcal{L}(\mathcal{H}_A)\mapsto \mathcal{L}(\mathcal{H}_S)\otimes \mathcal{L}(\mathcal{H}_A)$. If $\{K_j\}_j$ are a set of Kraus operators for $\mathcal{E}$ i.e., $\mathcal{E}(\sigma_{SA})=\sum_j K_j \sigma_{SA}K_j^{\dagger}$ for all $\sigma_{SA}\in \mathcal{L}(\mathcal{H}_S)\otimes \mathcal{L}(\mathcal{H}_A)$, then the action of $\mathcal{C}_{CTC}$ is given as follows.
\begin{align}
        \label{eq: PCTC_map_action}
    \mathcal{C}_{CTC}(\rho_S)=\frac{1}{d_A^2}\sum_j \tr_A [K_j] (\rho_S) \tr_A [K_j]^\dagger 
\end{align}

\end{restatable}

A proof of the above lemma can be found in appendix~\ref{appendix: PCTC_map_action}. In the particular case that $\mathcal{E}$ is a unitary $U$, then the corresponding P-CTC-assisted map is given as follows, as was originally shown in \cite{Brun2012}.

\begin{align}
            \label{eq: PCTC_unitary}
        \mathcal{C}_{CTC}(\rho_S)=\frac{1}{d_A^2}\tr_A [U] (\rho_S) \tr_A [U]^\dagger
\end{align}

\paragraph{Measurement probabilities}

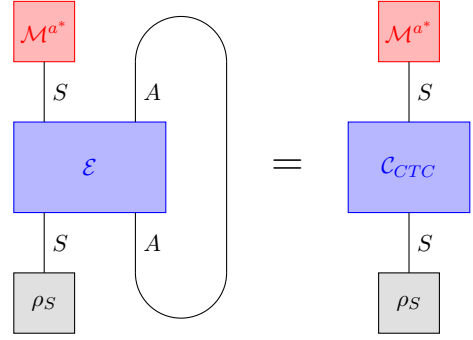
\begin{figure}
    \centering
\begin{tikzpicture}[scale=0.8, transform shape, trace/.pic={\draw [thick](-0.4,0)--(0.4,0);\draw [thick](-0.3,0.1)--(0.3,0.1);\draw [thick](-0.2,0.2)--(0.2,0.2);\draw [thick](-0.1,0.3)--(0.1,0.3);}]
    \draw[blue, fill=blue!70!white, fill opacity=0.4, text opacity=1] (0.5,0) rectangle node[align=center]{$\mathcal{E}$} (3,1.5); 
 \draw[red, fill=red!70!white, fill opacity=0.4, text opacity=1] (0.5,2.5) rectangle node[align=center]{$\mathcal{M}^{a^*}$} (1.5,3.5); 

  \draw[ fill=gray!60!white, fill opacity=0.4, text opacity=1] (0.5,-1) rectangle node[align=center]{$\rho_S$} (1.5,-2); 

  \draw (1,-1)--node[anchor=west]{$S$}(1,0);
    \draw (1,1.5)--node[anchor=west]{$S$}(1,2.5);

    \draw (2.5,-1)--node[anchor=west]{$A$}(2.5,0);
    \draw (2.5,1.5)--node[anchor=west]{$A$}(2.5,2.5);
    \draw (4,-1)--(4,2.5);

    \draw (4,-1) arc (360:180:0.75); \draw (4,2.5) arc (0:180:0.75);

    \node at (5,0.75) {\huge{$=$}};

    \begin{scope}[shift={(6,0)}]
          \draw[blue, fill=blue!70!white, fill opacity=0.4, text opacity=1] (0,0) rectangle node[align=center]{$\mathcal{C}_{CTC}$} (2,1.5); 
 \draw[red, fill=red!70!white, fill opacity=0.4, text opacity=1] (0.5,2.5) rectangle node[align=center]{$\mathcal{M}^{a^*}$} (1.5,3.5); 

  \draw[ fill=gray!60!white, fill opacity=0.4, text opacity=1] (0.5,-1) rectangle node[align=center]{$\rho_S$} (1.5,-2); 

  \draw (1,-1)--node[anchor=west]{$S$}(1,0);
    \draw (1,1.5)--node[anchor=west]{$S$}(1,2.5);

    \end{scope}
\end{tikzpicture}
    \caption{A map $\mathcal{E}$ assisted by a P-CTC on $A$ applied on an initial state $\rho_S$, and where we measure the output state on $S$ through a measurement yielding outcome $a=a^*$ (tracing out the post-measurement state). The probability of the outcome given the map, state and measurement is obtained through \cref{eq: PCTC_map_prob} where $\mathcal{C}_{CTC}$ is the map on $S$ obtained from assisting $\mathcal{E}$ with the P-CTC.}
    \label{fig:PCTC_map_meas}
\end{figure}

Consider an initial state $\rho_S$, acted upon by a P-CTC-assisted map $\mathcal{C}_{CTC}$, and subsequently measured through a measurement $\mathcal{M}$ given by a \emph{positive operator valued measure} (POVM) $\{\mathcal{M}^a\}_{a}$, where $\sum_a \mathcal{M}^a=\id_S$ (\cref{fig:PCTC_map_meas}). Then the probability of obtaining a specific measurement outcome $a^*$ is given as 
\begin{equation}
 \label{eq: PCTC_map_prob}   P(a^*|\mathcal{C}_{CTC},\rho_S,\{\mathcal{M}^a\}_{a})=\frac{\tr[\mathcal{C}_{CTC}(\rho_S)\mathcal{M}^{a^*}]}{\tr[\mathcal{C}_{CTC}(\rho_S)]}
\end{equation}

This probability expression can be readily derived using the Born rule and the standard rule for conditional probabilities, noting that $P(a^*|\mathcal{C}_{CTC},\rho_S, \{\mathcal{M}^a\}_{a})$ corresponds to the probability of obtaining the outcome $a$ given that the post-selection involved in creating the P-CTC-assisted map $\mathcal{C}_{CTC}$ has succeeded. Thus the numerator of the above expression corresponds to the probability that the measurement outcome $a^*$ is obtained in the final measurement and the outcome corresponding to the desired maximally entangled state is obtained in the post-selecting measurement involved in $\mathcal{C}_{CTC}$, while the denominator is simply the success probability of this post-selection. For a formal derivation of this rule, see appendix~\ref{appendix: PCTC_map_prob}.

Notice that any two P-CTC-assisted maps $\mathcal{C}_{CTC}$ and $\mathcal{D}_{CTC}$ which differ by some constant $k$ i.e., $\mathcal{C}_{CTC}=k\mathcal{D}_{CTC}$  will yield the same outcome probabilities for all measurements, as the factor $k$ cancels out between the numerator and denominator of \cref{eq: PCTC_map_prob}. Hence such maps $\mathcal{C}_{CTC}$ and $\mathcal{D}_{CTC}$ will be operationally equivalent (see \cref{def: equiv_PCTC} later).

\paragraph{P-CTC-assisted unitaries and pure representation} Often in this paper, it will be sufficient to work with the case where the initial channel $\mathcal{E}$ is a unitary, i.e., $\mathcal{E}(\cdot)=U_{SA}(\cdot)U_{SA}^\dagger$, where $U_{SA}:\mathcal{H}_S\otimes\mathcal{H}_A\mapsto \mathcal{H}_S\otimes\mathcal{H}_A$ is a unitary. In such cases, we will directly work with the operator $U_{SA}$ acting on the Hilbert space of pure states. We will distinguish evolutions on a Hilbert space $\mathcal{H}$ (such as $U_{SA}$) vs those on linear operators or density matrices on a Hilbert space $\mathcal{L}(\mathcal{H})$ (such as $\mathcal{E}$), by using mathcal font vs regular font for naming the latter and former respectively.

In the pure case, the unnormalised CTC-assisted map corresponding to having a P-CTC on the $A$ system in the unitary $U_{SA}$ is given as follows (as implied by the proof of \cref{lemma: PCTC_map_action} and \cite{Brun2012})
\begin{equation}
\label{eq: partialtrace_CTC}
    C_{CTC}=\frac{1}{d_A}\Tr_A[U_{SA}]: \mathcal{H}_S\mapsto \mathcal{H}_S
\end{equation}

Invoking the operational equivalence of P-CTC assisted objects differing by an overall constant (\cref{def: equiv_PCTC}), it would suffice to consider $\Tr_A[U_{SA}]$ as the relevant operator here.

\subsection{Time labelled P-CTC-assisted quantum combs}
\label{subsec:comb}
\begin{figure*}
\centering
\includegraphics[scale=1.0]{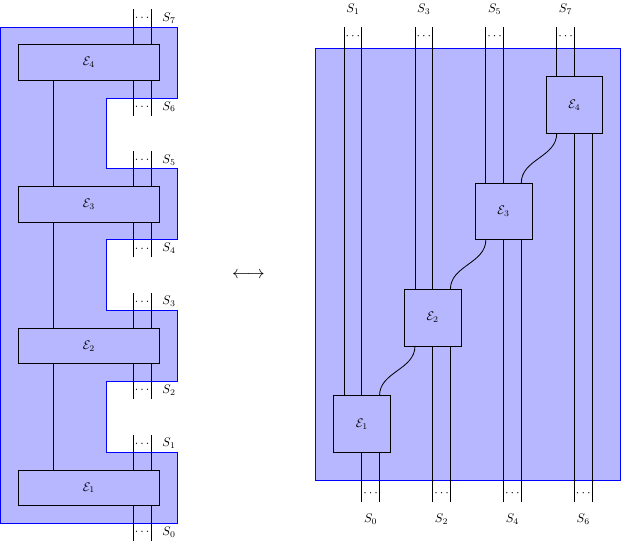}
\caption{The internal structure of a quantum comb \cite{Chiribella2008}  involves one CPTP map $\mathcal{E}_i$ for each tooth of the comb, where subsequent operations can be connected by an internal memory. The comb has a well-defined, acyclic causal order, with each $\mathcal{E}_i$ ordered later than $\mathcal{E}_{i-1}$ i.e., the outputs of $\mathcal{E}_{i-1}$ can (but does not necessarily) influence inputs of $\mathcal{E}_{i}$ but never the other way around. An $N$-slot comb can equivalently be represented as a single CPTP map (right) from inputs $\cup_{n=0}^N S_{2n}$ to outputs $\cup_{n=0}^N S_{2n+1}$. }
\label{fig:comb_sequence}
\end{figure*}

While the multi-time formalism as well as the P-CTC framework both involve post-selection, the type of objects considered in the two cases are a priory quite distinct. Multi-time states come with time labelled Hilbert-spaces which may be forward or backward evolving and include ``slots'' for plugging in different intermediate measurements while the P-CTC formalism considers a composition of fixed operations assisted by CTCs between Hilbert spaces that do not come with explicit time labels or allow ``slots'' for other operations. In this work, we define general P-CTC assisted objects that account for both these features: including explicit ``slots'' where other external operations can be plugged in and including time labels. For the former feature, we will associate time labels on the in/output systems in the P-CTC formalism. For the latter feature, we will use the concept of quantum combs introduced in \cite{Chiribella2008}, which can be thought of as (acyclic) quantum circuits with empty slots for plugging in channels, and extend these to also allow P-CTCs. The general objects we will thus obtain will be called time-labelled P-CTC assisted combs.

The full details of the definitions and properties of these general P-CTC assisted objects can be found in \cref{appendix: PCTC_Combs}. All of these details are not necessary for our main results but may be of independent interest in future work. 
For instance, \cref{appendix: PCTC_Combs} details how measurement probabilities for multi-time measurements are computed for arbitrary P-CTC assisted combs, along with how P-CTC-assisted combs act on external operations plugged into their slots. These will not be required in the main text, as we prove operational equivalence between objects in the two frameworks by constructing identical objects up to an overall constant, noting that in both the MTS and P-CTC frameworks, objects related in this way lead to the same measurement probabilities for all measurements (\cref{def: op_equiv} and \cref{def: equiv_PCTC}). Here, we provide an intuitive and illustrative overview of time-labelled P-CTC assisted combs, which will be relevant for the main results of the paper.

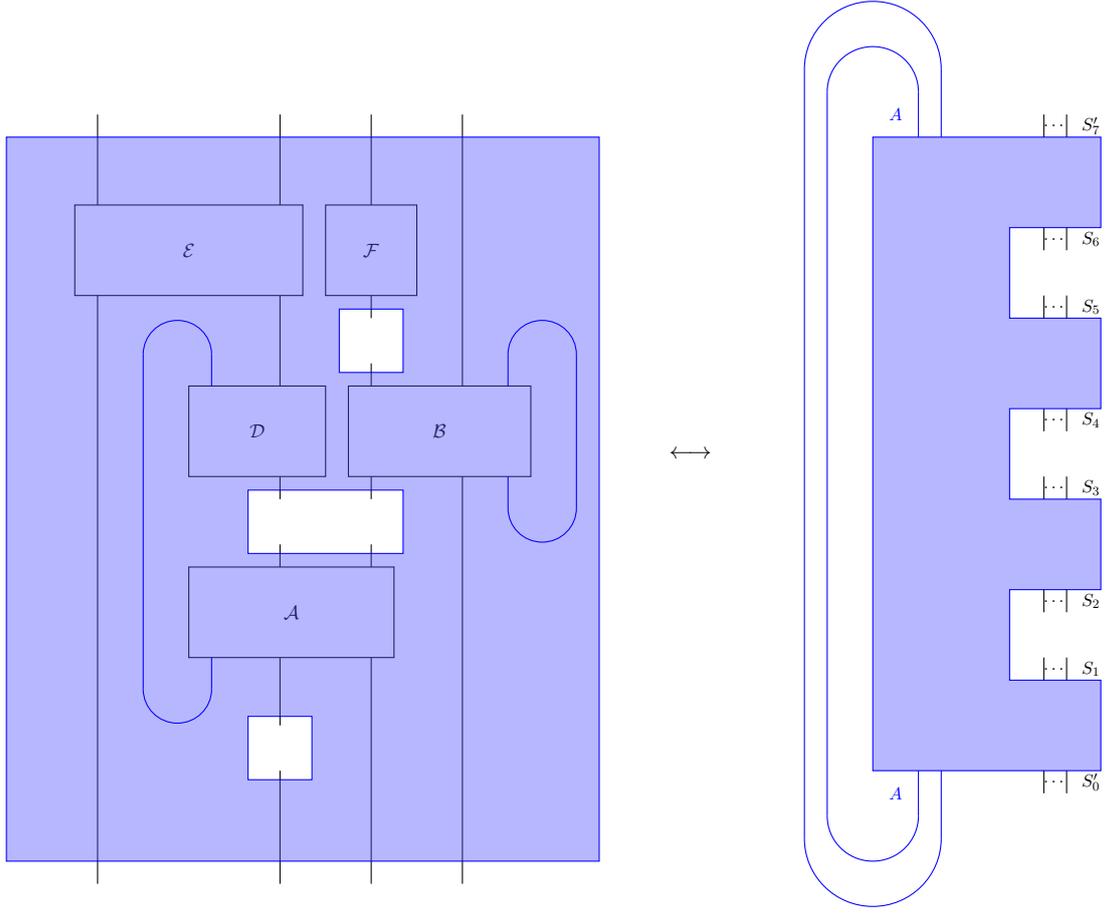
\begin{figure*}[ht]
 \centering
\begin{tikzpicture}[scale=0.6, transform shape]

\draw (-1.5,0) rectangle node[align=center]{\large{$\mathcal{A}$}} (3,2);

 \draw(0.5,-5)--(0.5,0);  \draw (2.5,-5)--(2.5,0); \draw (0.5,2)--(0.5,4); \draw (2.5,2)--(2.5,4); 
 
\draw (4.5,-5)--(4.5,4);\draw (2,4) rectangle node[align=center]{\large{$\mathcal{B}$}} (6,6);

\draw (-1.5,4) rectangle node[align=center]{\large{$\mathcal{D}$}} (1.5,6); \draw (2.5,6)--(2.5,8);\draw (4.5,6)--(4.5,10); \draw (0.5,6)--(0.5,8);

\draw (-4,8) rectangle node[align=center]{\large{$\mathcal{E}$}} (1,10);  \draw(-3.5,-5)--(-3.5,8);

\draw (-3.5,10)--(-3.5,12); \draw (0.5,10)--(0.5,12); \draw (2.5,10)--(2.5,12); \draw (4.5,10)--(4.5,12);

\draw (1.5,8) rectangle node[align=center]{\large{$\mathcal{F}$}} (3.5,10);
\draw[blue, fill=blue!70!white, fill opacity=0.4] (-5.5,-4.5) rectangle (7.5,11.5); 
\draw[blue, fill=white] (-0.2,-2.7) rectangle (1.2,-1.3); \draw (0.5,-2.7)--(0.5,-2.5); \draw (0.5,-1.5)--(0.5,-1.3);
\draw[blue, fill=white] (-0.2,2.3) rectangle (3.2,3.7);  \draw[blue, fill=white] (1.8,6.3) rectangle (3.2,7.7);

 \draw (0.5,2.3)--(0.5,2.5); \draw (0.5,3.5)--(0.5,3.7); \draw (2.5,2.3)--(2.5,2.5); \draw (2.5,3.5)--(2.5,3.7); \draw (2.5,6.3)--(2.5,6.5); \draw (2.5,7.5)--(2.5,7.7); 

\node at (9.5,4.5) {\Large{$\longleftrightarrow$}};

\draw[blue] (-1,6)--(-1,6.7); \draw[blue] (-1,0)--(-1,-0.7); \draw[blue] (-2.5,-0.7)--(-2.5,6.7); \draw[blue] (-1,6.7) arc (0:180:0.75); \draw[blue] (-1,-0.7) arc (0:-180:0.75);
\draw[blue] (5.5,6)--(5.5,6.7); \draw[blue] (5.5,4)--(5.5,3.3); \draw[blue] (7,3.3)--(7,6.7); \draw[blue] (7,6.7) arc (0:180:0.75);  \draw[blue] (7,3.3) arc (0:-180:0.75);

\begin{scope}[shift={(13.5,0)}]

 \draw[blue, fill=blue!70!white, fill opacity=0.4] (0,-2.5)--(0,11.5)--(5,11.5)--(5,9.5)--(3,9.5)--(3,7.5)--(5,7.5)--(5,5.5)--(3,5.5)--(3,3.5)--(5,3.5)--(5,1.5)--(3,1.5)--(3,-0.5)--(5,-0.5)--(5,-2.5)--cycle;

\draw(3.75,-0.5)--(3.75,0); \draw(4.25,-0.5)--(4.25,-0) node[midway, anchor=west, xshift=2mm] {$S_1$}; \node at (4,-0.25) {$\dots$}; \draw(3.75,1)--(3.75,1.5); \draw(4.25,1)--(4.25,1.5) node[midway, anchor=west, xshift=2mm] {$S_2$}; \node at (4,1.25) {$\dots$};
\begin{scope}[shift={(0,4)}]

\draw(3.75,-0.5)--(3.75,0); \draw(4.25,-0.5)--(4.25,-0) node[midway, anchor=west, xshift=2mm] {$S_3$}; \node at (4,-0.25) {$\dots$}; \draw(3.75,1)--(3.75,1.5); \draw(4.25,1)--(4.25,1.5) node[midway, anchor=west, xshift=2mm] {$S_4$}; \node at (4,1.25) {$\dots$};
\end{scope}
\begin{scope}[shift={(0,8)}]

\draw(3.75,-0.5)--(3.75,0); \draw(4.25,-0.5)--(4.25,-0) node[midway, anchor=west, xshift=2mm] {$S_5$}; \node at (4,-0.25) {$\dots$}; \draw(3.75,1)--(3.75,1.5); \draw(4.25,1)--(4.25,1.5) node[midway, anchor=west, xshift=2mm] {$S_6$}; \node at (4,1.25) {$\dots$};
\end{scope}

\begin{scope}[shift={(0,-4)}]
\draw(3.75,1)--(3.75,1.5); \draw(4.25,1)--(4.25,1.5) node[midway, anchor=west, xshift=2mm] {$S'_0$}; \node at (4,1.25) {$\dots$};
\end{scope}

\begin{scope}[shift={(0,12)}]
\draw(3.75,-0.5)--(3.75,0); \draw(4.25,-0.5)--(4.25,-0) node[midway, anchor=west, xshift=2mm] {$S'_7$}; \node at (4,-0.25) {$\dots$};
\end{scope}

\draw[blue] (1,11.5)--(1,12.5); \draw[blue] (1,-2.5)--(1,-3.5); \draw[blue] (-1,-3.5)--(-1,12.5);  \draw[blue] (-1,-3.5) arc (180:360:1); \draw[blue] (-1,12.5) arc (180:0:1);
\draw[blue] (1.5,11.5)--(1.5,13); \draw[blue] (1.5,-2.5)--(1.5,-4); \draw[blue] (-1.5,-4)--(-1.5,13); \draw[blue] (-1.5,-4) arc (180:360:1.5); \draw[blue] (-1.5,13) arc (180:0:1.5);
\node[blue] at (0.5,12) {$A$}; \node[blue] at (0.5,-3) {$A$};
\end{scope}

\end{tikzpicture}
    \caption{Any P-CTC-assisted quantum circuit with ``empty slots'' (left) can be represented in the form of a P-CTC assisted quantum comb (right), by stretching and rearranging internal wires}
    \label{fig: CTC_comb}
\end{figure*}

\paragraph{Quantum Combs} Quantum combs \cite{Chiribella2008} are useful for describing standard quantum circuits which include certain fixed operations along with ``slots'' for plugging in arbitrary operations. Any such ``quantum circuit with slots'' can be expressed in the general form of a quantum comb which is illustrated and explained in \cref{fig:comb_sequence}.

\paragraph{Including P-CTCs} In the same way that quantum circuits are built up through parallel and sequential composition of ordinary quantum channels, we can consider P-CTC-assisted circuits built up from parallel and sequential composition of P-CTC-assisted quantum channels. Then any such P-CTC-assisted circuit with ``slots'' for plugging in external operations corresponds to a P-CTC-assisted comb. This is illustrated in \cref{fig: CTC_comb}. Moreover, analogous to the general form of a regular quantum comb given in \cref{fig:comb_sequence}, we have the general form of a P-CTC-assisted comb illustrated in \cref{fig:CTCcomb_sequence}.
We will denote P-CTC-assisted combs as $\mathcal{C}_{CTC}$ and $C_{CTC}$ in the ``mixed'' (on density operators) and ``pure'' (on Hilbert spaces) representations respectively. Note that (P-CTC-assisted) quantum channels/maps are special cases of (P-CTC-assisted) quantum combs which justifies the identical notation for both. Wherever relevant throughout the paper, we will explicitly state whether we are considering a P-CTC-assisted map or a P-CTC-assisted comb.

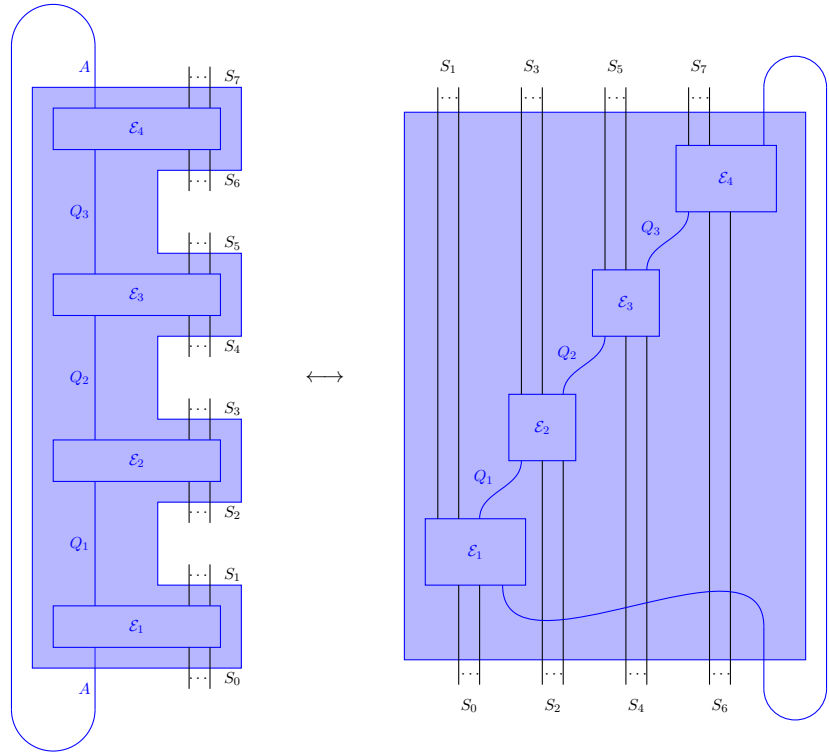
\begin{figure*}
    \centering
    \begin{tikzpicture}[scale=0.55, transform shape]

\draw[blue] (1.5,11)--(1.5,12.5); \draw[blue] (1.5,-2)--(1.5,-3.5); \draw[blue] (-0.5,-3.5)--(-0.5,12.5);  \draw[blue] (-0.5,-3.5) arc (180:360:1); \draw[blue] (-0.5,12.5) arc (180:0:1); \node[blue] at (1.25,12) {$A$}; \node[blue] at (1.25,-3) {$A$};

 \draw[blue, fill=blue!70!white, fill opacity=0.4] (0,-2.5)--(0,11.5)--(5,11.5)--(5,9.5)--(3,9.5)--(3,7.5)--(5,7.5)--(5,5.5)--(3,5.5)--(3,3.5)--(5,3.5)--(5,1.5)--(3,1.5)--(3,-0.5)--(5,-0.5)--(5,-2.5)--cycle;

\draw(3.75,-0.5)--(3.75,0); \draw(4.25,-0.5)--(4.25,-0) node[midway, anchor=west, xshift=2mm] {$S_1$}; \node at (4,-0.25) {$\dots$}; \draw(3.75,-1)--(3.75,-0.5); \draw(4.25,-1)--(4.25,-0.5);

\draw(3.75,1)--(3.75,1.5); \draw(4.25,1)--(4.25,1.5) node[midway, anchor=west, xshift=2mm] {$S_2$}; \node at (4,1.25) {$\dots$}; \draw(3.75,1.5)--(3.75,2);\draw(4.25,1.5)--(4.25,2);

\draw[blue] (0.5,-2) rectangle node[align=center]{$\mathcal{E}_1$} (4.5,-1); \draw[blue] (1.5,-1)--node[anchor=east]{$Q_1$}(1.5,2);
\begin{scope}[shift={(0,4)}]

\draw(3.75,-0.5)--(3.75,0); \draw(4.25,-0.5)--(4.25,-0) node[midway, anchor=west, xshift=2mm] {$S_3$}; \node at (4,-0.25) {$\dots$}; \draw(3.75,-1)--(3.75,-0.5); \draw(4.25,-1)--(4.25,-0.5);
\draw(3.75,1)--(3.75,1.5); \draw(4.25,1)--(4.25,1.5) node[midway, anchor=west, xshift=2mm] {$S_4$}; \node at (4,1.25) {$\dots$}; \draw(3.75,1.5)--(3.75,2);\draw(4.25,1.5)--(4.25,2);
\draw[blue] (0.5,-2) rectangle node[align=center]{$\mathcal{E}_2$} (4.5,-1); \draw[blue] (1.5,-1)--node[anchor=east]{$Q_2$}(1.5,2);
\end{scope}

\begin{scope}[shift={(0,8)}]
\draw(3.75,-0.5)--(3.75,0); \draw(4.25,-0.5)--(4.25,-0) node[midway, anchor=west, xshift=2mm] {$S_5$}; \node at (4,-0.25) {$\dots$}; \draw(3.75,-1)--(3.75,-0.5); \draw(4.25,-1)--(4.25,-0.5);
\draw(3.75,1)--(3.75,1.5); \draw(4.25,1)--(4.25,1.5) node[midway, anchor=west, xshift=2mm] {$S_6$}; \node at (4,1.25) {$\dots$}; \draw(3.75,1.5)--(3.75,2);\draw(4.25,1.5)--(4.25,2);
\draw[blue] (0.5,-2) rectangle node[align=center]{$\mathcal{E}_3$} (4.5,-1);\draw[blue] (1.5,-1)--node[anchor=east]{$Q_3$}(1.5,2);
\end{scope}

\begin{scope}[shift={(0,-4)}]
\draw(3.75,1)--(3.75,1.5); \draw(4.25,1)--(4.25,1.5) node[midway, anchor=west, xshift=2mm] {$S_0$}; \node at (4,1.25) {$\dots$}; \draw(3.75,1.5)--(3.75,2);\draw(4.25,1.5)--(4.25,2);
\end{scope}

\begin{scope}[shift={(0,12)}]
\draw(3.75,-0.5)--(3.75,0); \draw(4.25,-0.5)--(4.25,-0) node[midway, anchor=west, xshift=2mm] {$S_7$}; \node at (4,-0.25) {$\dots$}; \draw(3.75,-1)--(3.75,-0.5); \draw(4.25,-1)--(4.25,-0.5);
\draw[blue] (0.5,-2) rectangle node[align=center]{$\mathcal{E}_4$} (4.5,-1); 
\end{scope}

\node at (7,4.5) {\Large{$\longleftrightarrow$}};

\begin{scope}[shift={(11,-0.5)}, rotate=90, transform shape]
 \draw[blue, fill=blue!70!white, fill opacity=0.4] (-1.8,-7.5) rectangle (11.4,2.1);

\draw[blue] (0,-0.8) rectangle node[align=center]{\rotatebox{-90}{$\mathcal{E}_1$}} (1.6,1.6); \draw (1.6,1.3)--(12,1.3); \draw (1.6,0.8)--(12,0.8); \node[rotate=90] at (11.75,1.05) {$\dots$}; \node[rotate=-90] at (12.5,1.05) {$S_1$};
\draw[blue] (1.6,0.3) to[out=0,in=180] (3,-0.7); \node[blue, rotate=-90] at (2.6,0.2) {$Q_1$};

\draw (-2.4,0.3)--(0,0.3); \draw (-2.4,0.8)--(0,0.8); \node[rotate=90] at (-2.15,0.55) {$\dots$}; \node[rotate=-90] at (-2.9,0.55) {$S_0$};

\begin{scope}[shift={(3,-2)}, transform shape]
\draw[blue] (0,0) rectangle node[align=center]{\rotatebox{-90}{$\mathcal{E}_2$}} (1.6,1.6); \draw (1.6,1.3)--(9,1.3); \draw (1.6,0.8)--(9,0.8); \node[rotate=90] at (8.75,1.05) {$\dots$}; \node[rotate=-90] at (9.5,1.05) {$S_3$};
\draw[blue] (1.6,0.3) to[out=0,in=180] (3,-0.7); \node[blue, rotate=-90] at (2.6,0.2) {$Q_2$};
\draw (-5.4,0.3)--(0,0.3); \draw (-5.4,0.8)--(0,0.8); \node[rotate=90] at (-5.15,0.55) {$\dots$}; \node[rotate=-90] at (-5.9,0.55) {$S_2$};

\end{scope}

\begin{scope}[shift={(6,-4)}, transform shape]
\draw[blue] (0,0) rectangle node[align=center]{\rotatebox{-90}{$\mathcal{E}_3$}} (1.6,1.6); \draw (1.6,1.3)--(6,1.3); \draw (1.6,0.8)--(6,0.8); \node[rotate=90] at (5.75,1.05) {$\dots$}; \node[rotate=-90] at (6.5,1.05) {$S_5$};
\draw[blue] (1.6,0.3) to[out=0,in=180] (3,-0.7); \node[blue, rotate=-90] at (2.6,0.2) {$Q_3$};
\draw (-8.4,0.3)--(0,0.3); \draw (-8.4,0.8)--(0,0.8); \node[rotate=90] at (-8.15,0.55) {$\dots$}; \node[rotate=-90] at (-8.9,0.55) {$S_4$};
\end{scope}

\begin{scope}[shift={(9,-6)}, , transform shape]
\draw[blue] (0,-0.8) rectangle node[align=center]{\rotatebox{-90}{$\mathcal{E}_4$}} (1.6,1.6); \draw (1.6,1.3)--(3,1.3); \draw (1.6,0.8)--(3,0.8); \node[rotate=90] at (2.75,1.05) {$\dots$}; \node[rotate=-90] at (3.5,1.05) {$S_7$};
\draw (-11.4,0.3)--(0,0.3); \draw (-11.4,0.8)--(0,0.8); \node[rotate=90] at (-11.15,0.55) {$\dots$}; \node[rotate=-90] at (-11.9,0.55) {$S_6$};

\draw[blue] (1.6,-0.5) --(3,-0.5); \draw[blue] (-11.5,-0.5) --(-10,-0.5); \draw[blue] (-10,-0.5) to[out=0,in=180] (-9,5.75); \draw[blue] (-11.5,-2)--(3,-2);  \draw[blue] (3,-0.5) arc (90:-90:0.75); \draw[blue] (-11.5,-0.5) arc (90:270:0.75); 
\end{scope}

\end{scope}

\end{tikzpicture}
    \caption{Analogous to the case of regular quantum combs, illustrated in \cref{fig:comb_sequence}, any $N$-slot P-CTC assisted comb (left) can be equivalently represented as a single P-CTC assisted map (right). }
    \label{fig:CTCcomb_sequence}
\end{figure*}

\paragraph{Action of P-CTC-assisted combs on external operations}

The action of a P-CTC-assisted comb $\mathcal{C}_{CTC}$ on external operations $\{\mathcal{M}_i\}_i$ plugged in within each of its slots can be equivalently written in terms of a sequential composition of quantum channels associated with the comb and external operations followed by P-CTCs, as illustrated in \cref{fig:pctc_comb_meas}. In particular, if the operations $\{\mathcal{M}_i\}_i$ are measurements, their outcome probabilities can be computed through this procedure as detailed in \cref{appendix: PCTC_Combs}. In other words, the measurement probabilities in a P-CTC-assisted comb are equivalently computed in terms of composition of P-CTC-assisted maps. 

As we noted before in the case of P-CTC-assisted maps in \cref{ssec:ctcmap}, two such objects that differ only by a constant factor lead to the same probabilities for all measurements. This motivates the following definition, where P-CTC assisted object refers to maps and combs alike.

\begin{definition}[Operational equivalence for P-CTC assisted objects]
\label{def: equiv_PCTC}
Two P-CTC assisted objects, $\mathcal{C}_{CTC}$ and $\mathcal{D}_{CTC}$ are said to be \emph{operationally equivalent} whenever there exists  $k\in \mathbb{R}$ such that $\mathcal{C}_{CTC}=k\mathcal{D}_{CTC}$.
    
\end{definition}

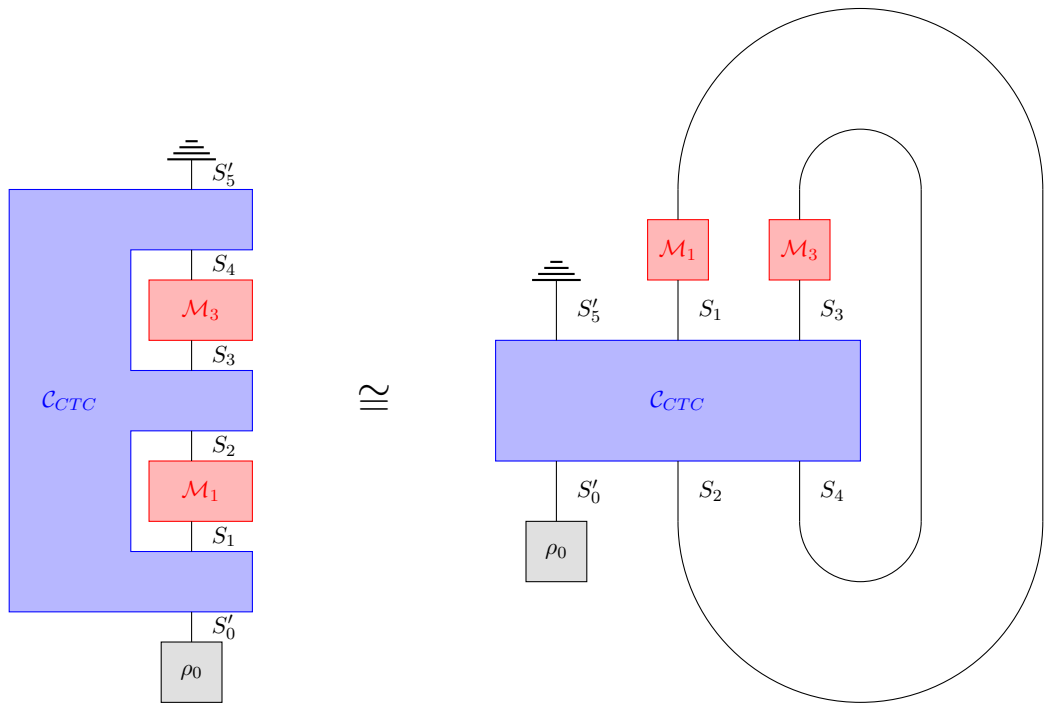
\begin{figure*}
    \centering
\begin{tikzpicture}[scale=0.8, transform shape, trace/.pic={\draw [thick](-0.4,0)--(0.4,0);\draw [thick](-0.3,0.1)--(0.3,0.1);\draw [thick](-0.2,0.2)--(0.2,0.2);\draw [thick](-0.1,0.3)--(0.1,0.3);}]

 \draw[blue, fill=blue!70!white, fill opacity=0.4] (0,0)--(0,1)--(-2,1)--(-2,3)--(0,3)--(0,4)--(-2,4)--(-2,6)--(0,6)--(0,7)--(-4,7)--(-4,0)--cycle; \node[blue] at (-3,3.5) {$\mathcal{C}_{CTC}$};
 \draw[red, fill=red!70!white, fill opacity=0.4, text opacity=1] (0,1.5) rectangle node[align=center]{$\mathcal{M}_1$} (-1.7,2.5);
 
 \draw[red, fill=red!70!white, fill opacity=0.4, text opacity=1] (-1.7,4.5) rectangle node[align=center]{$\mathcal{M}_3$} (0,5.5);

\draw(-1,-0.5)--(-1,0)  node[midway, anchor=west, xshift=2mm] {$S'_0$}; \draw(-1,1)--(-1,1.5)  node[midway, anchor=west, xshift=2mm] {$S_1$}; \draw(-1,2.5)--(-1,3)  node[midway, anchor=west, xshift=2mm] {$S_2$}; \draw(-1,4)--(-1,4.5)  node[midway, anchor=west, xshift=2mm] {$S_3$}; \draw(-1,5.5)--(-1,6)  node[midway, anchor=west, xshift=2mm] {$S_4$}; \draw(-1,7)--(-1,7.5)  node[midway, anchor=west, xshift=2mm] {$S'_5$}; 

\draw[fill=gray!60!white, fill opacity=0.4, text opacity=1] (-1.5,-1.5) rectangle node[align=center]{$\rho_0$} (-0.5,-0.5);

\draw (-1,7.5) pic {trace};

\node at (2,3.5) {\huge{$\cong$}};

\begin{scope}[shift={(5,0)}]
    \draw[blue, fill=blue!70!white, fill opacity=0.4, text opacity=1] (-1,2.5) rectangle node[align=center]{$\mathcal{C}_{CTC}$} (5,4.5); 

    \draw[red, fill=red!70!white, fill opacity=0.4, text opacity=1] (1.5,5.5) rectangle node[align=center]{$\mathcal{M}_1$} (2.5,6.5); 
    \draw[red, fill=red!70!white, fill opacity=0.4, text opacity=1] (3.5,5.5) rectangle node[align=center]{$\mathcal{M}_3$} (4.5,6.5); 

    \draw[fill=gray!60!white, fill opacity=0.4, text opacity=1] (-0.5,0.5) rectangle node[align=center]{$\rho_0$} (0.5,1.5);
   \draw(0,1.5)--(0,2.5)  node[midway, anchor=west, xshift=2mm] {$S'_0$}; 
      \draw(2,1.5)--(2,2.5)  node[midway, anchor=west, xshift=2mm] {$S_2$}; 
    \draw(4,1.5)--(4,2.5)  node[midway, anchor=west, xshift=2mm] {$S_4$};  

    \draw(0,4.5)--(0,5.5)  node[midway, anchor=west, xshift=2mm] {$S'_5$}; \draw (0,5.5) pic {trace};
    \draw(2,4.5)--(2,5.5)  node[midway, anchor=west, xshift=2mm] {$S_1$};
\draw(4,4.5)--(4,5.5)  node[midway, anchor=west, xshift=2mm] {$S_3$};
\draw(2,6.5)--(2,7); \draw(4,6.5)--(4,7);
\draw(6,1.5)--(6,7); \draw(8,1.5)--(8,7);
\draw (6,1.5) arc (360:180:1); \draw (6,7) arc (0:180:1);
\draw (8,1.5) arc (360:180:3); \draw (8,7) arc (0:180:3);
      
\end{scope}

\end{tikzpicture}
    \caption{The action of a P-CTC assisted comb on maps plugged into its slots can be equivalently viewed in terms of a composition of linear CP maps through additional P-CTCs (one for each slot) as shown on the right. The loops correspond to pre and post selection on appropriate maximally entangled states as shown in figure~\ref{fig:PCTC_map}.}
    \label{fig:pctc_comb_meas}
\end{figure*}

In \cref{fig:pctc_comb_meas}, we considered  a separate operation in each slot of a P-CTC-assisted comb, i.e., the overall external operation being plugged contains a tensor factor corresponding to each slot of the comb, and there are no external memory systems connecting the operations plugged into different slots. 

More generally, one may consider plugging in another P-CTC-assisted comb that fits into all the slots of the original P-CTC-assisted comb, which would allow to model the analogue of ``multi-time measurements''. However, as shown in \cref{fig:CTC_comb_nonprod} and explained in \cref{appendix: PCTC_Combs}, the action of any P-CTC-assisted comb $\mathcal{C}_{CTC}$ on another P-CTC-assisted comb $\mathcal{D}_{CTC}$ (where the ``teeth'' of one comb fit into the slots of the other) can equivalently be described by the action of another P-CTC-assisted comb $\mathcal{C}'_{CTC}$ on a memoriless or product quantum comb  $\mathcal{D}_{\otimes}$ (which is not assisted by any P-CTCs) i.e., the combination of the two P-CTC-assisted combs $\mathcal{C}_{CTC}$ and $\mathcal{D}_{CTC}$ are operationally equivalent to the combination of the P-CTC-assisted comb $\mathcal{C}'_{CTC}$ and a regular memoriless comb $\mathcal{D}_{\otimes}$. 

\paragraph{Time labelling} To bring our discussion of combs closer to the multi-time formalism, we will add time labels to the in/output systems. Specifically, for regular and P-CTC-assisted combs when expressed in the general form as in the left hand side of \cref{fig:comb_sequence} and \cref{fig:CTCcomb_sequence} we will consider time as flowing from bottom to the top of the page such that the systems $S_0$, $S_1$, $S_2$,... are assigned the ordered sequence of times $t_0$, $t_1$, $t_2$,.. with $t_0<t_1<t_2<...$. We will represent time-labelled combs in the representation of the left hand side of \cref{fig:comb_sequence} and \cref{fig:CTCcomb_sequence}, to make explicit this time ordering. Moreover, the order $t_0<t_1<t_2<...$ is the relevant aspect, not the actual labels.

%% file: 3-P-CTC/Figs/StandardTeleportation.tex
\mbox{
    \Qcircuit @C = 1em @R = 1em{
    & & \tick{\ket{\psi}} \\
    & & \qwl{^-{B}}\\
    & \control & \gate{U^{a}_{B}}{}{2} \cwl{^-{a}}\\
    \multigaten{1}{\mathcal{M}_{A}^{\text{Bell}}}{2} & \multinghost{\mathcal{M}_{\text{Alice}}^{\text{Bell}}}{2}\cwx & \qwl{^-{B}}\\
    \qwl{^-{Q}} & \multiDmeasure{1}{\Phi^+}{2}{^-{A}} & \multighost{\Phi^+}{2}{} \\
    \tick{\ket{\psi}}\qwx& &
    }
}

%% file: 3-P-CTC/Figs/PostSelectedTeleportation.tex
\mbox{
    \Qcircuit@C = 1em @R = 1em{
    & \multimeasureD{1}{\Phi^+}{2} & \multinghost{\Phi^+}{2} & \\
    & \tick{\ket{\psi}}\qwx & \qwx & & \\
    \ar[u]^-{t}& & \qwx & \tick{\ket{\psi}} \\
    & & \multiDmeasure{1}{\Phi^+}{2}{} & \multighost{\phi^+}{2}{} \\
    & & & 
    }
}

%% file: 4-Connection/4-Connection.tex
\section{Mapping P-CTC assisted combs to MTS}
\label{sec:PCTC-MTF}

It is useful to begin by formally distinguishing between the two types of multi-time states (MTS) which involve two relevant times: a two-time operator (2TO) vs two-time state (2TS), which were introduced more informally in \cref{fig:CompStandardQM}.

\begin{definition}[2TO and 2TS]
\label{def: 2TO2TS}
    We call a MTS $M$ a 2TO if all backward evolving systems in $M$ occur earlier in time than all of its forward evolving systems. 
       We call an MTS $M$ a 2TS if all backward evolving systems in $M$ occur later in time than all of its forward evolving systems.  
    
\end{definition}

Consider a P-CTC as depicted in \cref{fig:PCTC_map}, which connects an output system (say $A_2$) of some quantum map/channel $\mathcal{E}$ to an isomorphic input system (say $A_1$). We distinguish the in and output ancillary systems here (both labelled $A$ in \cref{fig:PCTC_map}) for clarity in the following equations. The P-CTC (bottom of the figure) involves maximally entangled states between $A_2B$ (post-selection) and $A_1B$ (pre-selection) which can be expressed as $\postS{\Phi^+}{A_2 B} \otimes \preS{\Phi^+}{A_1 B}$, where $B$ is an isomorphic ancillary system that is contracted between the pre- and post-selections. Performing this contraction while disregarding the overall normalization factor (which preserves operational equivalence, \cref{def: op_equiv}), yields a two-time state (2TS) between $A_1$ (evolving forward at time $t_1$) and $A_2$ (evolving backward at time $t_2 > t_1$) where $d$ is the dimension of the system $A_1$ and the isomprphic system $A_2$:
\begin{equation}
\label{eq: ent2TS}
    \Psi = \sum_{i=0}^{d-1} \postS{i}{A_2} \otimes \preS{i}{A_1}
\end{equation}
This is a 2TS since the pre-selection occurs before the post-selection in this scenario. 
Indeed Aharonov et. al. \cite{Aharonov2009} have already noted that a maximally entangled 2TS (\cref{eq: ent2TS}) describes a closed timelike curve from the later in time system $A_2$ to the earlier in time system $A_1$. Thus, the following proposition follows immediately from what is known in the literature.

\begin{proposition}
    A maximally entangled 2TS, represented as  $\Psi$ of \cref{eq: ent2TS} in the pure MTS formalism with $A_2$ later in time than $A_1$ (and as $\eta = \Psi \otimes \Psi^\dagger$ in the general, mixed MTS formalism) is operationally equivalent to a P-CTC linking $A_2$ to $A_1$.
\end{proposition}

More generally, this immediately yields the next proposition, because by definition, P-CTC assisted combs are formed by composing quantum state preparations, and quantum channels with P-CTCs, each of which has an operationally equivalent counterpart in the MTS formalism as discussed in \cref{sec:MTF} and the above. 
\begin{proposition}
\label{prop: CTC_to_MTS}
Every time-labelled P-CTC assisted comb $\mathcal{C}^{\mathbf{t}}_{\text{CTC}}$ (\cref{def: timelab_PCTC_comb}) is operationally equivalent to an MTS defined on the same systems and associated with the same time labels, with each input/output system of $\mathcal{C}^{\mathbf{t}}_{\text{CTC}}$ corresponding to a backward/forward evolving system in the MTS description. 

\end{proposition}

\section{Mapping MTS to P-CTC assisted combs}\label{sec:Connection}
This section describes the mapping from MTS to operationally equivalent P-CTC assisted objects in three main steps. This is first done for pure multi-time objects in two steps: mapping 2TOs to P-CTC-assisted maps, and mapping MTSs to 2TOs. The conclusions are generalised to mixed MT objects in the last step.
\subsection{Connecting pure 2TOs and P-CTC-assisted maps}\label{subsubsec:2TO-PCTC}

We first consider any operator $C\in \mathcal{H}_{S_1}\otimes \mathcal{H}^{S_2}$, this corresponds to a two time operator (2TO) which is a special MTS having a backward evolving space $\mathcal{H}_{S_1}$ at an earlier time and forward evolving space $\mathcal{H}^{S_2}$ at a later time. Then the question considered in this section is: given any operator $C\in \mathcal{H}_{S_1}\otimes \mathcal{H}^{S_2}$, does there exist a P-CTC-assisted map that implements an operator $C_{CTC}$ operationally equivalent to $C$?

For simplicity, we first consider $S_1$ and $S_2$ to be isomorphic to a $d$-dimensional quantum system, and denote both by $S$. We begin by observing that an arbitrary operator $C$ on a d-dimensional system $S$, i.e. on qudits, can be written as
\begin{equation}
    C = r \sum_{i=0}^{d-1} a_i\proj{\psi_i}{i}_{S}
    \label{eq:C_arbitrary}
\end{equation}
where $r\in\mathbb{C}$, and for all $i$, $a_i\in \mathbb{R}$ and $\left|a_i\right|\in[0,1]$, each $\ket{\psi_i}\in\mathbb{C}^d$ is a normalized state with $\braket{\psi_i|\psi_i}=1$, and $\{\ket{i}\}_{i=0}^{d-1}$ is the computational basis in $d$ dimensions. 
This is because, writing out an arbitrary such operator $C$ in the computational basis, we have $C=\sum_{i=0}^{d-1}\sum_{j=0}^{d-1} c_{ij} \ket{j}\bra{i}$ for some complex coefficients $c_{ij}$. We can then define for every $i$,  $\sum_j c_{ij} \ket{j}:= b'_i \ket{\psi'_i}$ where $b'_i\in \mathbb{C}_i$ and $\braket{\psi'_i|\psi'_i}=1$. Expressing $b'_i$ in the polar form and absorbing the complex phase into $\ket{\psi'_i}$ gives $b'_i \ket{\psi'_i}=b_i \ket{\psi_i}$ where $b_i$ is a non-negative real number and $\braket{\psi_i|\psi_i}=1$. Therefore $C=\sum_i b_i\ket{\psi_i}\bra{i}$. Let $r$ denote the largest such $b_i$, then $a_i:=\frac{b_i}{r}\in [0,1]$ and we obtain \cref{eq:C_arbitrary}.

The circuit of \cref{fig:GeneralCircuit} adapted from \cite{Brun2012} implements the following operator that is operationally equivalent (\cref{def: op_equiv}) to $C$ of \cref{eq:C_arbitrary},
\begin{equation}
    C_{CTC} = \sum_i a_i\proj{\psi_i}{i}_{S}
    \label{eq:CCTC_arbitrary},
\end{equation}
if there exists a set $\{C_i\}_{i=1}^{d-1}$ of operators acting upon $S$ such that for each $i\in\{0,...,d-1\}$
\begin{equation}
    C_i\ket{i}_S = a_i \ket{\psi_i}_S\label{eq:Cond1}.
\end{equation}

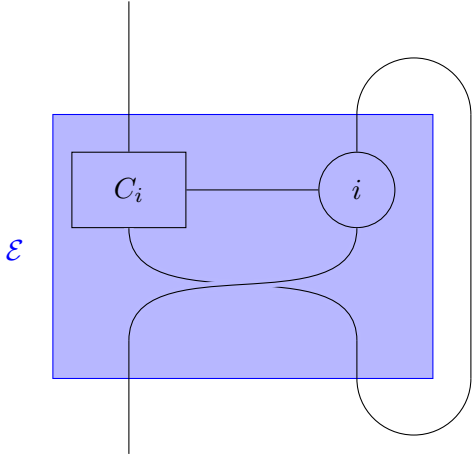
\begin{figure}[ht!]
    \centering
    \scalebox{1}{
        \input{4-Connection/Figs/GeneralCircuitC}
    }
    \caption{The CP map $\mathcal{E}$ that can implement any P-CTC-assisted map $C_{CTC}$ of the form of \cref{eq:CCTC_arbitrary} given a set of operators $\{C_i\}_{i=0}^{d-1}$ of the form of \cref{eq:Cond1}. Here $\mathcal{E}$ is formed by a sequential composition of a SWAP operation (a unitary), followed by a controlled operation (generically, a CP map): $ \sum_{i=0}^{d-1} \left(C_i\right)_S \otimes \pure{i}_{A}$. }
    \label{fig:GeneralCircuit}
\end{figure}

We now show that it is always possible to construct such a set of operators $\{C_i\}_{i=0}^{d-1}$ on $S$, using P-CTCs. First, we show through construction that each such operator $C_i$ on $S$ can be obtained by assisting a unitary $U_i$ on $S$ and a qubit $Q$ by a single P-CTC on $Q$. Plugging this into \cref{fig:GeneralCircuit} would give a construction for our arbitrary 2TO $C$ (operationally equivalent to \cref{eq:CCTC_arbitrary}) in terms of two P-CTCs, of $d$ and 2 dimensions respectively. We then build on this to prove the existence of such a construction involving only one $d$-dimensional P-CTC.

\paragraph{Explicit construction of a circuit assisted by two P-CTCs}

Consider a qubit $Q$ and a unitary operation on it defined as follows, for a given choice of $i\in\{0,...,d-1\}$,

\begin{equation}
    V_i = \begin{pmatrix}
    a_i & -b_i\\
    b_i & a_i
    \end{pmatrix},
    \label{eq:Vi_bis}
\end{equation}
where $b_i = \sqrt{1-a_i^2}$. We extend this to a controlled operation that is a unitary on $S$ and $Q$,

\begin{equation}
\label{eq: ctrl_Ci}
  C_{V_i}:= \pure{i}_S\otimes (V_i)_Q + \sum_{j\neq i}\pure{j}_S\otimes\id_Q 
\end{equation}

Next, consider a unitary $W_i$ on $S$ which acts as follows
\begin{equation}
    \label{eq: Wi}
    W_i\ket{i}_S=\ket{\psi_i}_S
\end{equation}
This is a unitary since $\ket{\psi_i}$ is a normalised state by construction. Then taking the unitary 
\begin{equation}
\label{eq: Ui}
    U_i:=(W_i\otimes \id_Q)(C_{V_i})
\end{equation}
on $S$ and $Q$ formed by applying $W_i$ on $S$ after  $C_{V_i}$, and assisting it with a P-CTC on the qubit $Q$ yields precisely the operator $C_i$ as in \cref{eq:Cond1} (see \cref{fig:Ci}), as the P-CTC together with the controlled operation allows to extract the factor $a_i$ (via the partial trace on the CTC system, c.f. \cref{eq: partialtrace_CTC}) in addition to the state transformation given by $W_i$. An explicit demonstration of this can be found in \hyperref[subsec:DetailsExplConstr]{appendix ~\ref*{subsec:DetailsExplConstr}}. Each such $C_i$ can be constructed analogously, and plugging in the construction of the $C_i$ operators (\cref{fig:Ci}) into the construction of $C_{CTC}$ from these operators (\cref{fig:GeneralCircuit}) we obtain \cref{fig:OverallCircuit}. This shows that an arbitrary 2TO $C$ on the qudit $S$ is operationally equivalent to a P-CTC assisted circuit (\cref{fig:OverallCircuit}) assisted by one $d$ dimensional and one 2 dimensional P-CTC. 

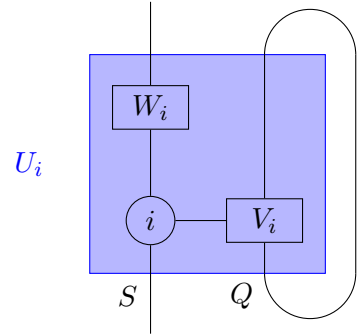
\begin{figure}[ht!]
    \centering
    \scalebox{1}{\input{4-Connection/Figs/Ci_alt}}
    \caption{P-CTC-assisted circuit that implements a general operator $C_i$ as in \cref{eq:Cond1}. The bottom part depicts the controlled operation $C_{V_i}$ of \cref{eq: ctrl_Ci} and the following operation $W_i$ is as in \cref{eq: Wi}. Each of these is a unitary, and hence their composition $U_i$ is unitary.}
    \label{fig:Ci}
\end{figure}

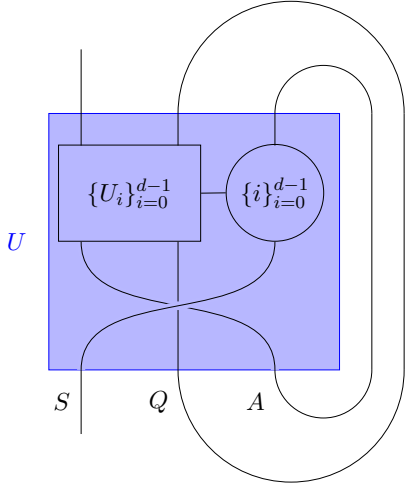
\begin{figure}[ht!]
    \centering
    \scalebox{.85}{\input{4-Connection/Figs/OverallCircuit_alt}}
    \caption{Overall circuit that implements $C_{CTC}=\sum_i a_i\proj{\psi_i}{i}$ (operationally equivalent to an arbitrary 2TO) as a unitary $U$ assisted by two P-CTCs, of dimensions $d$ (on $A$) and 2 (on $Q$). The controlled operation depicted here is of the form: $\sum_{i=0}^{d-1}\pure{i}\otimes U_i$, where $U_i$ is featured in Figure \ref{fig:Ci}.}
    \label{fig:OverallCircuit}
\end{figure}

\paragraph{Existence of a circuit assisted by a single P-CTC}
If all the coefficients $a_i$ in equations \ref{eq:C_arbitrary} and \ref{eq:CCTC_arbitrary} are the same, which means that they can be factored out of the sum, the 2-dimensional P-CTC that is used to extract those factors is not needed anymore. This is because such overall factors preserve operational equivalence. It is possible to find a basis in which all the coefficients are the same, as stated by the following proposition (c.f. \cref{subsec:ExistenceBasis} for the proof and for details).
\begin{restatable}[]{proposition}{ExistenceBasis}
\label{prop: existencebasis}
    Let $C = r\sum_{i=0}^{d-1} a_i\ket{\psi_i}\bra{i}$ with $r\in\mathbb{C}$, $\left|a_i\right|\in[0,1]$ , $\ket{\psi_i}$ are normalized, i.e. $\braket{\psi_i|\psi_i}=1$, and $\ket{i}$ is the d-dimensional computational basis. Then there exists an orthonormal basis $\{\ket{\phi_i}\}_{i=0}^{d-1}$ such that $C=r'\sum_{i=0}^{d-1}\ket{\psi'_i}\bra{\phi_i}$, where $\ket{\psi'_i}$ are normalized $\braket{\psi'_i|\psi'_i}=1$.
\end{restatable}

Given the above proposition, we know that an arbitrary 2TO $C$ can be expressed in the form $C=r'\sum_{i=0}^{d-1}\ket{\psi'_i}\bra{\phi_i}$ for an orthonormal basis $\{\ket{\phi_i}\}_{i=0}^{d-1}$ and normalised states $\ket{\psi'_i}$. Thus an operator 
\begin{equation}
\label{eq: CTC_arbit2}
   C_{CTC}=\sum_{i=0}^{d-1}\ket{\psi'_i}\bra{\phi_i} 
\end{equation} would be operationally equivalent to it. We can construct such a $C_{CTC}$ as a P-CTC assisted map following the same arguments as we did for the operator in \cref{eq:CCTC_arbitrary}. However, in place of \cref{eq:Cond1} where we used the set of operators $\{C_i\}_{i=0}^{d-1}$, we now  use a set $\{W'_i\}_{i=0}^{d-1}$ of operators on $S$, that act as
\begin{equation}
\label{eq: Wp_i}
    W'_i\ket{\phi_i}=\ket{\psi'_i}
\end{equation}
Contrary to $C_i$ in \cref{eq:Cond1}, these $W'_i$ are unitaries (analogous to the $W_i$ in \cref{eq: Wi}).

This means that we no longer need the two dimensional P-CTC of \cref{fig:Ci} and can directly construct the operator of \cref{eq: CTC_arbit2} (operationally equivalent to an arbitrary 2TO $C$) by a single $d$-dimensional P-CTC as shown in \cref{fig:SinglePCTC}.

\begin{figure}[ht!]
    \centering
    \scalebox{0.8}{
        \input{4-Connection/Figs/SinglePCTC_alt}
    }
    \caption{P-CTC-assisted circuit that implements, up to operational equivalence, an arbitrary 2TO $C$ on $S$ expressed in the basis such that the coefficients $a_i$ can be factored out. As shown in \cref{prop: existencebasis}, such a basis always exists. The set of $d$ unitaries $\{W'_i\}_{i=0}^{d-1}$ on $S$ is given by \cref{eq: Wp_i}. This gives the illustrated controlled operation on $S$ and $A$ which is of the form, $\sum_{i=0}^{d-1} (W'_i)_S\otimes \ket{\phi_i}\bra{\phi_i}_A$.
The pre- and post-selection on maximally entangled states are also expressed in the new basis, $\Phi^+ = \sum_{i=0}^{d-1} \ket{\phi_i\phi_i}$. However, noticing that such pre and post-selection involved in a P-CTC on $A$ has equivalent effect to a partial trace on $A$ (\cref{eq: partialtrace_CTC}), which is basis independent, one can see that the result would be the same if we expressed these maximally entangled states in the computational basis.\footnote{Another way to see this basis independence is as follows. Let $U$ be the basis change unitary from the computational to the new basis i.e., $\sum_{i=0}^{d-1} \ket{\phi_i\phi_i}=\sum_{i=1}^{d-1} (U\otimes U) \ket{ii}$. Then forming a P-CTC by pre and post-selecting on $\sum_{i=0}^{d-1} \ket{\phi_i\phi_i}$, and using $UU^\dagger=U^\dagger U=\id$ removes the unitaries on the contracted subsystem in the right-most wire of \cref{fig:SinglePCTC}. Further using the transpose property of maximally entangled states, $\sum_{i=1}^{d-1} (U\otimes \id) \ket{ii}=\sum_{i=0}^{d-1}(\id \otimes U^T)\ket{ii}$, the unitaries $U$ and $U^\dagger$ on the $A$ subsystem can also be moved to the right-most wire and multiplied to the identity, thus resulting in a P-CTC in the computational basis, as required. }
}
    \label{fig:SinglePCTC}
\end{figure}
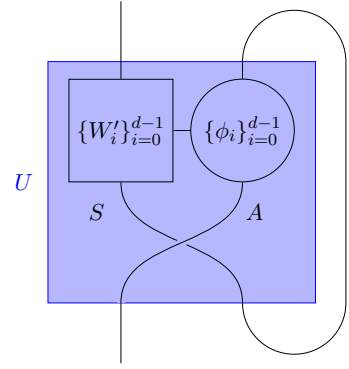

\paragraph{2TOs on multiple systems}
So far, we considered 2TOs where the systems $S_1$ and $S_2$ are the two times were both single $d$-dimensional systems. More generally, we could consider 2TOs on multiple systems, and the above construction would readily extend to these. Consider for instance a 2TO on two systems, where $S_2$ and $S_4$ are at the earlier time while $S_1$ and $S_3$ at the later time:
\begin{equation}
   C = \sum_{i,j,k,l} \alpha_{i,j,k,l} \ket{ki}^{S_3 S_1}\bra{lj}_{S_4 S_2}
    \label{eq:2TO_2Sys}
\end{equation}
The latter can be transformed in a 2TO on one system. Suppose for example that each space is a 2-dimensional Hilbert space, i.e. we have 2 qubits at each of the times. We could then unitarily transform the two qubit computational basis $\{\ket{00},\ket{01},\ket{10},\ket{11}\}$ to the computational basis of a single 4 dimensional system $\{\ket{0},\ket{1},\ket{2},\ket{3}\}$, apply our previous construction and transform back to construct the 2TO $C$ given above (illustrated in \cref{fig:Decoding}). More generally, if the product of the dimensions $d_2d_4$ and $d_1d_3$ differ, then we can choose the maximum of these two and embed the smaller space into the larger one without loss of generality. Thus it follows from the arguments presented thus far that any pure 2TO associated with arbitrary sets of systems at the two times can be recovered through an operationally equivalent P-CTC assisted map, as summarised in the following proposition.

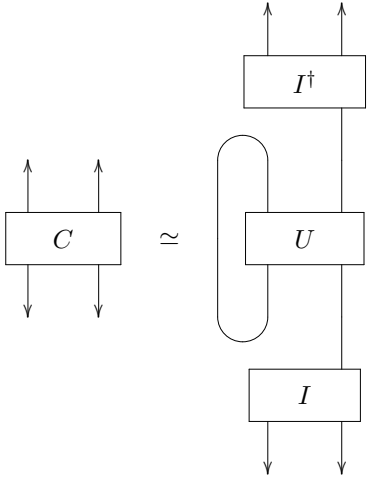
\begin{figure}[ht!]
    \centering
    \scalebox{0.9}{\input{4-Connection/Figs/Decoding}}
    \caption{Any pure 2TO $C$ on $n$ systems of dimensions $d_1$,...,$d_n$ (as forward and as backward evolving spaces)
    can be obtained by
sandwiching our earlier P-CTC assisted construction for arbitrary pure single system 2TOs (e.g., \cref{fig:SinglePCTC}) between encoding $I$ and decoding $I^\dagger$ isometries, where $I$ encodes the states of the $n$ $d_1$,...,$d_n$ dimensional systems into a single $d_1d_2...d_n$-dimensional system and the decoding unitary is its inverse. The dimension of the P-CTC system is therefore also given the product $d_1d_2...d_n$. }

    \label{fig:Decoding}
\end{figure}

\begin{proposition}
\label{lemma: pure2TO_gen}
Any pure 2TO with $n$ backward-evolving state spaces with dimensions $\{d_{i}\}_{i=1}^n$ and $m$ forward-evolving state spaces with dimension $\{\bar{d}_{j}\}_{j=1}^m$, can be obtained in terms of a unitary map assisted by a single P-CTC of dimension $\text{max}\left(\prod_{i=1}^n d_i, \prod_{j=1}^m \bar{d}_j\right)$.
\end{proposition}

\subsection{Connecting pure 2TOs and MTS}\label{subsubsec:2TO-MTS}

In the previous section, we showed that arbitrary pure 2TOs can be obtained through P-CTC assisted maps. Here, we will show how arbitrary pure MTS can be transformed to pure 2TOs, possibly using additional P-CTCs. 
Whether or not this transformation requires additional P-CTCs necessitates distinguishing two approaches.

\paragraph{First approach: disregarding the temporal structure}
Consider a multi-time object
\begin{equation}
\label{eq: 2TS_eg}
    \Psi_1 = \sum_{i,j}\alpha_{ij} \bra{j}_{S_1}\otimes\ket{i}^{S_2}
\end{equation} 
If we regard the forward evolving $S_2$ as being earlier in time than the backward evolving $S_1$, then this would correspond to a 2TS, however if we regard $S_1$ as being earlier in time than $S_2$, then this would correspond to a 2TO. Given only the coefficients $\alpha_{ij}$, the systems $S_1$ and $S_2$ and their direction of evolution (i.e., bras vs kets) is insufficient to distinguish a 2TS from 2TO. We additionally require information on the time order. This can generally be captured by explicitly including time labels on the systems (as in \cref{def:MTspace}), or for visual clarity and simplicity, we can express this information in the written order without including the time labels. For instance, in the absence of explicit time labels, when a two-time object is expressed as in \cref{eq: 2TS_eg}, we would regard it as a 2TS. To regard the object as a 2TO, we would denote it as 
\begin{equation}
\label{eq: 2TO_eg}
    \Psi_2 = \sum_{i,j}\alpha_{ij} \ket{i}^{S_2}\otimes\bra{j}_{S_1}.
\end{equation}

Generally, we will say that two multi-time objects (such as $\Psi_1$ and $\Psi_2$ above) are \emph{isomorphic} if they have the same Hilbert space structure and the same coefficients, and can be transformed reversibly into each other solely by permuting the order of the bras and kets, which corresponds to permuting the time order of the systems. For a more formal definition of this concept, see \cref{appendix: partialorder}. Notice that because they carry distinct labels, the bras and kets in the multi-time notation commute, and can be freely permuted in this manner. 

Figure \ref{fig:2TS_Iso} depicts a 2TS which can be deformed to a 2TO as in \cref{fig:2TO_Iso_b}, through the isomorphism mentioned above. That is, if a 2TS $\Psi$ is contracted with a multi-time object $M$, the resulting object is isomorphic to the one obtained by contracting the isomorphic 2TO with a multi-time object isomorphic to $M$. In particular, two isomorphic multi-time objects contracted with isomorphic measurements yield the same probabilities using the ABL rule. Despite being isomorphic, the two objects are still physically distinct, as they have a different temporal order. 

\begin{figure}[ht!]
    \centering
    \begin{subfigure}{.5\columnwidth}
        \centering
        \scalebox{1}{
            \input{4-Connection/Figs/2TS_Iso}
        }
        \caption{}
        \label{fig:2TS_Iso}
    \end{subfigure}%
    \begin{subfigure}{.5\columnwidth}
        \centering
        \scalebox{1}{
            \input{4-Connection/Figs/2TO_Iso}
        }
        \caption{}
        \label{fig:2TO_Iso_b}
    \end{subfigure}
    \caption{The 2TS in (a) and 2TO in (b) defined with respect to the tensor product Hilbert space $\mathcal{H}^{S_2}\otimes\mathcal{H}_{S_1}$ are isomorphic if they have the same coefficients. In this case, when contracted with the same object $M$ (also deformed by an appropriate isomorphism) with tensor product Hilbert space $\mathcal{H}_{S_2}\otimes\mathcal{H}^{S_1}\otimes\mathcal{H}_{A_1}\otimes\mathcal{H}^{A_2}$, they both yield the same resulting object on $\mathcal{H}_{A_1}\otimes\mathcal{H}^{A_2}$.}
    \label{fig:EqOpEx}
\end{figure}
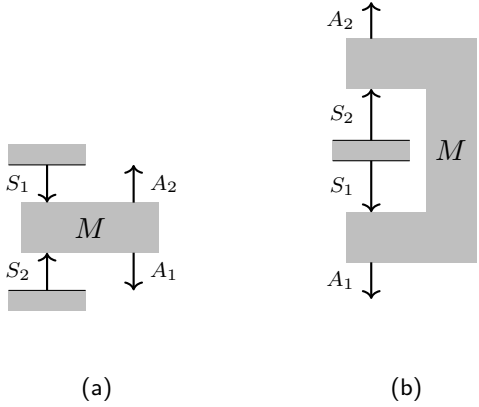
In the same way, any MTS is isomorphic to some 2TO with the same Hilbert space structure and the same coefficients, but where all backward evolving states of the MTS are permuted to be at the earlier time and all forward evolving at the later time. As long as all contractions of the MTS with other operators (e.g., measurements) are also deformed through the same isomorphism, additional P-CTCs would not be required. The 2TO thus obtained can then be related to an operationally equivalent P-CTC-assisted map using the results of \cref{subsubsec:2TO-PCTC}. 

One may however question the physicality of the deformations depicted in \cref{fig:EqOpEx}, despite the mathematical isomorphism. Indeed it conflicts with physical intuitions how one may move a forward evolving system such as $S_2$ in \cref{fig:2TO_Iso_b} backwards in time to obtain \cref{fig:2TS_Iso}, without using a causality violating object like a CTC. In the next paragraph we consider how an MTS can be mapped to a 2TO without trivially collapsing its temporal structure via a mathematical isomorphism, but instead considering how this might be explicitly realized using CTCs.

\paragraph{Second approach: preserving the temporal structure}

Consider a 2TO such $\Phi_2$ of \cref{eq: 2TO_eg}. This can be transformed to an isomorphic 2TS such as $\Psi_1$ of \cref{eq: 2TS_eg} using a single P-CTC as shown in \cref{fig:Relation_2TS_2TO}.

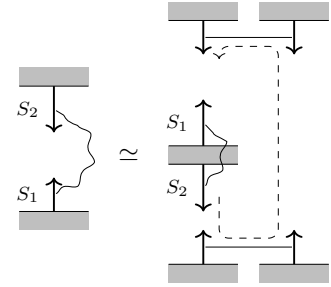
\begin{figure}[ht!]
    \centering
    \scalebox{0.9}{\input{4-Connection/Figs/Relation_2TS_2TO}}
    \caption{The backward-evolving state $S_2$ of the 2TO is teleported to the future using a P-CTC with an open end, to transform the 2TO into an isomorphic 2TS. }
    \label{fig:Relation_2TS_2TO}
\end{figure}

At the level of the circuit representation, a 2TO $C$ is transformed to an isomorphic 2TS-like object 
by bending the in and output wires appropriately, as shown in \cref{fig:Cir_2TS_2TO_U}. Such circuit representations are common in compositional frameworks for operational theories, such as operational probabilistic theories and process theories \cite{Chiribella2010,dAriano2017,Coecke2018,coecke_kissinger_2017}.

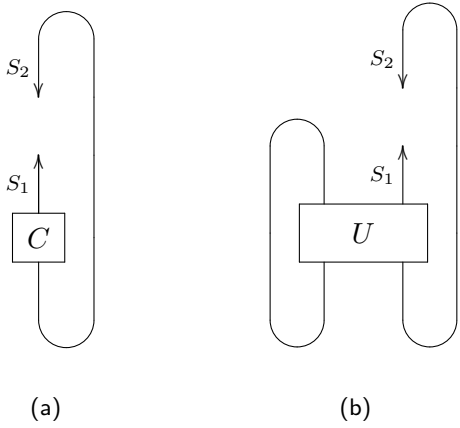
\begin{figure}[ht!]
    \centering
    \begin{subfigure}{.5\columnwidth}
        \centering
        \scalebox{1}{
            \input{4-Connection/Figs/Cir_2TS_2TO}
        }
        \caption{}
        \label{fig:Cir_2TS_2TO}
    \end{subfigure}%
    \begin{subfigure}{.5\columnwidth}
        \centering
        \scalebox{1}{
            \input{4-Connection/Figs/Cir_2TS_2TO_U}
        }
        \caption{}
        \label{fig:Cir_2TS_2TO_U}
    \end{subfigure}
    \caption{(a) The construction of a 2TS from a 2TO in Figure \ref{fig:Relation_2TS_2TO} can be represented in a operational manner, in which the 2TO is represented by the operation $C$ and the wire $S_2$ (which was an open backward evolving system) is bent adequately. (b) Using previous results (\cref{subsubsec:2TO-PCTC}), any 2TO $C$ as in (a) can be obtained in an operationally equivalent manner, using a P-CTC assisted unitary map. This illustrates that the 2TS is operationally equivalent to a time-labelled P-CTC assisted comb, with $S_2$ at a later time than $S_1$. Note that this figure (and similar figures in our paper) is a visual representation of MTS in the style of operational quantum circuits. However it is strictly speaking not a P-CTC assisted map, circuit or comb as it still involves forward and backward evolving spaces as opposed to in an outputs.   }
    \label{fig:Cir_2TS_2TO_}
\end{figure}

The above ideas for two-time states, can be generalized to the multi-time case. To obtain an MTS from a 2TO (such that the two are isomorphic), some of the backward-evolving spaces of the 2TO can be teleported to the appropriate times, using P-CTCs. The number of those depends on the number of forward- and backward-evolving spaces in the MTS. Let $\mathcal{B}$ the set of all backward-evolving spaces and $\mathcal{F}$ the set of all forward-evolving spaces of a given MTS. Consider also the following disjoints sets, where before and after are relative to the time order of the MTS.
\begin{enumerate}
    \item $\mathcal{B}_1$: The set of backward-evolving spaces that are ordered before all the forward-evolving spaces.
    \item $\mathcal{B}_2$: The set of backward-evolving spaces that are ordered after at least one forward-evolving space.
    \item $\mathcal{F}_1$: The set of forward-evolving spaces that are ordered before at least one backward-evolving space.
    \item $\mathcal{F}_2$: The set of forward-evolving spaces that are ordered after all the backward-evolving spaces.
\end{enumerate}
The sets are such that $\mathcal{B} = \mathcal{B}_1\cup \mathcal{B}_2$ and $\mathcal{F} = \mathcal{F}_1\cup \mathcal{F}_2$. We denote by $\left|\cdot\right|$ the cardinality of a set, $d_{\mathcal{B}}$ to be the product of dimensions of all spaces in the set $\mathcal{B}$ and similarly for the remaining sets. 
The 2TS in \cref{fig:Cir_2TS_2TO} is one where $|\mathcal{B}_2|=|\mathcal{B}|=|\mathcal{F}_1|=|\mathcal{F}|=1$. An MTS with a different cardinality of these sets is depicted in \cref{fig:1F3B}.

Using these definitions, we have the following general result for MTS, a proof of which can be found in \cref{appendix: proof_2TOtoMTS}. Figures \ref{fig:IsoMTS} and \ref{fig:IsoMTS_Cir} illustrate the construction entailed in this result for the case of a 4-time state.

\begin{restatable}[]{proposition}{MtwoTO}
\label{prop: MTS_2TO}
   For any given MTS associated with sets of forward and backward evolving spaces $\mathcal{B}_1$, $\mathcal{B}_2$, $\mathcal{F}_1$, $\mathcal{F}_2$ there exists a 2TO from which it can be obtained either using $\left|\mathcal{B}_2\right|$ P-CTCs of dimensions $\{d_{S}\}_{S\in \mathcal{B}_2}$ or using $\left|\mathcal{F}_1\right|$ P-CTCs of dimensions $\{d_{S}\}_{S\in \mathcal{F}_1}$. This 2TO has, in total, $|\mathcal{B}|$ backward evolving spaces of dimensions $\{d_S\}_{S\in \mathcal{B}}$ and $|\mathcal{F}|$ forward-evolving spaces of dimensions $\{d_S\}_{S\in \mathcal{F}}$.
   
\end{restatable}

In the above, the construction using either set $\mathcal{B}_2$ or $\mathcal{F}_1$ work equivalently, one could therefore consider the set with minimal cardinality 
   $\min\left(\left|\mathcal{B}_2\right|, \left|\mathcal{F}_1\right|\right)$ or minimum total dimension $\min\left(d_{\mathcal{B}_2},d_{\mathcal{F}_1}\right)$ for simplicity and without loss of generality. 
   
 Replacing the 2TO in the above result by the operationally equivalent P-CTC assisted map given by \cref{lemma: pure2TO_gen} (in this case, this would involve a single P-CTC of dimension $\max \left(d_{\mathcal{B}},d_{\mathcal{F}}\right)$) would then allow to map an arbitrary pure MTS to a time-labelled P-CTC assisted comb (see \cref{fig:Cir_2TS_2TO_} for an illustration). In total, this will involve as many P-CTCs as given by either one of the options in \cref{prop: MTS_2TO} plus an additional P-CTC of $\max\left(d_{\mathcal{B}},d_{\mathcal{F}}\right)$ dimensions.

\begin{figure}[ht!]
    \centering
    \scalebox{0.9}{\input{4-Connection/Figs/IsoMTS}}
    \caption{A 2TO on two systems can be transformed into a 4TS on a single system using two P-CTCs, that teleport the backward-evolving states to the future, together with a single SWAP.}
    \label{fig:IsoMTS}
\end{figure}
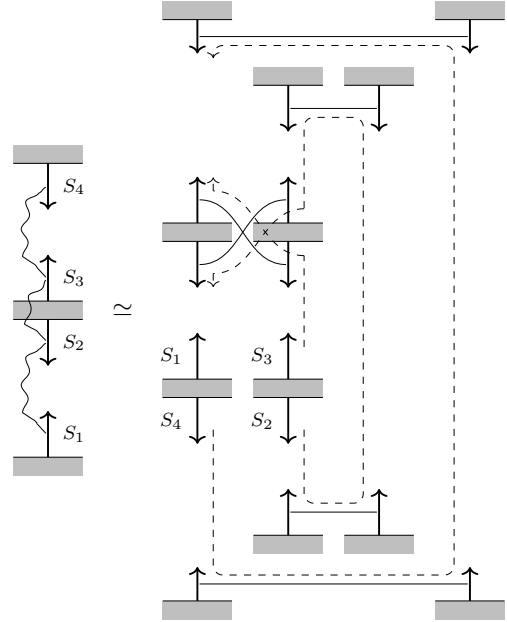

\begin{figure}[ht!]
    \centering
    \scalebox{0.9}{\input{4-Connection/Figs/IsoMTS_Cir}}
    \caption{Circuit style representation of the MTS in \ref{fig:IsoMTS}.} 
    \label{fig:IsoMTS_Cir}
\end{figure}
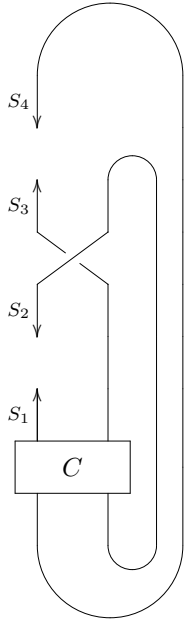

\subsection{Generalisation to mixed objects}\label{subsec:Mixed}
So far we focussed on pure multi-time objects, which we now generalise to the mixed case.

We start again with the two time case before stating the general result. 
Let $\{C_r\}_r$ be a set of pure 2TOs defined on the same multi-time Hilbert space. In \cref{subsubsec:2TO-PCTC}, a construction involving P-CTCs that implements any of the $C_r$ up to some proportionality factor $k_r\in\mathbb{C}$ (which preserves operational equivalence) was given. We denote by $C_r^{CTC}$ the operator obtained using the P-CTC construction, which is such that $C_r =k_rC_r^{CTC}$. We can then see that for any mixture  of these 2TOs given by a corresponding density vector $\eta=\sum_r p_r C_r\otimes C_r^{\dagger}$ with $\sum_r p_r=1$, there exists a corresponding mixture $\xi=\sum_r p_r' C_r^{CTC}\otimes C_r^{CTC\dagger}$ such that $\eta = k\xi$, where $k\in\mathbb{C}$ is a constant. We will denote a mixture such as $\eta$ as $\{p_r,C_r\}_r$. Explicitly, 
\begin{align}
    \eta &= \sum_r p_r C_r\otimes C_r^{\dagger}\nonumber\\
    &= \sum_r p_r\left|k_r\right|^2 C_r^{CTC}\otimes C_r^{CTC\dagger}
\end{align}
Now let $p'_r = \frac{p_r\left|k_r\right|^2}{\sum_s p_s\left|k_s\right|^2}, \forall r$. Then the mixture $\xi=\{p'_r, C_r^{CTC}\}_r$ has the associated density vector $\xi$:
\begin{align}
    \xi &= \sum_r p_r' C_r^{CTC}\otimes C_r^{CTC\dagger} \nonumber\\
    &= \frac{1}{\sum_s p_s\left|k_s\right|^2}\sum_r p_r\left|k_r\right|^2 C_r^{CTC}\otimes C_r^{CTC\dagger} \nonumber\\
    &= \frac{1}{k}\eta
\end{align}

A mixed 2TO $\xi=\{p'_r, C_r^{CTC}\}_r$ can be viewed as a circuit with an additional system that controls which $C_r^{CTC}$ is applied and with which probability (Fig. \ref{fig:MixPCTC_2Figs}). In particular, this is achieved by taking the input of the additional system to be $\rho=\text{diag}\left(\{p_r'\}_r\right)$. Overall the construction is now a P-CTC-assisted circuit with a CPTP map $\mathcal{E}^{CTC}$ instead of a unitary. In other words, we get an analogous result as \cref{lemma: pure2TO_gen} for mixed 2TOs (where each term in the mixture is a pure 2TO, and all of these have the same multi-time Hilbert space), with P-CTC assisted CPTP as opposed to unitary map.

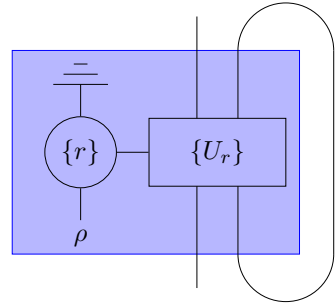
\begin{figure}[ht!]
    \centering
    \scalebox{0.9}{
            \input{4-Connection/Figs/MixPCTC2_alt}
        }
    \caption{A mixed 2TO can be viewed a P-CTC-assisted circuit, with an additional system that controls which operator $C_r^{CTC}$, or equivalently which unitary $U_r$, to perform with which probability $p_r'$. The chosen input state $\rho$ is $\text{diag}\left(\{p_r'\}_r\right)$.}
    \label{fig:MixPCTC_2Figs}
\end{figure}

More generally, consider a mixture $\eta=\{p_r, \psi_r\}_r$ of arbitrary multi-time objects $\psi_r$ associated with the same multi-time Hilbert space. Since the sets $\mathcal{B}_1$, $\mathcal{B}_2$, $\mathcal{F}_1$ and $\mathcal{F}_2$ are fully specified by the multi-time Hilbert space, each $\psi_r$ (a pure MTS) in this mixture has the same structure of these sets. We could thus apply the arguments of \cref{prop: MTS_2TO} to each pure MTS $\psi_r$ to reduce it to a 2TO using appropriate P-CTCs. The number and dimensions of the P-CTCs can be ensured to be the same for each $\psi_r$ in the mixture as these depend only on the above-mentioned sets. 

\begin{proposition}
\label{prop: mixedMTS_to_CTC}
For any mixture $\eta=\{p_r, \psi_r\}_r$ of multi-time objects defined on the same multi-time Hilbert space, associated with sets of spaces $\mathcal{B}_1$, $\mathcal{B}_2$, $\mathcal{F}_1$ and $\mathcal{F}_2$, 
there exists an operationally equivalent time-labelled P-CTC-assisted comb composed of a CPTP map together with either one of the following structure of P-CTCs, where $\mathcal{B}=\mathcal{B}_1\cup \mathcal{B}_2$ and $\mathcal{F}=\mathcal{F}_1\cup \mathcal{F}_2$

\begin{enumerate}
    \item $\left|\mathcal{B}_2\right|$ P-CTCs of dimensions $\{d_{S}\}_{S\in \mathcal{B}_2}$ and an additional P-CTC of $\max \left(d_{\mathcal{B}},d_{\mathcal{F}}\right)$ dimensions.
     \item $\left|\mathcal{F}_1\right|$ P-CTCs of dimensions $\{d_{S}\}_{S\in \mathcal{F}_1}$ and an additional P-CTC of $\max \left(d_{\mathcal{B}},d_{\mathcal{F}}\right)$ dimensions.
\end{enumerate}
\end{proposition}

In particular, \cref{prop: CTC_to_MTS} and \cref{prop: mixedMTS_to_CTC} imply the following theorem that establishes an operational equivalence between the multi-time and P-CTC frameworks.

\Main*

\section{A partial order on isomorphic MTS}
\label{sec: partialorder}

Having proven an operational equivalence between the MTS and P-CTC frameworks in \cref{thm: main}, one can ask a number of further interesting questions. For example, we have looked at translating between different multi-time objects using P-CTCs in \cref{subsubsec:2TO-MTS}. When can we translate between multi-time objects without using any P-CTCs? Since time-labelled P-CTC assisted combs can also be formulated within the MTS formalism, this question can shed light on inter-conversions between time-labelled P-CTC assisted objects and more generally on the resource-theoretic aspect of P-CTCs.

Here we consider the question for isomorphic MTS: as seen in \cref{subsubsec:2TO-MTS}, operationally implementing such an isomorphism can generally involve P-CTCs. We define \emph{free operations} (see \cref{def: free_op} in \cref{appendix: partialorder}) as transformations that map one MTS to another isomorphic MTS without involving any P-CTCs. These operations correspond to time translations along the direction of evolution, shifting forward- or backward-evolving spaces forward or backward in time. 
Using such free operations, we can define a partial order on isomorphic MTS.

\begin{definition}[Partial order on isomorphic MTSs]
\label{def:partialorder}
 Let $M_1$ and $M_2$ be two isomorphic MTS. Whenever $M_1$ can be transformed, only using free operations (\cref{def: free_op}) to $M_2$, we denote this as $M_1 \succeq M_2$ or equivalently $M_2\preceq M_1$. If we have $M_1 \succeq M_2$ and $M_2 \succeq M_1$, we denote this as $M_1=M_2$, 
 and if neither holds, we denote it as $M_1\not\preceq \not\succeq M_2$. Finally, if $M_1 \succeq M_2$ holds but it is impossible to transform $M_2$ to $M_1$ solely using free operations, we denote this as a strict ordering $M_1\succ M_2$ or equivalently $M_2\prec M_1$.
\end{definition}

An example is illustrated in \cref{fig:2TS_swap}, which shows that a 2TS $M_{2TS}$ with a backward evolving system $B$ and forward evolving system $F$ can be freely transformed to a 2TO $M_{2TO}$ with the same systems, giving $M_{2TO}\preceq M_{2TS}$. Note that in this case the free operations manage to flip the relative time-order between $B$ and $F$. Moreover, it is easy to see that going from $M_{2TO}$ to $M_{2TS}$ is impossible with free operations alone: this would require moving $F$ to the past so it comes earlier than $B$ in time, or moving $B$ to the future so it comes later than $F$ in time (both of which can be achieved using a non-free operation, a single P-CTC). Hence we have a strict order, $M_{2TO}\prec M_{2TS}$. The following theorem (proven in \cref{appendix: partialorder}) shows that 2TOs and 2TSs are in fact extremal points in this partial order on isomorphic MTSs, with all other isomorphic MTSs lying in between. 

\begin{figure}
    \centering
    \scalebox{1}{
        \input{4-Connection/Figs/2TS_swap}
    }
    \caption{The order of the states composing a 2TS-like object can be changed using a SWAP operation.}
    \label{fig:2TS_swap}
\end{figure}
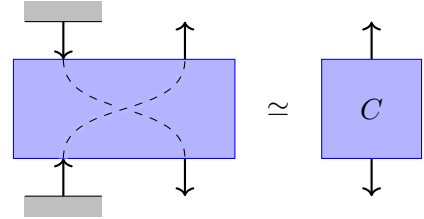

\begin{restatable}[]{theorem}{PartialOrder}
  
  \label{thm:partialorder}
  Consider any 2TO $M_{2TO}$ and 2TS $M_{2TS}$ that are isomorphic to each other. Then for all MTS $M$ which are isomorphic to $M_{2TO}$ and $M_{2TS}$, we have 
  \begin{equation}
  \label{eq: MTS_partialorder}
      M_{2TO} \prec M \prec M_{2TS}.
  \end{equation}
  Moreover, this is not generally a total order, i.e., there exist isomorphic MTS $M_1$ and $M_2$ such that $M_1\not\preceq \not\succeq M_2$.
 
\end{restatable}

{\bf Strictness of the partial order}
\cref{def: 2TO2TS} of 2TO and 2TS does not require
all backward/forward evolving spaces to have the same time label. 
For example, a situation where $S_1$ is backward evolving at a time $t_1$, while $S_2$ and $S_3$ are forward evolving at distinct times $t_2$ and $t_3$, where $t_2,t_3>t_1$ corresponds to a 2TO $M$ by \cref{def: 2TO2TS}. If we also required the same time labels for identically directed systems, then the order obtained in \cref{eq: MTS_partialorder} would not be strict. This is because a 2TO $M'$ where $S_2$ and $S_3$ have the same time label $t>t_1$ is equivalent to the above 2TO $M$ under the partial order, noting that equivalence of MT objects only considers the time order and not the specific time labels (\cref{def:timeorder_equiv}): $M$ can be transformed to $M'$ by stretching $S_2$ forward in time from $t_2$ to $t_3$ (taking $t_3>t_2$ w.l.o.g.), while $M'$ can be transformed to $M$ by stretching $S_3$ forward in time to any $t'>t$. 

{\bf Isomorphism vs operational equivalence} It is worth noting that the concept of operational equivalence vs isomorphism between MTS are distinct. The former is operational and the latter is mathematical. Two MTS that differ only by an overall constant $k\in \mathbb{C}$ are operationally equivalent as they yield the same outcome probabilities under all measurements. This also means that such MTS have the same MT Hilbert space (and can thus be contracted with the same MT measurement operators), and we do not speak of operational equivalence between MTS of ``different types'' such as a 2TO and 2TS. On the other hand, two MTS which are isomorphic need not have the same MT Hilbert space as they can have different time ordering, and MTS of different types can therefore be isomorphic. The concepts are however related in some cases, for instance the 2TOs $M$ and $M'$ of the previous paragraph are isomorphic and operationally equivalent.

Based on these results, an interesting future direction would be to consider whether $M_1\succ M_2$ allows an MTS $M_1$ to strictly outperform the MTS $M_2$ in some information processing task. This is further discussed in \cref{sec:Discussion}.

{\bf Induced order on P-CTC assisted objects} Due to the operational equivalence between MTS and P-CTC assisted objects established in \cref{thm: main}, the partial order defined here for isomorphic MTS induces a partial order on P-CTC assisted objects. In particular, this implies, together with \cref{thm:partialorder}, that any P-CTC assisted comb (analogous to an MTS) can be translated to an isomorphic P-CTC assisted map (analogous to a 2TO) using free operations. 

%% file: 4-Connection/Figs/GeneralCircuitC.tex
\begin{tikzpicture}

    \draw[blue, fill=blue!70!white, fill opacity=0.4] (-1,-2.5) -- (4, -2.5) -- (4, 1) -- (-1, 1) -- cycle;
    
    \node[rectangle, draw, minimum width = 1.5cm, minimum height=1cm] (c) at (0,0) {$C_i$};
    \node[circle, draw, minimum size=1cm] (d) at (3, 0) {$i$};

    \draw (c.east) -- (d.west);

    \node[rectangle] (e) at (0,-2) {};
    \node[rectangle] (f) at (3,-2) {};

    \draw[draw=blue!28!white,double=black,double distance=\pgflinewidth,ultra thick] (f.center) to[out=90, in=-90] (c.south);
    \draw [draw=blue!28!white,double=black,double distance=\pgflinewidth,ultra thick](e.center) to[out=90,in=-90] (d.south);

    \draw (c.north) -- (0, 2.5);
    \draw (e.center) -- (0, -3.5);
    \draw (f.center) -- (3,-2.5);
    \draw (d.north) -- (3,1);

    \draw (4.5, -2.5) -- (4.5, 1);
    \draw (3,-2.5) arc (180:360:0.75);
    \draw (3, 1) arc (180:0:0.75);

    \node[blue] at (-1.5, -0.8) {$\mathcal{E}$};

    
\end{tikzpicture}

%% file: 4-Connection/Figs/Ci_alt.tex
\begin{tikzpicture}

    \draw[blue, fill=blue!70!white, fill opacity=0.4] (-0.8, -0.7) -- (2.3,-0.7) -- (2.3, 2.2) -- (-0.8, 2.2) -- cycle;

    \node[circle, draw] (a) at (0,0){$i$};
    \node[rectangle, draw, minimum width=1cm] (b) at (1.5,0){$V_i$};
    \node[rectangle, draw, minimum width=1cm] (c) at (0,1.5){$W_i$};

    \draw (0,-1.5) -- (a.south);
    \draw (1.5,-0.7) -- (b.south);
    \draw (b.north) -- (1.5, 2.2);
    \draw (a.north) -- (c.south);
    \draw (c.north) -- (0, 2.9);
    \draw (a.east) -- (b.west);
    \draw (2.7, -0.7) -- (2.7, 2.2);

    \draw (1.5,-0.7) arc (180:360:0.6);
    \draw (1.5, 2.2) arc (180:0:0.6);

    \node[blue] at (-1.6, 0.75) {$U_i$};

    \node[rectangle] at (-0.3,-1) {$S$};
    \node[rectangle] at (1.2, -1) {$Q$};
    
\end{tikzpicture}

%% file: 4-Connection/Figs/OverallCircuit_alt.tex
\begin{tikzpicture}

    \draw[blue, fill=blue!70!white, fill opacity=0.4] (-0.5, 0) -- (4, 0) -- (4, 4) -- (-0.5, 4) -- cycle;

    \draw (1.5,0) -- (1.5,2);
    \draw[draw=blue!28!white,double=black,double distance=\pgflinewidth,ultra thick] (3, 0) to[out=90, in=-90] (0,2);
    \draw [draw=blue!28!white,double=black,double distance=\pgflinewidth,ultra thick](0, 0) to[out=90,in=-90] (3, 2);

    \node[rectangle, draw, minimum width = 2.2cm, minimum height=1.5cm, anchor=south] (c) at (0.75,2) {$\{U_i\}_{i=0}^{d-1}$};
    \node[circle, draw, anchor=south, minimum size=1cm] (d) at (3, 2) {$\{i\}_{i=0}^{d-1}$};

    \draw (c.east) -- (d.west);
    \draw (1.5, 3.5) -- (1.5, 4);
    \draw (0, 3.5) -- (0, 5);
    \draw (0,0) -- (0,-1);
    \draw (3,3.5) -- (3,4);

    \draw (4.5, 0) -- (4.5, 4);
    \draw (3,0) arc (180:360:0.75);
    \draw (3, 4) arc (180:0:0.75);

    \draw (5, 0) -- (5, 4);
    \draw (1.5,0) arc (180:360:1.75);
    \draw (1.5, 4) arc (180:0:1.75);

    \node[blue] at (-1, 2) {$U$};

    \node[rectangle] at (-0.3,-0.5) {$S$};
    \node[rectangle] at (1.2, -0.5) {$Q$};
    \node[rectangle] at (2.7, -0.5) {$A$};
    
\end{tikzpicture}

%% file: 4-Connection/Figs/SinglePCTC_alt.tex
\begin{tikzpicture}
    \draw[blue, fill=blue!70!white, fill opacity=0.4] (-1.2, 0) -- (3.2, 0) -- (3.2, 4) -- (-1.2, 4) -- cycle;

    \draw (2, 0) to[out=90, in=-90] (0,2);
    \draw [draw=blue!28!white,double=black,double distance=\pgflinewidth,ultra thick](0, 0) to[out=90,in=-90] (2, 2);
    
    \node[rectangle, draw, minimum height=1.7cm, anchor=south] (c) at (0,2) {$\{W_i'\}_{i=0}^{d-1}$};
    \node[circle, draw, anchor=south, minimum size=1cm] (d) at (2, 2) {$\{\phi_i\}_{i=0}^{d-1}$};

    \draw (c.east) -- (d.west);
    \draw (c.north) -- (0,5);
    \draw (d.north) -- (2,4);
    \draw (0,0) -- (0,-1);

    \draw (3.7, 0) -- (3.7, 4);
    \draw (2,0) arc (180:360:0.85);
    \draw (2, 4) arc (180:0:0.85);

    \node[blue] at (-1.6, 2) {$U$};

    \node[rectangle] at (-0.4, 1.5) {$S$};
    \node[rectangle] at (2.2, 1.5) {$A$};
\end{tikzpicture}

%% file: 4-Connection/Figs/Decoding.tex
\mbox{
    \Qcircuit @C = 1em @R = 2em{
    &&&&\ar@{<-}[d]&\\
    &&&&\multigaten{1}{I^{\dagger}}{2} & \multinghost{I^{\dagger}}{2}\ar[u]\\
    \ar@{<-}[d]&&&\protect\postselection&&\qwx\\
    \multigaten{1}{C}{2}&\multinghost{C}{2}\ar[u]\ar[d]&\tick{\simeq}&\qwx&\multigate{1}{U}{2}{}&\multinghost{U}{2}\qwx\\
    \ar@{<-}[u]&&&\protect\preselection{}&\qwx&\qwx\\
    &&&&\multigaten{1}{I}{2}&\multinghost{I}{2}\ar[d]\qwx\\
    &&&&\ar@{<-}[u]&\\
    }
}

%% file: 4-Connection/Figs/2TS_Iso.tex
\begin{tikzpicture}[
    op/.style={shape= Op, minimum width = \WidthPrePost,font=\scriptsize},
    pre/.style={shape = Pre, minimum width = \WidthPrePost,font=\scriptsize},
    post/.style={shape = Post, minimum width = \WidthPrePost,font=\scriptsize},
    Kop/.style={draw, minimum width=\WidthMeasurement, minimum height=\WidthMeasurement,font=\scriptsize},
    t/.style={font=\scriptsize},
    ghost/.style={minimum width=\WidthMeasurement, minimum height=\WidthMeasurement,font=\scriptsize},
    decoration={snake, segment length=4mm, amplitude=0.5mm}
]
\pgfsetmatrixrowsep{0.5cm}
\pgfsetmatrixcolumnsep{0.3cm}
\pgfmatrix{rectangle}{center}{mymatrix}
{\pgfusepath{}}{\pgfpointorigin}{\let\&=\pgfmatrixnextcell}
{
\\
\node[post](post1) {};   \& \node(post2) {}; \\
\node[ghost](g1){};      \& \node[ghost](g2){}; \\
\node[pre](pre1) {};     \& \node(pre2) {}; \\
\\
}

\draw[->,thick] (post1.south) -- (g1.north) node [midway,inner sep=0 cm] (1) {};
\draw[->,thick] (pre1.north) -- (g1.south) node [midway,inner sep=0 cm] (2) {};
\draw[->,thick] (g2.north) -- (post2.south) node [midway,inner sep=0 cm] (3) {};
\draw[->,thick] (g2.south) -- (pre2.north) node [midway,inner sep=0 cm] (4) {};

\node[t, left =.5 mm of 1](S1){$S_1$};
\node[t, left =.5 mm of 2](S2){$S_2$};
\node[t, right =.5 mm of 3](A2){$A_2$};
\node[t, right =.5 mm of 4](A1){$A_1$};

\begin{scope}[on background layer]
    \fill[lightgray] (g1.north west) rectangle (g2.south east);
\end{scope}

\node[fit=(g1.north west) (g2.south east), inner sep=0cm, minimum height=\WidthMeasurement,label=center:$M$](M){};

\end{tikzpicture}

%% file: 4-Connection/Figs/2TO_Iso.tex
\begin{tikzpicture}[
    op/.style={shape= Op, minimum width = \WidthPrePost,font=\scriptsize},
    pre/.style={shape = Pre, minimum width = \WidthPrePost,font=\scriptsize},
    post/.style={shape = Post, minimum width = \WidthPrePost,font=\scriptsize},
    Kop/.style={draw, minimum width=\WidthMeasurement, minimum height=\WidthMeasurement,font=\scriptsize},
    t/.style={font=\scriptsize},
    ghost/.style={minimum width=\WidthMeasurement, minimum height=\WidthMeasurement},
    decoration={snake, segment length=4mm, amplitude=0.5mm}
]
\pgfsetmatrixrowsep{0.5cm}
\pgfsetmatrixcolumnsep{0.2cm}
\pgfmatrix{rectangle}{center}{mymatrix}
{\pgfusepath{}}{\pgfpointorigin}{\let\&=\pgfmatrixnextcell}
{
\node[ghost](g1){};    \& \\
\node[ghost](g2){};    \& \node[ghost](g3){};\\
\node[op](o1){};        \& \node[ghost](g4){$M$}; \\
\node[ghost](g5){};      \& \node[ghost](g6){}; \\
\node[ghost](g7){};     \&  \\
}

\draw[->,thick] (g2.north) -- (g1.south) node [midway,inner sep=0 cm] (1) {};
\draw[->,thick] (o1.north) -- (g2.south) node [midway,inner sep=0 cm] (2) {};
\draw[->,thick] (o1.south) -- (g5.north) node [midway,inner sep=0 cm] (3) {};
\draw[->,thick] (g5.south) -- (g7.north) node [midway,inner sep=0 cm] (4) {};

\node[t, left =.5 mm of 1](A2){$A_2$};
\node[t, left =.5 mm of 2](S2){$S_2$};
\node[t, left =.5 mm of 3](S1){$S_1$};
\node[t, left =.5 mm of 4](A1){$A_1$};

\begin{scope}[on background layer]
    \fill[lightgray] (g2.north west) rectangle (g3.south east);
    \fill[lightgray] (g5.north west) rectangle (g6.south east);
    \fill[lightgray] (g3.north west) rectangle (g6.south east);
\end{scope}

\end{tikzpicture}

%% file: 4-Connection/Figs/Relation_2TS_2TO.tex
\begin{tikzpicture}[
    op/.style={shape= Op, minimum width = \WidthPrePost,font=\scriptsize},
    pre/.style={shape = Pre, minimum width = \WidthPrePost,font=\scriptsize},
    post/.style={shape = Post, minimum width = \WidthPrePost,font=\scriptsize},
    Kop/.style={draw, minimum width=\WidthMeasurement, minimum height=\WidthMeasurement,font=\scriptsize},
    t/.style={font=\scriptsize},
    ghost/.style={minimum width=\WidthMeasurement, minimum height=\WidthMeasurement,font=\scriptsize},
    decoration={snake, segment length=4mm, amplitude=0.5mm}
]
\pgfsetmatrixrowsep{0.5cm}
\pgfsetmatrixcolumnsep{0.3cm}
\pgfmatrix{rectangle}{center}{mymatrix}
{\pgfusepath{}}{\pgfpointorigin}{\let\&=\pgfmatrixnextcell}
{
\&\&\node[post](post2){};     \& \node[post](post3){};  \\
\node[post](post1){};\&\&\node[ghost](g21){};\& \node[rectangle,minimum size =\WidthMeasurement](g22){};\\
\node[ghost](Kop1){};\&\node{$\simeq$};\&\node[op](op2){}; \&  \\
\node[pre](pre1){};\&\&\node[rectangle](g41){};     \& \node[rectangle](g42){};  \\
\&\&\node[pre](pre2) {};    \& \node[pre](pre3) {}; \\
}

\draw[->,thick] (post2.south) -- (g21.north) node [midway,inner sep=0 cm] (21) {};
\draw[->,thick] (op2.north) -- (g21.south) node [midway,inner sep=0 cm] (31) {};
\draw[->,thick] (op2.south) -- (g41.north) node [midway,inner sep=0 cm] (41) {};
\draw[->,thick] (pre2.north) -- (g41.south) node [midway,inner sep=0 cm] (51) {};
\draw[->,thick] (post3.south) -- (g22.north) node [midway,inner sep=0 cm] (22) {};
\draw[->,thick] (pre3.north) -- (g42.south) node [midway,inner sep=0 cm] (52) {};
\draw[->,thick] (post1.south) -- (Kop1.north) node [midway,inner sep=0 cm] (00) {};
\draw[->,thick] (pre1.north) -- (Kop1.south) node [midway,inner sep=0 cm] (20) {};

\draw[-] (51) -- (52);
\draw[-] (21) -- (22);
\draw[decorate] (20) .. controls ($(20)+(0.7,0.3)$) and  ($(00)+(0.5,-0.3)$).. (00);
\draw[decorate] (41) .. controls ($(41)+(0.3,0.3)$) and  ($(31)+(0.3,-0.3)$).. (31);

\node[t, left=.5 mm of 31](A12){$S_1$};
\node[t, left=.5 mm of 41](A22){$S_2$};
\node[t, left=.5 mm of 20](s2){$S_1$};
\node[t, left=.5 mm of 00](s3){$S_2$};

\coordinate[below right=1mm and 2mm of 41](start){};
\coordinate[above right=1mm and 2mm of 51](c1){};
\coordinate[above left=1mm and 2mm of 52](c2){};
\coordinate[below left=1mm and 2mm of 22](c3){};
\coordinate[below right=1mm and 2mm of 21](c4){};
\coordinate[below right=2mm and 2mm of 21](finish){};

\draw[-<,dashed, rounded corners] (start) -- (c1) -- (c2) -- (c3) -- (c4) -- (finish);

\end{tikzpicture}

%% file: 4-Connection/Figs/Cir_2TS_2TO.tex
\mbox{
    \Qcircuit @C = 1em @R = 2em{
    & \\
    \protect\postselection\ar[d]_-{S_2} & \\
    & \qwx\\
    &\qwx\\
    \gate{C}{}{1}\ar[u]^-{S_1}& \qwx\\
    \protect\preselection{} & \qwx\\
    & \\
    }
}

%% file: 4-Connection/Figs/Cir_2TS_2TO_U.tex
\mbox{
    \Qcircuit @C = 1em @R = 2em{
    &&& \\
    & & \protect\postselection\ar[d]_-{S_2} & \\
     & & & \qwx \\
    \protect\postselection & & & \qwx \\
    \qwx & \multigate{1}{U}{2}{} & \multinghost{U}{2}\ar[u]^-{S_1} & \qwx \\
    \protect\preselection{} & \qwx & \protect\preselection{} & \qwx \\
    & & & \\
    }
}

%% file: 4-Connection/Figs/IsoMTS.tex
\begin{tikzpicture}[
    op/.style={shape= Op, minimum width = \WidthPrePost,font=\scriptsize},
    pre/.style={shape = Pre, minimum width = \WidthPrePost,font=\scriptsize},
    post/.style={shape = Post, minimum width = \WidthPrePost,font=\scriptsize},
    Kop/.style={draw, minimum width=\WidthMeasurement, minimum height=\WidthMeasurement,font=\scriptsize},
    t/.style={font=\scriptsize},
    ghost/.style={minimum width=\WidthMeasurement, minimum height=\WidthMeasurement,font=\scriptsize},
    decoration={snake, segment length=4mm, amplitude=0.5mm}
]
\pgfsetmatrixrowsep{0.5cm}
\pgfsetmatrixcolumnsep{0.3cm}
\pgfmatrix{rectangle}{center}{mymatrix}
{\pgfusepath{}}{\pgfpointorigin}{\let\&=\pgfmatrixnextcell}
{
\& \& \node[post](13){}; \& \& \& \node[post](16){}; \\
 \& \& \node[ghost](23){}; \& \node[post](24){}; \& \node[post](25){}; \& \node[ghost](26){};\\
\node[post](31){}; \& \& \node[ghost](33){}; \& \node[ghost](34){}; \& \node[ghost](35){}; \&  \\
\node[ghost](41){};\& \& \node[op](43){}; \& \node[op](44){}; \& \& \\
\node[op](51){};\& \node[rectangle]{$\simeq$}; \& \node[ghost](53){}; \& \node[ghost](54){}; \& \& \\
\node[ghost](61){};\&  \& \node[op](63){}; \& \node[op](64){}; \& \& \\
\node[pre](71){};\& \& \node[ghost](73){}; \& \node[ghost](74){}; \& \node[ghost](75){}; \& \\
\& \& \node[ghost](83){}; \& \node[pre](84){}; \& \node[pre](85){}; \& \node[ghost](86){}; \\
\& \&\node[pre](93){}; \& \& \& \node[pre](96){}; \\
}

\draw[->,thick] (31.south) -- (41.north) node [midway,inner sep=0 cm] (31to41) {};
\draw[->,thick] (51.north) -- (41.south) node [midway,inner sep=0 cm] (51to41) {};
\draw[->,thick] (51.south) -- (61.north) node [midway,inner sep=0 cm] (51to61) {};
\draw[->,thick] (71.north) -- (61.south) node [midway,inner sep=0 cm] (71to61) {};

\draw[->,thick] (13.south) -- (23.north) node [midway,inner sep=0 cm] (13to23) {};
\draw[->,thick] (43.north) -- (33.south) node [midway,inner sep=0 cm] (43to33) {};
\draw[->,thick] (43.south) -- (53.north) node [midway,inner sep=0 cm] (43to53) {};
\draw[->,thick] (63.north) -- (53.south) node [midway,inner sep=0 cm] (63to53) {};
\draw[->,thick] (63.south) -- (73.north) node [midway,inner sep=0 cm] (63to73) {};
\draw[->,thick] (93.north) -- (83.south) node [midway,inner sep=0 cm] (93to83) {};

\draw[->,thick] (24.south) -- (34.north) node [midway,inner sep=0 cm] (24to34) {};
\draw[->,thick] (44.north) -- (34.south) node [midway,inner sep=0 cm] (44to34) {};
\draw[->,thick] (44.south) -- (54.north) node [midway,inner sep=0 cm] (44to54) {};
\draw[->,thick] (64.north) -- (54.south) node [midway,inner sep=0 cm] (64to54) {};
\draw[->,thick] (64.south) -- (74.north) node [midway,inner sep=0 cm] (64to74) {};
\draw[->,thick] (84.north) -- (74.south) node [midway,inner sep=0 cm] (84to74) {};

\draw[->,thick] (25.south) -- (35.north) node [midway,inner sep=0 cm] (25to35) {};
\draw[->,thick] (85.north) -- (75.south) node [midway,inner sep=0 cm] (85to75) {};

\draw[->,thick] (16.south) -- (26.north) node [midway,inner sep=0 cm] (16to26) {};
\draw[->,thick] (96.north) -- (86.south) node [midway,inner sep=0 cm] (96to86) {};

\draw[-] (13to23) -- (16to26);
\draw[-] (24to34) -- (25to35);
\draw[-] (84to74) -- (85to75);
\draw[-] (93to83) -- (96to86);

\draw[-] (43to53) to[out=0,in=180] (44to34);
\draw[-] (43to33) to[out=0,in=180] (44to54);

\draw[decorate] (31to41) .. controls ($(31to41)+(-0.3,-0.1)$) and  ($(51to41)+(-.3,0.1)$).. (51to41);
\draw[decorate] (51to41) .. controls ($(51to41)+(-0.3,-0.3)$) and  ($(51to61)+(-0.3,+0.3)$).. (51to61);
\draw[decorate] (51to61) .. controls ($(51to61)+(-0.3,-0.1)$) and  ($(71to61)+(-.3,0.1)$).. (71to61);

\node[t, right=.5 mm of 31to41](S4){$S_4$};
\node[t, right=.5 mm of 51to41](S3){$S_3$};
\node[t, right=.5 mm of 51to61](S2){$S_2$};
\node[t, right=.5 mm of 71to61](S1){$S_1$};

\node[t, left=.5 mm of 64to74](S41){$S_2$};
\node[t, left=.5 mm of 64to54](S31){$S_3$};
\node[t, left=.5 mm of 63to73](S21){$S_4$};
\node[t, left=.5 mm of 63to53](S11){$S_1$};

\coordinate[below right=1mm and 2mm of 63to73](start1){};
\coordinate[above right=1mm and 2mm of 93to83](c11){};
\coordinate[above left=1mm and 2mm of 96to86](c12){};
\coordinate[below left=1mm and 2mm of 16to26](c13){};
\coordinate[below right=1mm and 2mm of 13to23](c14){};
\coordinate[below right=2mm and 2mm of 13to23](finish1){};
\draw[-<, dashed, rounded corners] (start1) -- (c11) -- (c12) -- (c13) -- (c14) -- (finish1);

\coordinate[below right=1mm and 2mm of 64to74](start2){};
\coordinate[above right=1mm and 2mm of 84to74](c21){};
\coordinate[above left=1mm and 2mm of 85to75](c22){};
\coordinate[below left=1mm and 2mm of 25to35](c23){};
\coordinate[below right=1mm and 2mm of 24to34](c24){};
\coordinate[below right=1mm and 2mm of 44to34](c25){};
\coordinate[below right=1mm and 2mm of 43to53](c26){};
\coordinate[below right=3mm and 2mm of 43to53](finish2){};

\draw[-, dashed, rounded corners] (start2) -- (c21) -- (c22) -- (c23) -- (c24) -- (c25);
\draw[-, dashed] (c25) to[out=180,in=0] (c26);
\draw[->, dashed] (c26) -- (finish2);

\coordinate[above right=1mm and 2mm of 64to54](start3){};
\coordinate[above right=1mm and 2mm of 44to54](c31){};
\coordinate[above right=1mm and 2mm of 43to33](c32){};
\coordinate[above right=3mm and 2mm of 43to33](finish3){};

\draw[-,dashed] (start3) -- (c31);
\draw[-,dashed] (c31) to[out=180,in=0] (c32);
\draw[->, dashed] (c32) -- (finish3);

\end{tikzpicture}

%% file: 4-Connection/Figs/IsoMTS_Cir.tex
\mbox{
    \Qcircuit @C = 1em @R = 2em{
    & & & \\
    & & & \\
    \protect\postselection[3]\ar[d]_-{S_4} &  & & \\
    & & & \qwx\\
    & \protect\postselection & & \qwx\\
    \ar[u]^-{S_3}\ar@{-}[dr]|\hole& \qwx & \qwx & \qwx \\
    \ar[d]_-{S_2}\ar@{-}[ur] & & \qwx & \qwx \\
     & \qwx &\qwx &\qwx \\
     & \qwx &\qwx &\qwx \\
     \multigate{1}{C}{2}{}\ar[u]^-{S_1} & \multinghost{C}{2}\qwx & \qwx & \qwx\\
    \protect\preselection[3]{} & \protect\preselection{} & \qwx & \qwx\\
    & & & \\
    & & &\\
    }
}

%% file: 4-Connection/Figs/MixPCTC2_alt.tex
\begin{tikzpicture}
    \draw[blue, fill=blue!70!white, fill opacity=0.4] (-1, -1.5) -- (3.2, -1.5) -- (3.2, 1.5) -- (-1, 1.5) -- cycle;
    
    \node[circle, draw](a) at (0,0) {$\{r\}$};
    \node[rectangle, draw, minimum width=2cm, minimum height= 1cm](b) at (2,0) {$\{U_r\}$};

    \draw (a.south) -- (0,-1);
    \draw (a.north) -- (0,1);
    \draw (a.east) -- (b.west);

    \node[anchor=north] at (0,-1) {$\rho$};

    \draw (-0.4,1) -- (0.4,1);
    \draw (-0.25,1.15) -- (0.25,1.15);
    \draw (-0.1, 1.3) -- (0.1, 1.3);

    \draw (1.7, -2) -- (1.7, -0.5);
    \draw (1.7, 0.5) -- (1.7, 2);
    \draw (2.3, -1.5) -- (2.3, -0.5);
    \draw (2.3, 0.5) -- (2.3, 1.5);

    \draw (3.7, -1.5) -- (3.7, 1.5);
    \draw (2.3,-1.5) arc (180:360:0.7);
    \draw (2.3,1.5) arc (180:0:0.7);
    
\end{tikzpicture}

%% file: 4-Connection/Figs/2TS_swap.tex
\begin{tikzpicture}[
    op/.style={shape= Op, minimum width = \WidthPrePost,font=\scriptsize},
    pre/.style={shape = Pre, minimum width = \WidthPrePost,font=\scriptsize},
    post/.style={shape = Post, minimum width = \WidthPrePost,font=\scriptsize},
    Kop/.style={draw, minimum width=\WidthMeasurement, minimum height=\WidthMeasurement,font=\scriptsize},
    t/.style={font=\scriptsize},
    ghost/.style={minimum width=(\WidthMeasurement*2), minimum height=(\WidthMeasurement*2),font=\scriptsize},
    decoration={snake, segment length=4mm, amplitude=0.5mm}
]
\pgfsetmatrixrowsep{0.5cm}
\pgfsetmatrixcolumnsep{0.3cm}
\pgfmatrix{rectangle}{center}{mymatrix}
{\pgfusepath{}}{\pgfpointorigin}{\let\&=\pgfmatrixnextcell}
{
\\
\node[post](post1) {};   \& \node(post2) {};    \&                          \& \node(post3){};\\
\node[ghost](g1){};      \& \node[ghost](g2){}; \& \node(op){$\simeq$};     \& \node[ghost](g3){};\\
\node[pre](pre1) {};     \& \node(pre2) {};     \&                          \& \node(pre3){};  \\
\\
}

\draw[->,thick] (post1.south) -- (g1.north) node [midway,inner sep=0 cm] (1) {};
\draw[->,thick] (pre1.north) -- (g1.south) node [midway,inner sep=0 cm] (2) {};
\draw[->,thick] (g2.north) -- (post2.south) node [midway,inner sep=0 cm] (3) {};
\draw[->,thick] (g2.south) -- (pre2.north) node [midway,inner sep=0 cm] (4) {};
\draw[->,thick] (g3.north) -- (post3.south) node [midway, inner sep=0 cm] (5) {};
\draw[->,thick] (g3.south) -- (pre3.north) node [midway, inner sep=0 cm] (6) {};


\node[draw=blue, fill=blue!30!white, fit=(g1.north west) (g2.south east), inner sep=0cm, minimum height=\WidthMeasurement](M){};
\node[draw=blue, fill=blue!30!white, fit=(g3.north west) (g3.south east), inner sep=0cm, minimum height=\WidthMeasurement,label=center:$C$](C){};

\draw[dashed] (g1.north) to[in=90, out=270] (g2.south);
\draw[dashed, draw=blue!28!white,double=black,double distance=\pgflinewidth,ultra thick] (g1.south) to[in=270, out=90] (g2.north);

\end{tikzpicture}

%% file: 5-Discussion/5-Discussion.tex
\section{Discussion and outlook}\label{sec:Discussion}

Our results (summarised in \cref{sec:Intro}) provide a foundation for investigating a number of further questions on the nature of causality and time in quantum theory, which we discuss below.

{\bf Efficiency} There exist two main preparation methods for MTS: the SWAP method of \cite{Aharonov2009} (Figure 6) and the one using P-CTCs, presented in our work. Both methods use post-selection on maximally entangled states, in the former case with arbitrary pre-selected states and in the latter case with pre-selection on the maximally entangled state but an arbitrary channel assisted by the resulting P-CTC. Hence a natural question for future work is: which preparation method is more efficient i.e., leads to the largest post-selection success probability? A related open question is whether the number and dimensions of P-CTCs required in our mapping from MTS to P-CTC assisted combs (\cref{prop: mixedMTS_to_CTC}) is optimal, or whether there exists a alternate mapping that preserves operational equivalence while minimizing these resources.

{\bf Resource theory of P-CTCs and MTS} 
In \cref{sec: partialorder}, we defined a partial order on isomorphic MTS, where free operations on such MTS are those that do not use any P-CTCs. We found that 2TOs are least useful while 2TS are maximally useful among all isomorphic MTS, when P-CTCs (which allow to reverse the direction of causality) are regarded as resources. By the equivalence established in \cref{thm: main}, the partial order defined on isomorphic MTS in \cref{sec: partialorder} would induce an operationally equivalent partial order on time-labelled P-CTC assisted combs. It would be interesting to explore if this can be developed into a full resource theory of objects in the P-CTCs and MTS frameworks. More concretely, it remains an open question if there are there any information processing tasks where objects higher in the partial order perform strictly better, e.g., a task where a 2TS strictly outperforms its isomorphic 2TO. Several interesting classes of causal loops have been found across the quantum information literature, including associated to processes violating causal inequalities \cite{Oreshkov2012,Baumeler_2014, Baumeler_2016,Tobar_2020}, to causal loops that can be embedded in Minkowski space-time without superluminal signalling \cite{VilasiniColbeckPRL, VilasiniColbeckPRA}. These can be seen as special cases of P-CTCs, and also described via cyclic quantum causal models \cite{ferradini2025_quantum, ferradini2025_classical}. It would be intriguing to explore whether such a resource-theory might shed light on the information-processing power of such causal loops (as also motivated in \cite{VilasiniColbeck2024}), how this might link to structural properties (e.g., graph separation properties) of the cyclic causal model associated with them and to relativistic principles when such models are embedded in a space-time.

{\bf Paradoxes} Both pre and post selection as in the MTF formalism and causal loops as in the P-CTC formalism entail paradoxes \cite{Aharonov1964,Aharonov1991,Aharonov_2013, aharonov2014,Lloyd2011,Lloyd2011a}. In the former case, pre and post selection paradoxes in quantum theory are linked to the non-classical resource of contextuality \cite{Pusey2015}. The latter case typically entails the grandfather and information/bootstrap paradoxes which are also found in purely classical theories with causal loops. Can the equivalence proven in this paper be applied to identify new forms of cyclic causality and time travel paradoxes of a genuinely quantum nature, with links to contextuality?

{\bf Related frameworks and generalisations} There are several related frameworks \cite{Coecke_2012a,Coecke_2012b,Oreshkov_2015,coecke_kissinger_2017,Leifer_2017,pinzani2019,selby2024,Oreshkov_2016,hardy2021,Chiribella_2022,Hardy_2012, Fitzsimons_2016,Fullwood2022}, with which it would be interesting to connect our results and thereby possibly generalize them. In particular, this includes category-theoretic formalisms such as process theories, where P-CTC-like cyclic structures can be defined within any compact closed symmetric monoidal category which come equipped with special states and effects called “cups” and “caps” modeling the relevant pre and post-selection in a P-CTC. Combining our work with this extensive literature, it would be interesting to explore whether the MTF and P-CTC approaches can be generally formulated  in any compact closed symmetric monoidal category (not necessarily one instantiated through quantum Hilbert spaces), and consequently whether an operational equivalence between cyclic causality and time symmetry can be established in a theory-independent manner, which generalises our results for the quantum case.

{\bf Physical realisability and fine-graining} In \cite{vilasini2022embedding, VilasiniRennerPRL}, it was shown that physical realizations of so-called indefinite causal order (ICO) processes respecting relativistic causality in spacetime necessarily unravel into a sequence of well-defined and causally ordered operations at a more fine-grained level, without requiring post-selection\footnote{See also \cite{Vilasini2025} for details on the assumptions regarding the underlying ``laboratories'' under which such an unraveling exists and a possible resolution to the debate \cite{Procopio_2015,Rubino2017,Chiribella2013,Oreshkov2012, Portmann2017,Vilasini_mastersthesis,Oreshkov2019,Paunkovic2020,Ormrod_2023,vilasini2022embedding,VilasiniRennerPRL,kabel2024} on the physical interpretation of ICO experiments in Minkowski spacetime.}. Such acyclic fine-grained descriptions have been found for a large class of processes, known as quantum circuits with quantum controlled superposition of causal order \cite{Salzger_2025, Salzger_thesis}. Noting that ICO processes are a subset of the MTS and P-CTC formalisms, it would be of fundamental interest to determine the largest set of P-CTCs and MTSs that can admit such an acyclic fine-grained description in spacetime, without involving post-selection. This will also shed light on the ability to simulate CTCs using causal well-ordered processes, while allowing (as in the causality frameworks of \cite{Portmann_2017,VilasiniColbeckPRA, Salzger_2025}) to consider physical protocols where messages are exchanged at superpositions of spacetime locations.

{\bf Spacetime and boundary conditions} The MTS formalism allows for arbitrary past and future (temporal) boundary conditions through arbitrary pre- and post-selections. It would be interesting to investigate if this framework (and consequently the P-CTC framework, via our mapping) can be extended to consider spatiotemporal boundaries, which would bring it closer to relativistic approaches \cite{Feynman1948,Oeckl_2003a,Oeckl_2003b}, and facilitate a more information-theoretic understanding of spatio-temporal boundary conditions.

\medskip

{\it Acknowledgements} We thank \"{A}min Baumeler for insightful discussions and interest in this work. VV acknowledges support from an ETH Postdoctoral Fellowship, ETH Zurich Quantum Center, the Swiss National Science Foundation via project No.\ 200021\_188541, the QuantERA programme via project No.\ 20QT21\_187724, and from a government grant managed by the Agence Nationale de la Recherche under the Plan France 2030 with the reference ANR-22-PETQ-0007.

\medskip

{\it Author contributions} Author names are listed in alphabetical order. VV proposed the project and all three authors, EJ, RS and VV, contributed equally to deriving the results presented in this manuscript. The project was part of EJ's master's thesis supervised by RS and VV, where the thesis includes additional examples and analysis performed by EJ. All authors contributed to writing the manuscript, with EJ and VV taking the lead equally.

%% file: 6-Appendix/6-0-MTF.tex
\section{MTS framework: details and proofs}\label{sec:AppendixMTF}

\subsection{Composition}\label{app:MTcomposition}

\begin{definition}[Composition]\label{def:MTcomposition}
Consider MT Hilbert spaces $\mathcal{H}_1, \mathcal{H}_2$ given by
\begin{align}
    \mathcal{H}_1 &= \bigotimes_{i \in G^1} \mathcal{H}^{A_i} \bigotimes_{j \in G_1} \mathcal{H}_{A_j}, \\
    \mathcal{H}_2 &= \bigotimes_{i \in G^2} \mathcal{H}^{A_i} \bigotimes_{j \in G_2} \mathcal{H}_{A_j},
\end{align}
where the sets of labels of the forward and backward evolving spaces in $\mathcal{H}_1$ are denoted by $\{G^1,G_1\}$ and similarly for $\mathcal{H}_2$. The Hilbert spaces are only composable if every system label occurs not more than once for each direction of evolution, i.e. $G^1 \cap G^2 = \emptyset = G_1 \cap G_2$. The composition of Hilbert spaces contracts (removes) all reversed pairs --- which are the intersections $G^1 \cap G_2$ and $G_1 \cap G^2$ --- while forming the usual tensor product on the rest,
\begin{align}
    \mathcal{H}_1 \bullet \mathcal{H}_2 &= \bigotimes_{i \in G^{12}} \mathcal{H}^{A_i} \bigotimes_{j \in G_{12}} \mathcal{H}_{A_j}, \\
    \text{where} \; G^{12} &= \left( G^1 - G_2 \right) \cup \left( G^2 - G_1 \right), \\
        G_{12} &= \left( G_1 - G^2 \right) \cup \left( G_2 - G^1 \right)
\end{align}

Analogously, given states $\Psi_1 \in \mathcal{H}_1, \Psi_2 \in \mathcal{H}_2$ expanded as
\begin{align}
    \Psi_1 &= \bigotimes_{i \in G^1} \ket{\psi_i}^{A_i} \bigotimes_{j \in G_1} \bra{\psi_j}_{A_j}, \\
    \Psi_2 &= \bigotimes_{i \in G^2} \ket{\phi_i}^{A_i} \bigotimes_{j \in G_2} \bra{\phi_j}_{A_j},
\end{align}
the composition forms the inner product on reversed pairs and the usual tensor product on the rest,
\begin{align}
    \Psi_1 \bullet \Psi_2 &= \beta \bigotimes_{i \in G^1 - G_2} \ket{\psi_i}^{A_i} \bigotimes_{i \in G^2 - G_1} \ket{\phi_i}^{A_i} \nonumber \\
    & \quad \bigotimes_{j \in G_1 - G^2} \bra{\psi_j}_{A_j} \bigotimes_{j \in G_2 - G^1} \bra{\phi_j}_{A_j}, \\
    \text{where} \; \beta &= \prod_{i \in G^1 \cap G_2} \braket{\phi_i|\psi_i} \prod_{j \in G_1 \cap G^2} \braket{\psi_j|\phi_j}
\end{align}
\end{definition}

\subsection{Positivity: Proof of Lemma \ref{lemma:spectral_thm} of the main text}\label{app:compositionpositivity}
\begin{proof}
For composition, let $\eta$ and $\lambda$ be positive vectors from $\mathcal{H}_1$ and $\mathcal{H}_2$ respectively, i.e.
\begin{align}
    \eta &= \sum_r a_r \Psi_r \otimes \Psi_r^\dagger, \quad a_r > 0, \; \Psi_r \in \mathcal{H}_1, \\
    \lambda &= \sum_s b_s \Phi_s \otimes \Phi_s^\dagger, \;\quad b_s > 0, \;\; \Phi_s \in \mathcal{H}_2.
\end{align}
Then the composition $\eta \bullet \lambda$ is also positive:
\begin{align}
    \eta \bullet \lambda &= \sum_{rs} a_r b_s \left( \Psi_r \otimes \Phi_s \right) \otimes \left( \Psi_r^\dagger \otimes \Phi_s^\dagger \right).
\end{align}

For the trace, we prove the statement for the trace over a forward evolving Hilbert space, the rest follows analogously. Let $\eta$ be a positive MT vector from the Hilbert space $\mathcal{H} \otimes \mathcal{H}^A$, i.e. $\eta \in \mathcal{H} \otimes \mathcal{H}^A \otimes \mathcal{H}^\dagger \otimes \mathcal{H}_{A^\dagger}$. Take any vector $u \in \bar{\mathcal{H}}$,
\begin{align}
    \Tr^A &\left[ \eta \right] \bullet \left( u \otimes u^\dagger \right) \nonumber \\
    &= \eta \bullet \left( \sum_i \bra{i}_A \otimes \ket{i}^{A^\dagger} \right) \bullet \left( u \otimes u^\dagger \right) \\
    &= \sum_i \eta \bullet \left( \left( u \otimes \bra{i}_A \right) \otimes \left( u^\dagger \otimes \ket{i}^{A^\dagger} \right) \right) \\
    &= \sum_i \eta \bullet \left( v_i \otimes v_i^\dagger \right) \\
    &\geq 0,
\end{align}
where $v_i = u_i \otimes \ket{i}_A \in \bar{\mathcal{H}} \otimes \mathcal{H}_A$, the reversed space to $\mathcal{H} \otimes \mathcal{H}^A$, and the last statement follows from the positivity of $\eta$. Thus $\Tr^A[\eta]$ is also positive.

\end{proof}

\subsection{Normalisation of states and instruments}\label{app:MTnormalisation}

From an operational point of view, there is no fixed norm for an MT state because such a norm would correspond to the probability that the post-selection required to create the MT state succeeded, however this can vary as there are many different post-selections that would give rise to the same MT object --- therefore the norm has no meaning on its own.

One can however pick a norm to single out a representative of an equivalence class of states, one standard choice is that an MT-state $\eta$ as in \eqref{eq:definitionMTstate} is normalised similar to a density matrix, i.e with each $\Psi_r \in \mathcal{H}$ being normalised, and the set $\{q_r\}$ being a normalised probability distribution. Note that this differs from the normalisation of process matrices: a process matrix $W$ that is equivalent to $\eta$ has norm equal to the total dimension of the backward-evolving Hilbert spaces.

There is also the normalisation of measurements in standard quantum theory: for any set of Kraus operators $\{A_k\}$ corresponding to a measurement one has that $\sum_k A_k^\dagger A_k = \mathds{1}$. This does not translate into the MT formalism as one can create more general measurements using the post-selection as discussed. Thus there is no set convention for normalising measurements. Note however a MT instrument is composed of density vectors, each being positive --- this restriction is the analog of the positivity of POVM elements $E_k = A_k^\dagger A_k$ in standard quantum theory.

%% file: 6-Appendix/PCTC_probabilities.tex
\section{P-CTC framework: details and proofs}
\label{sec:AppendixPCTC}

\subsection{Action of P-CTC assisted maps: proof of lemma~\ref{lemma: PCTC_map_action}}
\label{appendix: PCTC_map_action}

\PCTCAction*
\begin{proof}
We note that this proof is almost identical to the proof of the statement for the case of $\mathcal{E}$ being a unitary, which is given in \cite{Brun2012}. However, we repeat the proof for general CPTP maps here for completeness. 

We begin by expressing $\mathcal{E}$ in its Kraus decomposition: $\mathcal{E}(\sigma_{SA})=\sum_j K_j \sigma_{SA}K_j^{\dagger}$ for all $\sigma_{SA}\in \mathcal{L}(\mathcal{H}_S)\otimes \mathcal{L}(\mathcal{H}_A)$. Here the Kraus operators $\{K_j\}_j$ are linear operators satisfying $\sum_j K_j^\dagger K_j=\id_{SA}$. Using this and applying equation~\eqref{eq: PCTC_map}, we obtain the following, where $A'$ is a system isomorphic to $A$ and $\rho_S=\sum_i \lambda_i\ket{\psi_i}\bra{\psi_i}_S$ is an arbitrary input state.\footnote{Note that any density matrix $\rho_S$ can be expressed in this form, as a convex mixture of pure states through its spectral decomposition.}
\begin{align}
\label{eq: PCTC_action_proof1}
    \begin{split}
         \mathcal{C}_{CTC}(\rho_S)&=\bra{\Phi^+}_{AA'} \mathcal{E} \otimes \id_{A'} (\rho_S\otimes \ket{\Phi^+}\bra{\Phi^+}_{AA'})\ket{\Phi^+}_{AA'}\\
         &= \sum_j \bra{\Phi^+}_{AA'} K_j \otimes \id_{A'} (\rho_S\otimes \ket{\Phi^+}\bra{\Phi^+}_{AA'}) K_j^\dagger \otimes \id_{A'} \ket{\Phi^+}_{AA'}\\
         &= \sum_j \bra{\Phi^+}_{AA'} K_j \otimes \id_{A'} (\sum_i \lambda_i \ket{\psi_i}\bra{\psi_i}_S\otimes \ket{\Phi^+}\bra{\Phi^+}_{AA'}) K_j^\dagger \otimes \id_{A'} \ket{\Phi^+}_{AA'}\\
       &=  \sum_j \sum_i \lambda_i  \bra{\Phi^+}_{AA'} K_j \otimes \id_{A'} (\ket{\psi_i}\bra{\psi_i}_S\otimes \ket{\Phi^+}\bra{\Phi^+}_{AA'}) K_j^\dagger \otimes \id_{A'} \ket{\Phi^+}_{AA'}\\
       &=\sum_j \sum_i \lambda_i  \bra{\Phi^+}_{AA'} K_j \otimes \id_{A'} (\ket{\psi_i}_S\ket{\Phi^+}_{AA'}\otimes \bra{\psi_i}_S\bra{\Phi^+}_{AA'}) K_j^\dagger \otimes \id_{A'} \ket{\Phi^+}_{AA'}.
    \end{split}
\end{align}

Consider now the term $\bra{\Phi^+}_{AA'} K_j \otimes \id_{A'} (\ket{\psi_i}_S\ket{\Phi^+}_{AA'})$, which can be simplified as follows.
\begin{align}
    \begin{split}
   \bra{\Phi^+}_{AA'} K_j \otimes \id_{A'} (\ket{\psi_i}_S\ket{\Phi^+}_{AA'})&=   \frac{1}{\sqrt{d_A}} \sum_k\bra{kk}_{AA'} K_j \otimes \id_{A'} (\ket{\psi_i}_S \frac{1}{\sqrt{d_A}}\sum_l\ket{ll}_{AA'})\\
   &=  \frac{1}{d_A} \sum_k\bra{kk}_{AA'} K_j \otimes \id_{A'} (\ket{\psi_i}_S \sum_l\ket{ll}_{AA'})\\
   &= \frac{1}{d_A} \sum_k\sum_l \bra{k}K_j\ket{l}_A \ket{\psi_i} \bra{k}\ket{l}_{A'}\\
   &=\frac{1}{d_A} \sum_k\bra{k}K_j\ket{k}_A \ket{\psi_i}_S \\
   &=\frac{1}{d_A}\tr_A[K_j]\ket{\psi_i}_S 
    \end{split}
\end{align}

Using this in equation~\eqref{eq: PCTC_action_proof1}, we immediately have the following, which establishes the statement of lemma~\ref{lemma: PCTC_map_action}.
\begin{align}
\begin{split}
\label{eq:PCTC_map_action_final}
     \mathcal{C}_{CTC}(\rho_S)&=\frac{1}{d_A^2}\sum_j\sum_i \lambda_i \tr_A[K_j] (\ket{\psi_i}\bra{\psi_i}_S )\tr_A[K_j^\dagger]\\
    &= \frac{1}{d_A^2}\sum_j \tr_A[K_j] (\sum_i \lambda_i\ket{\psi_i}\bra{\psi_i}_S )\tr_A[K_j^\dagger]\\
    &= \frac{1}{d_A^2}\sum_j \tr_A[K_j] (\rho_S)\tr_A[K_j^\dagger].
\end{split}
\end{align}

\end{proof}

\subsection{Measurement probabilities for P-CTC-assisted maps}
\label{appendix: PCTC_map_prob}

Here we derive equation~\eqref{eq: PCTC_map_prob} for the measurement probabilities associated with P-CTC-assisted maps. Consider the P-CTC assisted map $\mathcal{C}_{CTC}:\mathcal{L}(\mathcal{H}_S)\mapsto \mathcal{L}(\mathcal{H}_S)$ which consists of the CPTP map $\mathcal{E}:\mathcal{L}(\mathcal{H}_S)\otimes \mathcal{L}(\mathcal{H}_A)\mapsto \mathcal{L}(\mathcal{H}_S)\otimes \mathcal{L}(\mathcal{H}_A)$ assisted by a P-CTC on $A$. The desired probability $P(a|\mathcal{C}_{CTC},\rho_S)$ corresponds to the probability that a measurement outcome $a$ is obtained, given that the map $\mathcal{C}_{CTC}$ is implemented on an initial state $\rho_S$ and a subsequent measurement $\mathcal{M}:=\{\mathcal{M}^a\}_{a\in A}$ is performed (\cref{fig:PCTC_map_meas}). By the rule of conditional probabilities this is given as 
\begin{equation}
 P(a|\mathcal{C}_{CTC},\rho_S,\{\mathcal{M}^a\}_{a})=\frac{P(a,\mathcal{C}_{CTC}|\rho_S,\{\mathcal{M}^a\}_{a})}{P(\mathcal{C}_{CTC}|\rho_S,\{\mathcal{M}^a\}_{a})}   
\end{equation}

where $P(a, \mathcal{C}_{CTC}|\rho_S,\{\mathcal{M}^a\}_{a})$ is the probability that the map $\mathcal{C}_{CTC}$ is successfully applied \emph{and} the outcome $a$ is obtained for the measurement $\mathcal{M}=\{\mathcal{M}^a\}_{a}$ performed on the final state of $S$ given the initial state $\rho_S$. $P(\mathcal{C}_{CTC}|\rho_S, \{\mathcal{M}^a\}_{a})$ is the probability that $\mathcal{C}_{CTC}$ is successfully applied given the initial state $\rho_S$ and that the measurement $\mathcal{M}=\{\mathcal{M}^a\}_{a}$ is performed on the final state of $S$. Each of these can be calculated using the Born rule. $P(a, \mathcal{C}_{CTC}|\rho_S,\{\mathcal{M}^a\}_{a})$ is the probability that a Bell measurement on the $AA'$ subsystem, when applied to the state $\sigma_{SAA'}:=\mathcal{E}\otimes \id_{A'} (\rho_S\otimes \ket{\Phi^+}\bra{\Phi^+}_{AA'})$ yields the outcome corresponding to $\ket{\Phi^+}_{AA'}$ and a subsequent measurement $\mathcal{M}$ on $S$ yields the outcome $a$. Applying the Born rule, we have 
\begin{align}
\begin{split}
      P(a, \mathcal{C}_{CTC}|\rho_S,\{\mathcal{M}^a\}_{a})&= \tr[\sigma_{SAA'}\mathcal{M}^a\otimes \ket{\Phi^+}\bra{\Phi^+}_{AA'}]\\&=\tr_{SAA'}[(\mathcal{E}\otimes \id_{A'} (\rho_S\otimes \ket{\Phi^+}\bra{\Phi^+}_{AA'}))\mathcal{M}^a\otimes \ket{\Phi^+}\bra{\Phi^+}_{AA'}]\\&= \tr_S\Big[ \tr_{AA'}[\ket{\Phi^+}\bra{\Phi^+}_{AA'}(\mathcal{E}\otimes \id_{A'} (\rho_S\otimes \ket{\Phi^+}\bra{\Phi^+}_{AA'}))]\mathcal{M}^a\Big]\\&=
       \tr_S\Big[\bra{\Phi^+}_{AA'}(\mathcal{E}\otimes \id_{A'} (\rho_S\otimes \ket{\Phi^+}\bra{\Phi^+}_{AA'}))\ket{\Phi^+}_{AA'}\mathcal{M}^a\Big]\\&=
       \tr[\mathcal{C}_{CTC}(\rho_S)\mathcal{M}^a]
\end{split} 
\end{align}

Now consider the term $P(\mathcal{C}_{CTC}|\rho_S,\{\mathcal{M}^a\}_{a})$. We have
\begin{align}
    \begin{split}
     P(\mathcal{C}_{CTC}|\rho_S,\{\mathcal{M}^a\}_{a})&=  \sum_{a\in A} P(a, \mathcal{C}_{CTC}|\rho_S,\{\mathcal{M}^a\}_{a})\\&=\sum_{a\in A} \tr[\mathcal{C}_{CTC}(\rho_S)\mathcal{M}^a]\\&= \tr[\mathcal{C}_{CTC}(\rho_S)\sum_{a\in A}\mathcal{M}^a]\\&=\tr[\mathcal{C}_{CTC}(\rho_S)],
    \end{split}
\end{align}
where we have used the fact that $\sum_{a\in A}\mathcal{M}^a=\id$. Putting these together yields the desired equation~\eqref{eq: PCTC_map_prob}, which we repeat here for convenience. 

\begin{equation}  
\label{eq: PCTC_map_proof}P(a|\mathcal{C}_{CTC},\rho_S,\{\mathcal{M}^a\}_{a})=\frac{\tr[\mathcal{C}_{CTC}(\rho_S)\mathcal{M}^a]}{\tr[\mathcal{C}_{CTC}(\rho_S)]}
\end{equation}

%% file: 6-Appendix/PCTC-Combs.tex
\subsection{Time labelled P-CTC assisted quantum combs}
\label{appendix: PCTC_Combs}

Here, we develop the concept of time labelled P-CTC-assisted combs introduced in the main text, in full technical details. This builds upon the quantum comb formalism introduced in \cite{Chiribella2008} by including P-CTCs and time-labells therein. We begin by reviewing the original quantum combs formalism. 

\subsubsection{Review of regular quantum combs}
\label{sec: combs}

We briefly review here the quantum combs formalism from \cite{Chiribella2008}, before extending them with P-CTCs in the next sub-section.

{\bf Choi representation and link product}
The basic primitives here are linear and completely positive maps, which describe physical operations (such as quantum channels and measurements) in quantum circuits. A linear completely positive map $\mathcal{E}: \mathcal{L}(\mathcal{H}_I)\mapsto \mathcal{L}(\mathcal{H}_O)$ from a set $I$ of input systems to a set $O$ of output systems can be conveniently represented as an operator living in the joint in-output space $\mathcal{L}(\mathcal{H}_I)\otimes \mathcal{L}(\mathcal{H}_O)$ through the Choi-Jamio\l{}kowski representation \cite{Choi1975, Jamiolkowski1972}.

\begin{equation}
\mathfrak{C}(\mathcal{E})=\sum_{i,j} \ket{i}\bra{j}_I\otimes \mathcal{E}(\ket{i}\bra{j}_I)
\end{equation}

\begin{figure}
    \centering
\begin{tikzpicture}[scale=0.5, transform shape]

\begin{scope}[shift={(2.5,15)}]
    \draw (0,0) rectangle node[align=center]{\large{$\mathcal{A}$}} (3,2);

\draw(0.5,0)--(0.5,-2) node[below, yshift=-2mm] {I};  \draw (2.5,0)--(2.5,-2) node[below, yshift=-2mm] {J}; \draw (0.5,2)--(0.5,4) node[above, yshift=2mm] {K}; \draw (2.5,2)--(2.5,4) node[above, yshift=2mm] {\textbf{L}}; \node at (4,0) {\Huge{,}}; \draw[->](4,-3.5)--node[anchor=west]{\Large{composition}}(4,-5.5);

\end{scope}

\begin{scope}[shift={(7.5,15)}]
    \draw (0,0) rectangle node[align=center]{\large{$\mathcal{B}$}} (3,2);

 \draw(0.5,0)--(0.5,-2) node[below, yshift=-2mm] {\textbf{L}};  \draw (2.5,0)--(2.5,-2) node[below, yshift=-2mm] {M}; \draw (0.5,2)--(0.5,4) node[above, yshift=2mm] {N}; \draw (2.5,2)--(2.5,4) node[above, yshift=2mm] {O};

\end{scope}

\draw (0,0) rectangle node[align=center]{\large{$\mathcal{A}$}} (3,2);
\draw(0.5,0)--(0.5,-2) node[below, yshift=-2mm] {I};  \draw (2.5,0)--(2.5,-2) node[below, yshift=-2mm] {J}; \draw (0.5,2)--(0.5,8) node[above, yshift=2mm] {K}; \draw (2.5,2)--(2.5,4);

\draw (2,4) rectangle node[align=center]{\large{$\mathcal{B}$}} (5,6);
\draw (4.5,4)--(4.5,-2) node[below, yshift=-2mm] {M};\draw (2.5,6)--(2.5,8)  node[above, yshift=2mm] {N};\draw (4.5,6)--(4.5,8) node[above, yshift=2mm] {O}; 

\node at (3,3) {\textbf{L}}; \node at (6.5,3) {\Huge{$=$}};

\draw (8.5,2) rectangle node[align=center]{\large{$\mathcal{C}$}} (12.5,4);
 \draw(9,2)--(9,0) node[below, yshift=-2mm] {I};   \draw(10.5,2)--(10.5,0) node[below, yshift=-2mm] {J}; \draw(12,2)--(12,0) node[below, yshift=-2mm] {M};
\draw(9,4)--(9,6) node[above, yshift=2mm] {K};   \draw(10.5,4)--(10.5,6) node[above, yshift=2mm] {N}; \draw(12,4)--(12,6) node[above, yshift=2mm] {O};

\end{tikzpicture}
    \caption{Composition of linear CP maps}
    \label{fig:composition}
\end{figure}
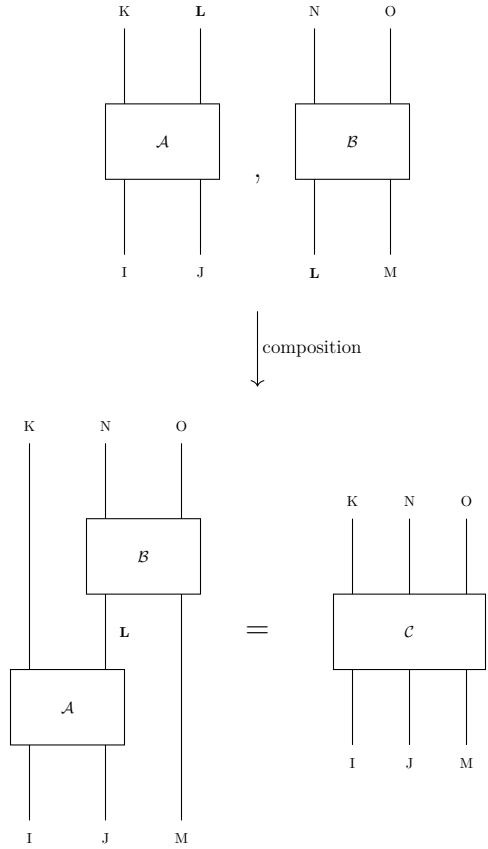

As short hand, we will denote this as $\mathcal{E}: I\mapsto O$ and $\mathfrak{C}(\mathcal{E})\in I\otimes O$ (referred to as the Choi operator) by using the system labels to also denote their state spaces, when they are clearly distinguished by the context. Moreover, we will simplify the notation for tensor products $I\otimes O$ to $IO$.

We can compose different maps together to form quantum circuits, which can involve connecting an output subsystem of one map to an input subsystem of another map where the connected systems have isomorphic state spaces. Each in/output subsystem of such a map comes with a distinct label, except pairs of in/output systems that are connected together, which will share the same label. Each input can be composed with a unique output with the same label and vice versa.

For instance, we can have a linear CP map $\mathcal{A}$ with two inputs labelled $I$ and $J$ and two outputs labelled $K$ and $L$, and another CP map $\mathcal{B}$ with two inputs labelled $L$ and $M$ and two outputs labelled $N$ and $O$ with the subsystem $L$ in both cases having the same state space. We can then connect the two maps through the system $L$ to obtain a new map $\mathcal{C}$ with inputs $I$, $J$, and $M$ and outputs $K$, $N$ and $O$ as shown in \cref{fig:composition}.
This composition can be neatly described in the Choi representation of the maps, through the concept of the link product.

\begin{definition}[Composition of linear CP maps]
The composition of two linear CP maps $\mathcal{A}: IJ\mapsto KL$ and $\mathcal{B}: LM\mapsto NO$ through the system $L$ results in a new CP map $\mathcal{C}: IJM\mapsto KNO$, whose Choi representation $\mathfrak{C}(\mathcal{C})\in IJMKNO$ is related to the Choi representations $\mathfrak{C}(\mathcal{A})\in IJKL$ and $\mathfrak{C}(\mathcal{B})\in LMNO$ of the original maps through the link product
\begin{equation}
    \mathfrak{C}(\mathcal{C})=\mathfrak{C}(\mathcal{A})*\mathfrak{C}(\mathcal{B}):=\tr_L\Big(\mathfrak{C}(\mathcal{A})^{T_L}\mathfrak{C}(\mathcal{B})\Big),
\end{equation}
   where $T_L$ denotes a partial transpose on the connecting subsystem $L$ and the product $\mathfrak{C}(\mathcal{A})^{T_L}\mathfrak{C}(\mathcal{B})$ is short hand for $(\mathfrak{C}(\mathcal{A})^{T_L}\otimes \id_{MNO})(\id_{IJK}\otimes \mathfrak{C}(\mathcal{B}))$. 
\end{definition}

The link product is commutative, $\mathfrak{C}(\mathcal{A})*\mathfrak{C}(\mathcal{B})=\mathfrak{C}(\mathcal{B})*\mathfrak{C}(\mathcal{A})$ due to cyclicity of the trace and can also be shown to be associative $(\mathfrak{C}(\mathcal{A})*\mathfrak{C}(\mathcal{B}))*\mathfrak{C}(\mathcal{C})=\mathfrak{C}(\mathcal{A})*(\mathfrak{C}(\mathcal{B})*\mathfrak{C}(\mathcal{C}))$ for any maps $\mathcal{A}$, $\mathcal{B}$ and $\mathcal{C}$.

It reduces to the tensor product $\mathfrak{C}(\mathcal{A})*\mathfrak{C}(\mathcal{B})=\mathfrak{C}(\mathcal{A})\otimes \mathfrak{C}(\mathcal{B})$ when the connecting system $L$ is trivial. Moreover, using the link product, the action of a map $\mathcal{E}: I\mapsto O$ on an input state $\rho_I$ can be written as $\mathcal{E}(\rho_I)=\mathfrak{C}(\mathcal{E})*\rho_I$, and since the Choi operator of the trace is the identity operator, we can write $\tr_O(\mathcal{E}(\rho))=\id_O*\mathfrak{C}(\mathcal{E})*\rho_I$.

\begin{figure}
    \centering
\begin{tikzpicture}[scale=0.5, transform shape]

\draw (0,0) rectangle node[align=center]{\large{$\mathcal{A}$}} (3,2);

 \draw(0.5,-3)--(0.5,0);  \draw (2.5,-3)--(2.5,0); \draw (0.5,2)--(0.5,4); \draw (2.5,2)--(2.5,4); 
 
\draw (4.5,-3)--(4.5,4);\draw (2,4) rectangle node[align=center]{\large{$\mathcal{D}$}} (5,6);

\draw (-0.5,4) rectangle node[align=center]{\large{$\mathcal{B}$}} (1.5,6); \draw (2.5,6)--(2.5,8);\draw (4.5,6)--(4.5,10); \draw (0.5,6)--(0.5,8);

\draw (-2,8) rectangle node[align=center]{\large{$\mathcal{E}$}} (1,10);  \draw(-1.5,-3)--(-1.5,8);

\draw (-1.5,10)--(-1.5,12); \draw (0.5,10)--(0.5,12); \draw (2.5,10)--(2.5,12); \draw (4.5,10)--(4.5,12);

\draw (1.5,8) rectangle node[align=center]{\large{$\mathcal{F}$}} (3.5,10);
\draw[blue, fill=blue!70!white, fill opacity=0.4] (-2.5,-2.5) rectangle (5.5,11.5); 
\draw[blue, fill=white] (-0.2,-1.7) rectangle (1.2,-0.3); 
\draw[blue, fill=white] (-0.2,2.3) rectangle (3.2,3.7);  \draw[blue, fill=white] (1.8,6.3) rectangle (3.2,7.7); 

\draw (0.5,-1.7)--(0.5,-1.5); \draw (0.5,-0.5)--(0.5,-0.3); \draw (0.5,2.3)--(0.5,2.5); \draw (0.5,3.5)--(0.5,3.7); \draw (2.5,2.3)--(2.5,2.5); \draw (2.5,3.5)--(2.5,3.7); \draw (2.5,6.3)--(2.5,6.5); \draw (2.5,7.5)--(2.5,7.7); 

\node at (7,4.5) {\Large{$\longleftrightarrow$}};

\begin{scope}[shift={(8.5,0)}]

 \draw[blue, fill=blue!70!white, fill opacity=0.4] (0,-2.5)--(0,11.5)--(5,11.5)--(5,9.5)--(3,9.5)--(3,7.5)--(5,7.5)--(5,5.5)--(3,5.5)--(3,3.5)--(5,3.5)--(5,1.5)--(3,1.5)--(3,-0.5)--(5,-0.5)--(5,-2.5)--cycle;

\draw(3.75,-0.5)--(3.75,0); \draw(4.25,-0.5)--(4.25,-0); \node at (4,-0.25) {$\dots$}; \draw(3.75,1)--(3.75,1.5); \draw(4.25,1)--(4.25,1.5); \node at (4,1.25) {$\dots$}; 

\begin{scope}[shift={(0,4)}]

\draw(3.75,-0.5)--(3.75,0); \draw(4.25,-0.5)--(4.25,-0); \node at (4,-0.25) {$\dots$}; \draw(3.75,1)--(3.75,1.5); \draw(4.25,1)--(4.25,1.5); \node at (4,1.25) {$\dots$}; 

\end{scope}
\begin{scope}[shift={(0,8)}]

\draw(3.75,-0.5)--(3.75,0); \draw(4.25,-0.5)--(4.25,-0); \node at (4,-0.25) {$\dots$}; \draw(3.75,1)--(3.75,1.5); \draw(4.25,1)--(4.25,1.5); \node at (4,1.25) {$\dots$}; 

\end{scope}

\begin{scope}[shift={(0,-4)}]
\draw(3.75,1)--(3.75,1.5); \draw(4.25,1)--(4.25,1.5); \node at (4,1.25) {$\dots$}; 
\end{scope}

\begin{scope}[shift={(0,12)}]
\draw(3.75,-0.5)--(3.75,0); \draw(4.25,-0.5)--(4.25,-0); \node at (4,-0.25) {$\dots$}; 

\end{scope}

\end{scope}

\end{tikzpicture}
    \caption{Any quantum circuit with ``empty slots'' (left) can be represented in the form of a quantum comb (right), by stretching and rearranging internal wires \cite{Chiribella2008}.}
    \label{fig:comb}
\end{figure}
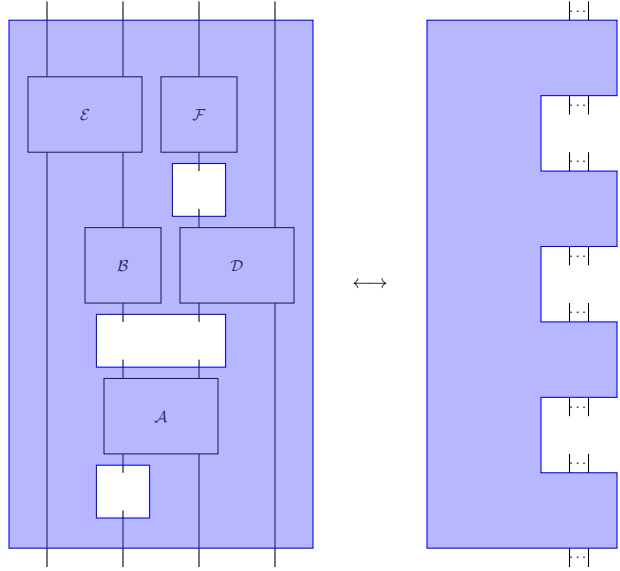

{\bf Quantum combs}
The composition of maps in \cref{fig:composition} is an example of a quantum circuit with no explicit ``empty slots'' for plugging in external operations. 

A quantum circuit with some number $N$ of ``empty slots'' can be described by a general object called a quantum comb \cite{Chiribella2008}, as depicted in \cref{fig:comb}. It is so called because its general structure as seen from the right hand side of the figure,
resembles a comb where the space between adjacent teeth of the comb correspond to the empty slots. The internal structure of a general comb is specified by the following definition.

\begin{definition}[Quantum comb] An $N$ slot quantum comb is a linear CPTP map $\mathcal{C}$ with inputs $S_0,S_2,...,S_{2N}$ and outputs $S_1,S_3,...,S_{2N+1}$ where each $S_j$ can generally be a set of systems. The comb has a well-defined causal order specified by a set $\{\mathcal{E}_i: S_{2i-2}\otimes Q_{i-1}\mapsto S_{2i-1}\otimes Q_{i}\}_{i=1}^{N+1}$ of CPTP maps, where $Q_j$ corresponds to some memory systems (with $Q_0$ and $Q_N$ being trivial i.e., 1D systems). Then the comb is given as (in the Choi representation)
\begin{equation}
    \mathfrak{C}(\mathcal{C})=\mathfrak{C}(\mathcal{E}_{N+1})*\mathfrak{C}(\mathcal{E}_N)*...*\mathfrak{C}(\mathcal{E}_2)*\mathfrak{C}(\mathcal{E}_1),
\end{equation}where each map $\mathcal{E}_i$ is composed with the next map $\mathcal{E}_{i+1}$ through the memory system $Q_{i}$ resulting in a sequential composition of the maps in the order or their indices. The j$^{th}$ slot of an $N$ slot comb can accept a CP map with input $S_{2j-1}$ and output $S_{2j}$ (with $j\in \{1,...,N\}$), where $S_0$ and $S_{2N+1}$ are systems associated with the global past and global future of the causal order of the comb.  
\end{definition}

The above definition and labelling convention is illustrated in \cref{fig:comb_sequence} for a $N=3$ slot comb. Notice that inputs to the comb are outputs to the slots and vice-versa. Moreover, by deforming the wires, we can see that a comb can be equivalently represented as depicted on the right of \cref{fig:comb_sequence} where the slot structure is not so apparent, but it becomes apparent that $\mathcal{C}$ is a CPTP map from $S_0,S_2,...,S_{2N}$ to $S_1,S_3,...,S_{2N+1}$ (since it is a sequential composition of the CPTP maps $\{\mathcal{E}_i\}_{i=1}^{N+1}$). A comb thus provides a general representation of a quantum network which can involve arbitrary preparations, measurements and channels.

{\bf Measurement probabilities}
Generically, the link product describes the composition of any two CP maps. In particular, composing two maps $\mathcal{M}$ and $\mathcal{W}$ where all the in/output systems of the maps are connected through the composition, results in a real number (a map with trivial in and outputs). If $\mathcal{M}$ is an element of a set of maps $\{\mathcal{M}^a\}_a$ associated with different outcomes $a$, with $\sum_a\mathcal{M}^a$ being CPTP, then the probability of obtaining a particular outcome $a=a^*$ when the measurement is composed with $\mathcal{W}$ is given as \cite{Chiribella2009}
\begin{equation}
    \label{eq:comb_prob}
    P(a^*|\mathcal{W},\{\mathcal{M}^a\}_a)=\mathfrak{C}(\mathcal{W})*\mathfrak{C}(\mathcal{M}^{a^*}).
\end{equation}
These maps can themselves be combs that describe two different quantum networks.

Specifically, consider an $N$-slot comb, with a preparation of a state $\rho_0$ on $S_0$ and a measurement performed in each slot, on the systems $S_{2j-1}$ entering the j$^{th}$ slot with a post-measurement state fed back to the comb through the output of the slot, $S_{2j}$. This can be modelled by assigning a quantum instrument $\mathcal{M}_j:=\{\mathcal{M}_j^{x_j}: S_{2j-1}\mapsto S_{2j}\}_{x_j\in X_j}$ to each slot, which is characterised by a set of linear CP maps $\mathcal{M}_j^{x_j}$, one for each outcome $x_j$ of the measurement, taking values in some set $X_j$ where $\sum_{x_j\in X_j}\mathcal{M}_j^{x_j}$ is a linear CPTP map. Then we can apply the probability rule $P(a|\mathcal{W},\{\mathcal{M}^a\}_a)=\mathfrak{C}(\mathcal{W})*\mathfrak{C}(\mathcal{M}^a)$ taking $\mathcal{W}$ to be the map obtained by applying the comb $\mathcal{C}$ on $\rho_0$ and tracing out the future systems $S_{2N+1}$ (i.e., $\mathfrak{C}(\mathcal{W})=\id_{S_{2N+1}}*\mathfrak{C}(\mathcal{C})*\rho_0$), and the measurement to be given by the tensor product of the $N$ measurements applied in each slot (thus each $a$ in this context is a set of measurement $N$ outcomes $(x_1,...,x_N)$). This yields, using the associativity of the link product, where we drop the conditioning on the $N$ measurements (one per slot) for brevity, since it is clear from context.
\begin{align}
\label{eq: comb_prob_product}
\begin{split}
     &P(x_1,...,x_N|\mathcal{C}, \rho_0)\\=&(\id_{S_{2N+1}}*\mathfrak{C}(\mathcal{C})*\rho_0)*(\bigotimes_{j=1}^N \mathfrak{C}(\mathcal{M}_j^{x_j}))\\
     =&\Tr\Big(\bigotimes_{j=1}^N \mathfrak{C}(\mathcal{M}_j^{x_j})*\mathfrak{C}(\mathcal{C})*\rho_0\Big)   
\end{split}
\end{align}
It can be shown that for all choices of instruments plugged in the slots and all states $\rho_0$, $P(x_1,...,x_N|\mathcal{C}, \rho_0)$ as defined here is a valid normalised distribution \cite{Chiribella2008, Chiribella2013} i.e., $\sum_{x_1\in X_1,...,x_N\in X_N}P(x_1,...,x_N|\mathcal{C}, \rho_0)=1$. In the above, the trace is over the systems $S_{2N+1}$ in the global future.

\subsubsection{Time labelling}
A comb as depicted on the left hand side of \cref{fig:comb_sequence} implies a direction of ``time'', in terms of the flow of information it describes, in this case: bottom to top. Notice that we can plug in a physical operation in the first slot, with inputs to the operation being $S_1$ and outputs being $S_2$ (and similarly for subsequent slots). 
On the right side of \cref{fig:comb_sequence} however, having an external channel from $S_1$ to $S_2$ will lead to a flow of information from the top to the bottom of the page, while the internal operations $\mathcal{E}_i$ lead to a flow from bottom to top. To bring our discussion of combs closer to the multi-time formalism, we will add time labels to the in/output systems and distinguish between the two sides of \cref{fig:comb_sequence} in our notation, when referring to time labelled objects.

\begin{definition}[Time labelled quantum comb]
    An $N$-slot time-labelled quantum comb $\mathcal{C}^{\mathbf{t}}$ is an $N$-slot quantum comb $\mathcal{C}$ together with a time label $t_j$ assigned to each set of in/output systems $S_j$ (for $j\in \{0,...,2N+1\}$) such that $t_i<t_j$ whenever $i<j$. We will always depict time-labelled combs in the form given on the left hand side of \cref{fig:comb_sequence}, where time flows from the bottom to the top of the page. 
\end{definition}

Notice that in particular, a single CPTP map is a zero slot comb with input systems $S_0$ and output systems $S_1$, where its time labelled version would be associated with in and output times $t_0$ and $t_1>t_0$.

\subsubsection{Including P-CTCs}

We now have all the ingredients to define the most general P-CTC assisted objects that will be considered in this paper. We define the time unlabelled and time labelled versions separately, starting with the former below.

\begin{definition}[P-CTC assisted comb]
\label{definition:ctc_comb}
A P-CTC assisted comb is a linear CP map $\mathcal{C}_{\text{CTC}}$ with inputs $S'_0,S_2,...,S_{2N}$ and outputs $S_1,S_3,...,S'_{2N+1}$ where each $S_j$, $S'_0$ and $S'_{2N+1}$ can generally be a set of systems. It is obtained from a comb $\mathcal{C}$ with inputs $S_0,S_2,...,S_{2N}$ and outputs $S_1,S_3,...,S_{2N+1}$, where $S_0:=S'_0\otimes A$ and $S_{2N+1}:=S'_{2N+1}\otimes A$, by including P-CTCs connecting the ancillaries $A$ (from the output to input of the comb). Noting that the comb is itself a CPTP map, and denoting $\ket{\Phi^+}_{AA'}$ as the maximally entangled state on the set of systems $A$ and isomorphic systems $A'$, this is given as
\begin{equation}
\label{eq: CTC_comb}
 \mathcal{C}_{\text{CTC}}=\bra{\Phi^+}_{AA'}\mathcal{C}\otimes \id_{A'} (\ket{\Phi^+}\bra{\Phi^+}_{AA'}) \ket{\Phi^+}_{AA'}.
\end{equation}

\end{definition}

Notice that in the above equation $\mathcal{C}\otimes \id_{A'}$ is a map with inputs $A',A,S'_0,S_2,...,S_{2N}$ to outputs $A',A,S'_1,S_3,...,S'_{2N+1}$. Upon acting on the state $(\ket{\Phi^+}\bra{\Phi^+}_{AA'})$ and with the final post-selection on the same state, this becomes a linear CP (but not necessarily TP) map from inputs $S'_0,S_2,...,S_{2N}$ to outputs $S_1,S_3,...,S'_{2N+1}$, which is the desired map $ \mathcal{C}_{\text{CTC}}$.

{\bf General form} Analogous to the case of quantum combs, any P-CTC assisted circuit with ``empty slots'' can be expressed in the form of a P-CTC assisted comb, by stretching and rearranging wires. This is illustrated in \cref{fig: CTC_comb}. More explicitly, by a P-CTC assisted circuit with ``empty slots'' we mean any standard quantum circuit (flowing from bottom to top of the page) which can include additional P-CTCs from outputs of later maps (higher on the page) to inputs of earlier maps (lower on the page). Such a P-CTC assisted circuit includes compositions of the maps involved in forming the original circuit along with compositions coming from the CTCs (both connect an output to an input). It can be shown that the description of such a circuit is independent of the order in which the all the compositions are performed \cite{vilasini2022embedding}, and in particular, we choose to perform the original circuit compositions first (at which stage we can equivalently transform to a regular quantum comb as in \cref{fig:comb}) and then perform all the CTC compositions. By stretching the CTC wires in both directions, we can immediately obtain the form on the right hand side of \cref{fig: CTC_comb} where all the CTCs flow from global future to global past. This also follows from the associativity of the link product, by modelling the Bell state preparation and measurement involved in the CTC as CP maps.

Finally, we can also represent a P-CTC assisted comb equivalently as a P-CTC assisted map through a circuit isomorphism as shown in \cref{fig:CTCcomb_sequence}. Notice however that plugging in maps in the slot of a P-CTC assisted comb then require additional P-CTCs in order to be replicated (in an operationally equivalent way) in the representation of the comb as a CPTP map.

{\bf Measurement probabilities} 
Similarly to the discussion of the general probability rule for combs, consider a CP map $\mathcal{W}$ and another CP map $\mathcal{M}^a$ which is an element of a set of maps for different outcomes $a$. Generally, if $\mathcal{W}$ is not also trace preserving, then the expression of \cref{eq:comb_prob} does not yield normalised probabilities i.e., $\sum_a \mathfrak{C}(\mathcal{W})*\mathfrak{C}(\mathcal{M}^a)\neq 1$, even when $\sum_a\mathcal{M}^a$ is CPTP. However, the following expression for the probability of a particular outcome $a=a^*$ in this scenario does, by construction, define a normalised probability distribution, due to the presence of the denominator.
\begin{equation}
\label{eq:PCTC_comb_prob_gen}
    P(a^*|\mathcal{W},\{\mathcal{M}^a\}_a)=\frac{\mathfrak{C}(\mathcal{W})*\mathfrak{C}(\mathcal{M}^{a^*})}{\sum_a \mathfrak{C}(\mathcal{W})*\mathfrak{C}(\mathcal{M}^{a^*})}.
\end{equation}

In particular, if $\mathcal{W}$ is obtained by applying a P-CTC assisted map $\mathcal{C}_{CTC}:\mathcal{L}(\mathcal{H}_S)\mapsto \mathcal{L}(\mathcal{H}_S)$ to an initial state $\rho_S$ i.e., $\mathfrak{C}(\mathcal{W})=\mathfrak{C}(\mathcal{C}_{CTC})*\rho_S=\mathcal{C}_{CTC}(\rho_S)$ and $\{\mathcal{M}^a\}_a$ is a POVM on $S$, then the above expression gives 

\begin{align}
    \begin{split}
   &P(a^*|\mathcal{W},\{\mathcal{M}^a\}_a)\\&=  \frac{ \tr(\mathcal{C}_{CTC}(\rho_S)\mathfrak{C}(\mathcal{M}^{a^*})^T)}{\sum_a \tr(\mathcal{C}_{CTC}(\rho_S)\mathfrak{C}(\mathcal{M}^{a})^T)}  \\&=\frac{\tr(\mathcal{C}_{CTC}(\rho_S)\mathcal{M}^{a^*})}{\sum_a\tr(\mathcal{C}_{CTC}(\rho_S)\mathcal{M}^a)}\\&=
   \frac{\tr(\mathcal{C}_{CTC}(\rho_S)\mathcal{M}^{a^*})}{\tr(\mathcal{C}_{CTC}(\rho_S))}=P(a^*|\mathcal{C}_{CTC},\rho_S,\{\mathcal{M}^a\}_a),
    \end{split}
\end{align}
where we used the properties of the link product discussed in \cref{sec: combs} and the POVM condition that $\sum_a\mathcal{M}^a=\id_S$. This is precisely the expression we obtained for the special case of CTC assisted maps in \cref{eq: PCTC_map_prob}.

More generally, we can consider an $N$-slot P-CTC assisted comb and plug in a quantum instrument $\mathcal{M}_j=\{\mathcal{M}_j^{x_j}\}_{x_j\in X_j}$ in each slot, along with an initial state $\rho_0$ on $S'_0$ and tracing out the global future $S'_{2N+1}$. The resulting measurement probabilities are given as follows, using \cref{eq:PCTC_comb_prob_gen} with $\mathfrak{C}(\mathcal{W})=\id_{S'_{2N+1}}*\mathfrak{C}(\mathcal{C}_{CTC})*\rho_0$, and $\mathfrak{C}(\mathcal{M}^a)$ replaced with $\bigotimes_{j=1}^N \mathfrak{C}(\mathcal{M}_j^{x_j})$ and invoking the associativity of the link product. In the following, we drop the conditioning on the $N$ measurements, as we did in \cref{eq: comb_prob_product}, for simplicity and since it is clear from context.

\begin{align}
\begin{split}
     \label{eq: pctc_comb_prob}
P(x_1,...,x_N|\mathcal{C}_{CTC},\rho_0)&= \frac{(\id_{S'_{2N+1}}*\mathfrak{C}(\mathcal{C}_{CTC})*\rho_0)*(\bigotimes_{j=1}^N \mathfrak{C}(\mathcal{M}_j^{x_j}))}{\sum_{x_1,...,x_N}(\id_{S'_{2N+1}}*C_{CTC}*\rho_0)*(\bigotimes_{j=1}^N \mathfrak{C}(\mathcal{M}_j^{x_j}))}\\
&=\frac{\Tr\Big(\bigotimes_{j=1}^N \mathfrak{C}(\mathcal{M}_j^{x_j})*\mathfrak{C}(\mathcal{C}_{CTC})*\rho_0\Big)}{\sum_{x_1,...,x_N}\Tr\Big(\bigotimes_{j=1}^N \mathfrak{C}(\mathcal{M}_j^{x_j})*\mathfrak{C}(\mathcal{C}_{CTC})*\rho_0\Big)} 
\end{split}
\end{align}

Notice that this is identical to \cref{eq: comb_prob_product} whenever the denominator $\sum_{x_1,...,x_N}\Tr\Big(\bigotimes_{j=1}^N \mathfrak{C}(\mathcal{M}_j^{x_j})*\mathfrak{C}(\mathcal{C}_{CTC})*\rho_0\Big)=1$ (with $\mathfrak{C}(\mathcal{C}_{CTC})$ taking the place of $\mathfrak{C}(\mathcal{C})$), and therefore recovers \cref{eq: comb_prob_product} when there are no P-CTCs. This is generally not the case when we have non-trivial P-CTCs, which lead to an unnornmalised numerator due to the post-selection involved. Indeed, recalling from \cref{eq: CTC_comb} that $\mathcal{C}_{CTC}$ is obtained by assisting a CPTP map $\mathcal{C}$ representing a comb through P-CTCs, we recognise the numerator of \cref{eq: pctc_comb_prob} as the joint probability of the outcomes $x_1,...,x_N$ along with the outcome $\Phi^+$ in the post-selecting measurement of the P-CTC, while the denominator is the probability of the latter alone. This makes $P(x_1,...,x_N|\mathcal{C}_{CTC},\rho_0)$ a conditional probability, where we condition on the success of the P-CTC's post-selection.

An important observation to make at this point is that replacing $\mathcal{C}$ with $k\mathcal{C}$ for any constant $k$, changes $\mathcal{C}_{\text{CTC}}$ by the same factor of $k$, but leaves the probabilities according to the above expressions, unaffected as the factor would appear both in the numerator and denominator and cancel out.

We can also consider non-product operations plugged in the slots of a P-CTC assisted comb $\mathcal{C}_{\text{CTC}}$, such as another P-CTC assisted comb $\mathcal{D}_{\text{CTC}}$ of the right type. Again it follows from the arguments for regular combs (see \cref{fig:CTC_comb_nonprod})) and the order independence of compositions involving P-CTCs noted before \cite{vilasini2022embedding}, that the composition of any such $\mathcal{C}_{\text{CTC}}$ and $\mathcal{D}_{\text{CTC}}$ can be equivalently seen as a composition of a new P-CTC assisted comb $\mathcal{C}'_{\text{CTC}}$ with a product operation (or memoriless quantum comb) $\mathcal{D}_{\otimes}$ i.e., all the P-CTCs of $\mathcal{D}_{\text{CTC}}$ can be absorbed into $\mathcal{C}'_{\text{CTC}}$. This is shown in \cref{fig:CTC_comb_nonprod}.

In other words, when computing probabilities of non-product measurements in a P-CTC assisted comb $\mathcal{C}_{\text{CTC}}$ (where the measurement itself may involve P-CTCs), we can always transform the computation to an equivalent one involving a new P-CTC assisted comb $\mathcal{C}_{\text{CTC}}$ and product measurements where the measurements no longer contain any P-CTCs.

\begin{figure*}[ht!]
 \centering
\begin{tikzpicture}[scale=0.8, transform shape]

\draw[blue] (-3,6.75)--(-3,8); \draw[blue] (-3,0.25)--(-3,-1);
\draw[blue] (-4.5,-1)--(-4.5,8); \draw[blue] (-4.5,8) arc (180:0:0.75); \draw[blue] (-4.5,-1) arc (180:360:0.75); \node[blue] at (-3.25,7.5) {$A$}; \node[blue] at (-3.25,-0.5) {$A$};

\draw[red] (1,5.25)--(1,6.5); \draw[red] (1,1.75)--(1,0.5); \draw[red] (2.5,0.5)--(2.5,6.5); \draw[red] (2.5,6.5) arc (0:180:0.75); \draw[red] (1,0.5) arc (180:360:0.75); \node[red] at (1.25,6) {$B$}; \node[red] at (1.25,1) {$B$};

 \draw[blue, fill=blue!70!white, fill opacity=0.4] (0,0)--(0,1)--(-2,1)--(-2,3)--(0,3)--(0,4)--(-2,4)--(-2,6)--(0,6)--(0,7)--(-4,7)--(-4,0)--cycle;
 \draw[red, fill=red!70!white, fill opacity=0.4] (1.7,1.5)--(-1.7,1.5)--(-1.7,2.5)--(0.5,2.5)--(0.5,4.5)--(-1.7,4.5)--(-1.7,5.5)--(1.7,5.5)--cycle;

\draw(-0.75,-0.5)--(-0.75,0)  node[midway, anchor=west, xshift=2mm] {$S_0$}; \draw(-1.25,-0.5)--(-1.25,-0); \node at (-1,-0.25) {$\dots$}; \draw(-1.25,0)--(-1.25,0.25);\draw(-0.75,0)--(-0.75,0.25); \draw(-0.75,0.75)--(-0.75,1.75) node[midway, anchor=west, xshift=2mm] {$S_1$}; \draw(-1.25,0.75)--(-1.25,1.75); \node at (-1,1.25) {$\dots$};

\begin{scope}[shift={(0,3)}]
 \draw(-0.75,-0.75)--(-0.75,0.25)  node[midway, anchor=west, xshift=2mm] {$S_2$}; \draw(-1.25,-0.75)--(-1.25,0.25); \node at (-1,-0.25) {$\dots$}; \draw(-0.75,0.75)--(-0.75,1.75) node[midway, anchor=west, xshift=2mm] {$S_3$}; \draw(-1.25,0.75)--(-1.25,1.75); \node at (-1,1.25) {$\dots$};

\end{scope}

\begin{scope}[shift={(0,6)}]
 \draw(-0.75,-0.75)--(-0.75,0.25)  node[midway, anchor=west, xshift=2mm] {$S_4$}; \draw(-1.25,-0.75)--(-1.25,0.25); \node at (-1,-0.25) {$\dots$}; \draw(-0.75,1)--(-0.75,1.5) node[midway, anchor=west, xshift=2mm] {$S_5$}; \draw(-1.25,1)--(-1.25,1.5); \node at (-1,1.25) {$\dots$};  \draw(-1.25,0.75)--(-1.25,1); \draw(-0.75,0.75)--(-0.75,1);
 
\end{scope}

\draw[blue] (-3.5,0.25) rectangle node[align=center]{$\mathcal{E}_1$}(-0.5,0.75);
\draw[blue] (-3.5,3.25) rectangle node[align=center]{$\mathcal{E}_2$} (-0.5,3.75);
\draw[blue] (-3.5,6.25) rectangle node[align=center]{$\mathcal{E}_3$} (-0.5,6.75); \draw[blue] (-3,0.75)-- node[anchor=east] {$Q_1$}(-3,3.25); \draw[blue] (-3,3.75)-- node[anchor=east] {$Q_2$} (-3,6.25);

\draw[red] (-1.5,1.75) rectangle node[align=center]{$\mathcal{F}_1$} (1.25,2.25);
\draw[red] (-1.5,4.75) rectangle node[align=center]{$\mathcal{F}_2$} (1.25,5.25); \draw[red] (1,2.25)-- node[anchor=west] {$R_1$}(1,4.75);

\node at (-1,-2.5) {\huge{\textcolor{blue}{$\mathcal{C}_{\mathrm{CTC}}$}*\textcolor{red}{$\mathcal{D}_{\mathrm{CTC}}$}}};

\node at (4.25,3.5) {\huge{$\cong$}};

\begin{scope}[shift={(11,0)}]
\draw[blue] (-3,6.75)--(-3,8); \draw[blue] (-3,0.25)--(-3,-1);
\draw[blue] (-4.5,-1)--(-4.5,8); \draw[blue] (-4.5,8) arc (180:0:0.75); \draw[blue] (-4.5,-1) arc (180:360:0.75);
\node[blue] at (-3.5,8.25) {$A$}; \node[blue] at (-3.5,-1.25) {$A$};

\draw (1,5.25)--(1,6.25); \node at (1.25,5.75) {$B$};
\draw (1,1.75)--(1,0.75); \node at (1.25,1.25) {$B$};
\draw[blue] (1,6.25)--(1,6.75); \draw[blue] (1,0.75)--(1,0.25); \draw[blue] (1,6.75) to[out=90,in=270](-2.5,8);
\draw[blue] (1,0.25) to[out=270,in=90](-2.5,-1); \draw[blue] (-2.5,8)--(-2.5,9); \draw[blue] (-2.5,-1)--(-2.5,-2); \draw[blue] (-5,-2)--(-5,9); \draw[blue] (-5,-2) arc (180:360:1.25); \draw[blue] (-5,9) arc (180:0:1.25); \node[blue] at (-2.75,8.25) {$B$}; \node[blue] at (-2.75,-1.25) {$B$};

\draw[dashed, blue] (-0.5,0.75)--(1.25,0.75)--(1.25,0.25)--(-0.5,0.25); \node[blue] at (0,0.5) {$\mathcal{E}'_1$};
\draw[dashed, blue] (-0.5,3.75)--(1.25,3.75)--(1.25,3.25)--(-0.5,3.25); \node[blue] at (0,3.5) {$\mathcal{E}'_2$};
\draw[dashed, blue] (-0.5,6.75)--(1.25,6.75)--(1.25,6.25)--(-0.5,6.25); \node[blue] at (0,6.5) {$\mathcal{E}'_3$};

\node at (0,-2.5) {\huge{\textcolor{blue}{$\mathcal{C}'_{\mathrm{CTC}}$}*\textcolor{red}{$\mathcal{D}_{\otimes}$}}};

     \draw[blue, fill=blue!70!white, fill opacity=0.4] (1.7,-1)--(1.7,1)--(-2,1)--(-2,3)--(1.7,3)--(1.7,4)--(-2,4)--(-2,6)--(1.7,6)--(1.7,8)--(-4,8)--(-4,-1)--cycle;

     \draw(-0.75,-1)--(-0.75,-1.5)  node[midway, anchor=west, xshift=2mm] {$S_0$}; \draw(-1.25,-1)--(-1.25,-1.5) ;
      \draw(-0.75,-0.25)--(-0.75,0.25);
      \draw(-1.25,-0.25)--(-1.25,0.25); 
      \draw(-0.75,-0.45)--(-0.75,-1);
      \draw(-1.25,-0.45)--(-1.25,-1); 
      \node at (-1,-1.25) {$\dots$};

     \draw(-0.75,0.75)--(-0.75,1.75) node[midway, anchor=west, xshift=2mm] {$S_1$}; \draw(-1.25,0.75)--(-1.25,1.75); \node at (-1,1.25) {$\dots$};

\begin{scope}[shift={(0,3)}]
 \draw(-0.75,-0.75)--(-0.75,0.25)  node[midway, anchor=west, xshift=2mm] {$S_2$}; \draw(-1.25,-0.75)--(-1.25,0.25); \node at (-1,-0.25) {$\dots$}; \draw(-0.75,0.75)--(-0.75,1.75) node[midway, anchor=west, xshift=2mm] {$S_3$}; \draw(-1.25,0.75)--(-1.25,1.75); \node at (-1,1.25) {$\dots$};

\end{scope}

\begin{scope}[shift={(0,6)}]
 \draw(-0.75,-0.75)--(-0.75,0.25)  node[midway, anchor=west, xshift=2mm] {$S_4$}; \draw(-1.25,-0.75)--(-1.25,0.25); \node at (-1,-0.25) {$\dots$};

\draw(-0.75,2)--(-0.75,2.5) node[midway, anchor=west, xshift=2mm] {$S_5$};  \draw(-1.25,2)--(-1.25,2.5);\node at (-1,2.25) {$\dots$};
 \draw(-0.75,0.75)--(-0.75,1.25);  \draw(-1.25,0.75)--(-1.25,1.25); \draw(-0.75,1.45)--(-0.75,2);  \draw(-1.25,1.45)--(-1.25,2); 
 
\end{scope}

\draw[blue] (-3.5,0.25) rectangle node[align=center]{$\mathcal{E}_1$}(-0.5,0.75);
\draw[blue] (-3.5,3.25) rectangle node[align=center]{$\mathcal{E}_2$} (-0.5,3.75);
\draw[blue] (-3.5,6.25) rectangle node[align=center]{$\mathcal{E}_3$} (-0.5,6.75); \draw[blue] (-3,0.75)-- node[anchor=east] {$Q_1$}(-3,3.25); \draw[blue] (-3,3.75)-- node[anchor=east] {$Q_2$} (-3,6.25);

\draw[red, fill=red!70!white, fill opacity=0.4, text opacity=1] (-1.5,1.75) rectangle node[align=center]{$\mathcal{F}_1$} (1.25,2.25);
\draw[red, fill=red!70!white, fill opacity=0.4, text opacity=1]  (-1.5,4.75) rectangle node[align=center]{$\mathcal{F}_2$} (1.25,5.25); \draw (1,2.25)--(1,3.25);\draw (1,3.75)--(1,4.75); \draw[blue] (1,3.25)--(1,3.75);  \node at (1.5,2.65) {$R_1$}; \node at (1.5,4.35) {$R_1$};
\end{scope}
\end{tikzpicture}
    \caption{A composition of two P-CTC assisted combs $\mathcal{C}_{\mathrm{CTC}}$ and $\mathcal{D}_{\mathrm{CTC}}$ (left) can be equivalently viewed as a composition of a new P-CTC assisted comb $\mathcal{C}'_{\mathrm{CTC}}$ with a comb $\mathcal{D}_{\otimes}$ (which has no CTCs), where $\mathcal{D}_{\otimes}$ factorises into a product of the internal CPTP maps $\mathcal{F}_i$ of $\mathcal{D}_{\mathrm{CTC}}$ (right).   }
    \label{fig:CTC_comb_nonprod}
\end{figure*}
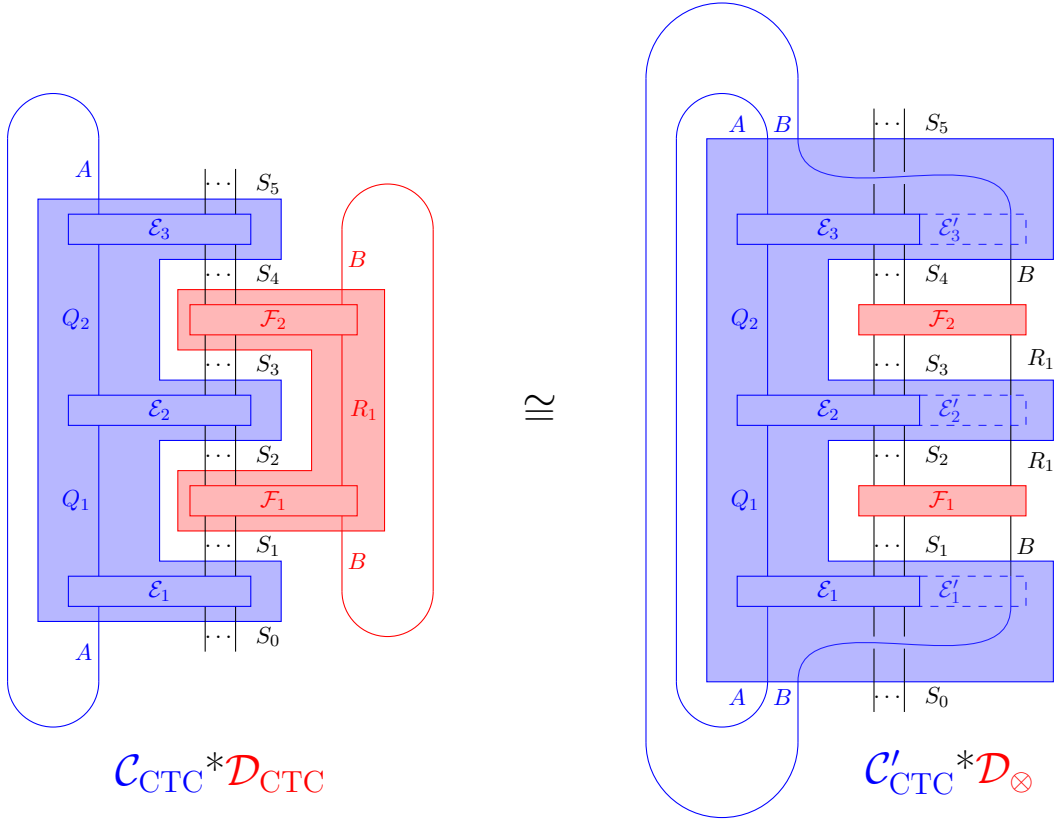

{\bf Time labelling}

\begin{definition}[Time-labelled P-CTC assisted comb]
\label{def: timelab_PCTC_comb}
      An $N$-slot time-labelled P-CTC assisted comb $\mathcal{C}^{\mathbf{t}}_{\text{CTC}}$ is an $N$-slot P-CTC assisted comb $\mathcal{C}_{\text{CTC}}$ together with a time label $t_j$ assigned to each set of in/output systems $S'_0,S_1,S_2,...S_N,S'_{2N+1}$ (according to the index $j$ of the set), such that $t_i<t_j$ whenever $i<j$. 
\end{definition}

\begin{remark}
    Some of the sets of in and output systems $S'_0,S_1,S_2,...S_N,S'_{2N+1}$ can be trivial, such that time labelled P-CTC assisted comb can also model situations where we have two subsequent non-trivial inputs at some times $t$ and $t'>t$ and then a non-trivial output at time $t''>t'$, by setting $t=t_{2j}$, $t'=t_{2j+2}$ for the two inputs, introducing a trivial output at $t_{2j+1}$ and the non-trivial output at $t''=t_{2j+3}$.
\end{remark}

%% file: 6-Appendix/6-1-Connection.tex
\section{Connection between the two frameworks: details and proofs}\label{sec:Appendix-Connection}

\subsection{Constructing a pure 2TO from a P-CTC assisted unitary}\label{subsec:DetailsExplConstr}

Here we show that assisting the unitary $U_i$ of \cref{eq: Ui} on $S$ and $Q$ with a P-CTC on the qubit $Q$ yields an operator which is operationally equivalent to $C_i$ of \cref{eq:Cond1} on the system $S$, as depicted in the circuit of \cref{fig:Ci}. We have the following defining equations, which we repeat here for convenience.
\begin{align}
    \begin{split}
         U_i&=(W_i\otimes \id_Q)(C_{V_i})\\
           C_{V_i}&= \pure{i}_S\otimes (V_i)_Q + \sum_{j\neq i}\pure{j}_S\otimes\id_Q\\
             V_i &= \begin{pmatrix}
    a_i & -\sqrt{1-a_i^2}\\
   \sqrt{1-a_i^2} & a_i
    \end{pmatrix}\\
    W_i\ket{i}_S&=\ket{\psi_i}_S\\
    C_i \ket{i}_S&=a_i\ket{\psi_i}_S
    \end{split}
\end{align}

The controlled operation $C_{V_i}$ on $S$ ($d$-dimensional) and $Q$ (2-dimensional) in matrix form is composed of $d$ $2\times 2$ blocks on the diagonal. The i-th block is a $V_i$, while the other blocks are $2\times 2$ identity matrices. For instance, if $i=0$, then the matrix has the form
\begin{equation}
\label{eq: CV0}
    C_{V_0} = 
    \begin{pmatrix}
        V_0 & 0 \\
        0 & \id_{2\left(d-1\right)}
    \end{pmatrix}
\end{equation}
Assisting $U_i$ through a P-CTC on $Q$ yields the following operator (c.f. \cref{eq: partialtrace_CTC})
\begin{equation}
\label{eq: trUi}
    \tr_Q(U_i)=\tr_Q \Bigg((W_i\otimes \id_Q)(C_{V_i})\Bigg)= (W_i)_S\Big(\tr_Q(C_{V_i})\Big)
\end{equation}

Consider $\tr_Q(C_{V_i})$ for $i=0$. Using \cref{eq: CV0}, this reduces to a $d\times d$ matrix (recalling that $V_0$ acts on $Q$):
\begin{equation}
    \tr_{Q}\left(C_{V_0}\right) = \begin{pmatrix}
        \tr V_0 & 0 & \cdots & 0  \\
        0 & 2 & \cdots & 0  \\
        \vdots & \vdots & \ddots & \vdots \\
        0 & 0 & \cdots & 2
    \end{pmatrix} = 
    2\begin{pmatrix}
        a_0 & 0 \\
        0 & \id_{d-1}
    \end{pmatrix}
\end{equation}

Then it is clear that $ \tr_{Q}\left(C_{V_0}\right)$ when applied to an input state $\ket{0}_S$ yields $2a_0\ket{0}_S$. By a similar argument, it follows that for every $i\in\{0,...,d-1\}$, $ \tr_{Q}\left(C_{V_i}\right)\ket{i}_S=2a_i\ket{i}_S$.

Together with \cref{eq: trUi} , this implies that the P-CTC assisted unitary $\tr_Q(U_i)$ of interest acts as follows for every $i$
\begin{equation}
    \tr_Q(U_i)\ket{i}_S=2a_i\ket{\psi_i}_S,
\end{equation}
which as it only differs by an overall constant 2, is operationally equivalent to $C_i$, as desired.

\subsection{Existence of desired basis: proof of \cref{prop: existencebasis}}\label{subsec:ExistenceBasis}

\ExistenceBasis*
\begin{proof}
Without loss of generality, we can take $r=1$, $a_i\in\mathbb{R}$ and $\sum_i \left|a_i\right|^2 = 1$. This follows from the arguments in \cref{subsubsec:2TO-PCTC} along with noting that the overall constant $r$ is irrelevant for operational equivalence, which means that transforming each coefficient as $a_i \mapsto \frac{a_i}{\sum_i \left|a_i\right|^2}$ to ensure that the sum of the squares of the new coefficients equals 1, will also preserve operational equivalence.

Now, since $\sum_i \left|a_i\right|^2 = 1$, the average $\left|a_i\right|^2$ is equal to $\frac{1}{d}$. If all $a_i$ are equal to $\frac{1}{\sqrt{d}}$, that would complete the proof as it could be absorbed into the overall constant $r'$, so we suppose this is not the case.  Then, in particular, there exist some $a_k$ and $a_j$ such that $\left|a_k\right|^2<\frac{1}{d}$ and $\left|a_j\right|^2>\frac{1}{d}$ respectively. Let $a_{\text{max}}$ be such that  $\left|a_{\text{max}}\right|^2=\text{max}\left(\{\left|a_i\right|^2\}_{i=0}^{d-1}\right)$ and let $\ket{\text{max}}$ be the corresponding state of the computational basis. Let $a_{\text{min}}$ be such that $\left|a_{\text{min}}\right|^2=\text{min}\left(\{\left|a_i\right|^2\}_{i=0}^{d-1}\right)$ and let $\ket{\text{min}}$ be the corresponding state of the computational basis. Then $C$ writes
\begin{equation}
    C = a_{\text{max}}\ket{\psi_{\text{max}}}\bra{\text{max}} + a_{\text{min}}\ket{\psi_{\text{min}}}\bra{\text{min}} + \sum_{j\neq \text{max}, \text{min}} a_j\ket{\psi_j}\bra{j}
    \label{eq:Proof_C1}
\end{equation}
Now the idea is to rotate the sub-basis $\{\ket{\text{max}}, \ket{\text{min}}\}$ such that, after rotation, $\left|a_{\text{max}}\right|^2$ becomes equal to $\frac{1}{d}$. Because the rotation is an orthogonal operation, the orthonormality of the transformed basis is guaranteed. Consider the following rotation in the plane defined by $\{\ket{\text{max}}, \ket{\text{min}}\}$:
\begin{align}
\label{eq: new basis}
\begin{split}   
    \ket{\theta} &= \cos{\theta}\ket{\text{max}} + \sin{\theta}\ket{\text{min}}\\
    \ket{\theta^{\perp}} &= \sin{\theta}\ket{\text{max}} - \cos{\theta}\ket{\text{min}}\\
    \ket{j} &= \ket{j}, \quad j\neq \text{max}, \text{min}.
    \end{split}
\end{align}

The result of \cref{lemma: ExistenceBasis_proof} (stated and proven immediately after the current proof) for expressing $C$ in this new basis will be central in the rest of this proof.

The third property of \cref{lemma: ExistenceBasis_proof} together with the construction of $a_{\text{max}}$ and $a_{\text{min}}$ described earlier in this proof imply
\begin{align}
\begin{split}
    \left|a_{\theta}\left(0\right)\right|^2=\left|a_{\text{max}}\right|^2 &> \frac{1}{d} \\
    \left|a_{\theta}\left(\frac{\pi}{2}\right)\right|^2=\left|a_{\text{min}}\right|^2 &< \frac{1}{d}
    \end{split}
    \label{eq:Proof_Cond2}
\end{align}

Hence, by the intermediate value theorem, there exists $\theta^*\in[0,\frac{\pi}{2}]$ such that $\left|a_{\theta}\left(\theta^*\right)\right|^2 = \frac{1}{d}$. This fixes the first basis element $\ket{\phi_0} = \ket{\theta^*}$ and $C$ becomes
\begin{equation}
    C = \frac{1}{\sqrt{d}}\ket{\psi_{\theta^*}}\bra{\theta^*} + a_{\theta^{\perp}}\left(\theta^*\right)\ket{\psi_{\theta^{*\perp}}}\bra{\theta^{*\perp}} + \sum_{j\neq \text{max}, \text{min}} a_j\ket{\psi_j}\bra{j}
    \label{eq:Proof_C3}
\end{equation}
The above procedure can then be repeated $d-2$ times to fix all the coefficients to $\frac{1}{\sqrt{d}}$. It is guaranteed that a basis element will not be rotated twice in the procedure because its coefficient has been set to the mean value of $\left|a_i\right|^2=\frac{1}{d}$. Hence at each step, a different pair of basis elements is rotated in the plane they define and one basis element is fixed appropriately. The last step of the iteration fixes the last coefficients to $\frac{1}{\sqrt{d}}$ simultaneously.
\end{proof}

\begin{restatable}[]{lemma}{ExistenceBasisProof}
\label{lemma: ExistenceBasis_proof}
A pure 2TO $C$ of the form of \cref{eq:Proof_C1}, when expressed in the basis $\{\ket{\theta},\ket{\theta^{\perp}}, \{\ket{j}\}_{j\neq \text{min},\text{max}}\}$ of \cref{eq: new basis} takes the following form.
\begin{equation}
    C = a_{\theta}\left(\theta\right)\ket{\psi_{\theta}}\bra{\theta} + a_{\theta^{\perp}}\left(\theta\right)\ket{\psi_{\theta^{\perp}}}\bra{\theta^{\perp}} + \sum_{j\neq \text{max}, \text{min}} a_j\ket{\psi_j}\bra{j},
    \label{eq:Proof_C2}
\end{equation}
with the following properties
\begin{enumerate}
    \item The coefficients $a_{\theta}\left(\theta\right)$ and $a_{\theta^{\perp}}\left(\theta\right)$ are both continuous functions in $\theta$. 
    \item  The states $\ket{\psi_{\theta}}$, $\ket{\psi_{\theta^{\perp}}}$ and all $\ket{\psi_j}$ are normalized and can depend on $\theta$.
    \item The coefficients satisfy $\left|a_{\theta}\left(0\right)\right|^2=\left|a_{\text{max}}\right|^2$ and $  \left|a_{\theta}\left(\frac{\pi}{2}\right)\right|^2=\left|a_{\text{min}}\right|^2$. 

\end{enumerate}
\end{restatable}
\begin{proof}
Recall from \cref{eq:Proof_C1} that $C$ can be written as
\begin{equation}
    C = a_{\text{max}}\ket{\psi_{\text{max}}}\bra{\text{max}} + a_{\text{min}}\ket{\psi_{\text{min}}}\bra{\text{min}} + \sum_{j\neq \text{max}, \text{min}} a_j\ket{\psi_j}\bra{j}
    \label{eq:ProofAn_C1}
\end{equation}
Then consider the rotation in the plane defined by $\{\ket{\text{max}}, \ket{\text{min}}\}$ as given in \cref{eq: new basis}, which is
\begin{align}
\begin{split}
    \ket{\theta} &= \cos{\theta}\ket{\text{max}} + \sin{\theta}\ket{\text{min}}\\
    \ket{\theta^{\perp}} &= \sin{\theta}\ket{\text{max}} - \cos{\theta}\ket{\text{min}}\\
    \ket{j} &= \ket{j}, \quad j\neq \text{max}, \text{min}.
    \end{split}
\end{align}
which can be inverted, giving:
\begin{align}
\label{eq:newbasis_2}
\begin{split}
    \ket{\text{max}} &= \cos{\theta}\ket{\theta} + \sin{\theta}\ket{\theta^{\perp}}\\ 
    \ket{\text{min}} &= \sin{\theta}\ket{\theta} - \cos{\theta}\ket{\theta^{\perp}}\\ 
    \ket{j} &= \ket{j}, \quad j\neq \text{max}, \text{min} 
    \end{split}
\end{align}

The state $\ket{\psi_{\text{max}}}$ expanded in the basis $\{\ket{\theta},\ket{\theta^{\perp}}, \{\ket{j}\}_{j\neq \text{min},\text{max}}\}$ writes
\begin{align}
    \ket{\psi_{\text{max}}}&= \alpha_{\text{max}}\ket{\text{max}}+\alpha_{\text{min}}\ket{\text{min}} + \sum_{j\neq \text{min}, \text{max}}\alpha_j\ket{j}\nonumber\\
    &= \left(\alpha_{\text{max}}\cos{\theta}+\alpha_{\text{min}}\sin{\theta}\right)\ket{\theta}+\left(\alpha_{\text{max}}\sin{\theta}-\alpha_{\text{min}}\cos{\theta}\right)\ket{\theta^{\perp}} +\sum_{j\neq \text{min},\text{max}}\alpha_j\ket{j}\nonumber\\
    &= \alpha_{\theta}\left(\theta\right)\ket{\theta}+\alpha_{\theta^{\perp}}\left(\theta\right)\ket{\theta^{\perp}}+\sum_{j\neq \text{min},\text{max}}\alpha_j\ket{j}
    \label{eq:psi_max}
\end{align}
where $\left|\alpha_{\text{max}}\right|^2+\left|\alpha_{\text{min}}\right|^2+\sum_{j\neq\text{min},\text{max}}\left|\alpha_j\right|^2 =1$, as $\ket{\psi_{\text{max}}}$ and all other $\ket{psi_i}$ are normalised in \cref{eq:Proof_C1}.
Here, the $\alpha_{\text{max}}$, $\alpha_{\text{min}}$ and $\alpha_j$ coefficients are arbitrary up to this normalisation, and we have defined
\begin{align}
\label{eq: alpha}
    \begin{split} \alpha_{\theta}\left(\theta\right)&:=\alpha_{\text{max}}\cos{\theta}+\alpha_{\text{min}}\sin{\theta}\\
        \alpha_{\theta^{\perp}}\left(\theta\right)&:= \alpha_{\text{max}}\sin{\theta}-\alpha_{\text{min}}\cos{\theta}
    \end{split}
\end{align}

The same can be done to $\ket{\psi_{\text{min}}}$:
\begin{align}
    \ket{\psi_{\text{min}}}&= \beta_{\text{max}}\ket{\text{max}}+\beta_{\text{min}}\ket{\text{min}} + \sum_{j\neq \text{min}, \text{max}}\beta_j\ket{j}\nonumber\\
    &= \beta_{\theta}\left(\theta\right)\ket{\theta}+\beta_{\theta^{\perp}}\left(\theta^{\perp}\right)\ket{\theta^{\perp}}+\sum_{j\neq \text{min},\text{max}}\beta_j\ket{j},
    \label{eq:psi_min}
\end{align}
where $\left|\beta_{\text{max}}\right|^2+\left|\beta_{\text{min}}\right|^2+\sum_{j\neq\text{min},\text{max}}\left|\beta_j\right|^2 =1$ and

\begin{align}
\label{eq: beta}
    \begin{split} \beta_{\theta}\left(\theta\right)&:=\beta_{\text{max}}\cos{\theta}+\beta_{\text{min}}\sin{\theta}\\
        \beta_{\theta^{\perp}}\left(\theta\right)&:= \beta_{\text{max}}\sin{\theta}-\beta_{\text{min}}\cos{\theta}.
    \end{split}
\end{align}

Replacing the expressions \cref{eq:newbasis_2}, \cref{eq:psi_max} and \cref{eq:psi_min} into the expression \cref{eq:ProofAn_C1} of $C$ leads to 
\begin{align}
\label{eq: C_newbasis}
    C &= \Bigg[\left(a_{\text{max}}\cos{\theta}\alpha_{\theta}\left(\theta\right)+a_{\text{min}}\sin{\theta}\beta_{\theta}\left(\theta\right)\right) \ket{\theta} + \left(a_{\text{max}}\cos{\theta}\alpha_{\theta^{\perp}}\left(\theta\right)+a_{\text{min}}\sin{\theta}\beta_{\theta^{\perp}}\left(\theta\right)\right) \ket{\theta^{\perp}} \nonumber\\
    &\phantom{=} +\sum_j\left(a_{\text{max}}\cos{\theta}\alpha_j+a_{\text{min}}\sin{\theta}\beta_j\right)\ket{j} \Bigg]\bra{\theta}+ \Bigg[...\Bigg]\bra{\theta^{\perp}} +\sum_{j\neq\text{max}, \text{min}} a_j\ket{\psi_{j}}\bra{j}\\
    &:= \ket{\psi'_{\theta}}\bra{\theta}+\ket{\psi'_{\theta^{\perp}}}\bra{\theta^{\perp}}+\sum_{j\neq \text{min},\text{max}} a_j\ket{\psi_{j}}\bra{j}
\end{align}
where $\ket{\psi'_{\theta}}$ and $\ket{\psi'_{\theta^{\perp}}}$ (as defined via the last equation) are not normalized, whereas the states $\ket{\psi_j}$ are. The states $\ket{\psi'_{\theta}}$ and $\ket{\psi'_{\theta^{\perp}}}$ can be normalized:
\begin{equation}
    \ket{\psi_{\theta}}=\frac{\ket{\psi'_{\theta}}}{\sqrt{\braket{\psi'_{\theta}|\psi'_{\theta}}}},\qquad  \ket{\psi_{\theta^{\perp}}}=\frac{\ket{\psi'_{\theta^{\perp}}}}{\sqrt{\braket{\psi'_{\theta^{\perp}}|\psi'_{\theta^{\perp}}}}}
    \label{eq:NormPsiTheta}
\end{equation}
Hence 
\begin{align}
    C &= \sqrt{\braket{\psi'_{\theta}|\psi'_{\theta}}}\ket{\psi_{\theta}}\bra{\theta} + \sqrt{\braket{\psi'_{\theta^{\perp}}|\psi'_{\theta^{\perp}}}}\ket{\psi_{\theta^{\perp}}}\bra{\theta^{\perp}} + \sum_{j\neq \text{min},\text{max}} a_j\ket{\psi_{j}}\bra{j}\nonumber\\
    &= a_{\theta}\left(\theta\right)\ket{\psi_{\theta}}\bra{\theta} + a_{\theta^{\perp}}\left(\theta\right)\ket{\psi_{\theta^{\perp}}}\bra{\theta^{\perp}} + \sum_{j\neq \text{min},\text{max}} a_j\ket{\psi_{j}}\bra{j}
\end{align}
The functions $a_{\theta}\left(\theta\right)=\sqrt{\braket{\psi'_{\theta}|\psi'_{\theta}}}$ and $a_{\theta^{\perp}}\left(\theta\right)=\sqrt{\braket{\psi'_{\theta^{\perp}}|\psi'_{\theta^{\perp}}}}$ are continuous functions with respect to $\theta$, and the relevant states are normalised. We have thus established the first two properties of the lemma statement. For the last property, we show it for the case of $\theta=0$. the $\theta=\frac{\pi}{2}$ is entirely analogous. 

From \cref{eq: C_newbasis}, and using $\theta=0$, we have
\begin{equation}
    \ket{\psi'_{\theta=0}}=a_{\text{max}}\alpha_{\theta}\left(\theta=0\right) \ket{\theta} + a_{\text{max}}\alpha_{\theta^{\perp}}\left(\theta=0\right) \ket{\theta^{\perp}} +\sum_ja_{\text{max}}\alpha_j\ket{j}
\end{equation}
From \cref{eq: alpha}, we have $\alpha_{\theta}\left(\theta=0\right)=\alpha_{\text{max}}$ and $\alpha_{\theta^{\perp}}\left(\theta=0\right)=-\alpha_{\text{min}}$. This gives 
\begin{align}
    \begin{split}
        |a_{\theta}(0)|^2=\braket{\psi'_{\theta}|\psi'_{\theta}}=|a_{\text{max}}|^2\Bigg(|\alpha_{\text{max}}|^2+|\alpha_{\text{min}}|^2+\sum_{j\neq \text{max},\text{min}}|\alpha_j|^2\Bigg)=|a_{\text{max}}|^2.
    \end{split}
\end{align}
The last line follows from the normalisation of $\ket{\psi_{\text{max}}}$ established before. This completes the proof. 
\end{proof}

\begin{figure*}[ht!]
    \centering
    \begin{subfigure}{.33\linewidth}
        \centering
        \scalebox{1}{
            \input{4-Connection/Figs/MTS_1F3B}
        }
        \caption{}
        \label{fig:MTS_1F3B}
    \end{subfigure}%
    \begin{subfigure}{.33\linewidth}
        \centering
        \scalebox{1}{
            \input{4-Connection/Figs/2TO_int}
        }
        \caption{}
        \label{fig:2TO_int}
    \end{subfigure}%
    \begin{subfigure}{.33\linewidth}
        \centering
        \scalebox{1}{
            \input{4-Connection/Figs/2TO_1F3B}
        }
        \caption{}
        \label{fig:2TO_1F3B}
    \end{subfigure}
    \caption{An example where $\left|\mathcal{B}_2\right| = 3$ and $\left|\mathcal{F}_1\right| = 1$, illustrated the steps of the proof of \cref{prop: MTS_2TO}, given in \cref{appendix: proof_2TOtoMTS}. Here we use $\mathcal{B}_2$ in the first step and $\mathcal{F}_1$ in the second step of the proof. (a) The target MTS. (b) From the 2TO $C$, we first elongate the backward-evolving spaces in $\mathcal{B}_2$ to match their relative ordering. (c) We then use a single P-CTC on $S_1\in \mathcal{F}_1$, sending it back to the past according to the time order of the target MTS. }
    \label{fig:1F3B}
\end{figure*}
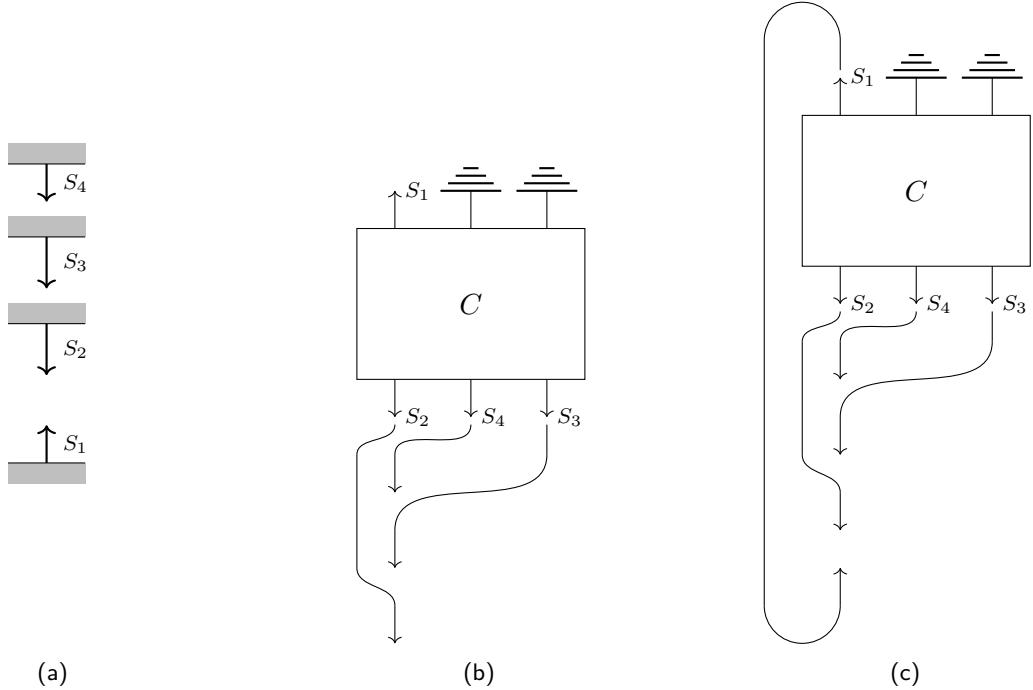

\subsection{Constructing MTS from 2TO: proof of \cref{prop: MTS_2TO}}
\label{appendix: proof_2TOtoMTS}
\MtwoTO*

\begin{proof}
Given any MTS, we can consider a 2TO that is isomorphic to it i.e., has the same coefficients and Hilbert spaces, and directionality of the spaces (backward vs forward) but where all backward evolving spaces are permuted to appear before all forward evolving spaces. In other words, expressed in bra-ket notation, all bras appear before all kets when reading the expression from right to left. Then it is clear that if $\mathcal{B}$ and $\mathcal{F}$ are the set of all backward and forward evolving spaces of the MTS respectively, this 2TO has $|\mathcal{B}|$ backward evolving spaces of dimensions $\{d_S\}_{S\in \mathcal{B}}$ and $|\mathcal{F}|$ forward-evolving spaces of dimensions $\{d_S\}_{S\in \mathcal{F}}$.
We now show that the MTS can be obtained from this isomorphic 2TO using P-CTCs are described in the proposition statement.

In a 2TO, all backward evolving spaces occur before all forward evolving spaces. This means that it is precisely the spaces in $\mathcal{B}_2$ (backward evolving spaces ordered after some forward evolving space) or $\mathcal{F}_1$ (forward evolving spaces ordered before some backward evolving space) that will need to be teleported using P-CTC, as this will involving moving a backward evolving space forward in time or a forward evolving space backward in time to reach the target MTS from the 2TO. We start by arbitrarily picking either $\mathcal{F}_1$ or $\mathcal{B}_2$, suppose we pick $\mathcal{B}_2$.

In the first step, we take our 2TO and change the relative time ordering of the systems in $\mathcal{B}_2$ using free operations (\cref{def: free_op}) to match the time ordering of the corresponding systems in the target MTS. By definition (\cref{def: 2TO2TS}), the obtained object is still a 2TO (see \cref{fig:2TO_int} for an example).

In the second step, we take the 2TO obtained after the first step above and teleport each forward evolving system of $\mathcal{F}_1$ backward in time (according to the time order of the target MTS) using a P-CTC of the same dimension as $S_i$.
By construction, this will reproduce the time order of the target MTS. And we have used $|\mathcal{F}_1|$ number of P-CTCs of dimensions $\{d_S\}_{S\in \mathcal{F}_1}$ to transform the 2TO to the target MTS.

One could apply the same two steps, but instead by picking $\mathcal{F}_1$ at the start. Then we would apply the free operations to systems in $\mathcal{F}_1$ in the first step, and the P-CTCs to systems in $\mathcal{B}_2$ in the second step. In this case, we would use $|\mathcal{B}_2|$ number of P-CTCs of dimensions $\{d_S\}_{S\in \mathcal{B}_2}$ to transform the 2TO to the target MTS. Both options achieve the desired transformation, and thus establish the proposition. 

Depending on the specific MTS, there can be reasons to favour one construction over the other. An illustrated example is shown in \cref{fig:1F3B}. There $\mathcal{F}_1$ has a lower cardinality than $\mathcal{B}_2$. If the total dimensions of the former are also less that the latter, then it is clearly more optimal (i.e., minimises the usage of P-CTCs) to use free operations on the larger set $\mathcal{B}_2$ and P-CTCs on the smaller set $\mathcal{F}_1$.

\end{proof}

%% file: 4-Connection/Figs/MTS_1F3B.tex
\begin{tikzpicture}[
    op/.style={shape= Op, minimum width = \WidthPrePost},
    pre/.style={shape = Pre, minimum width = \WidthPrePost},
    post/.style={shape = Post, minimum width = \WidthPrePost},
    Kop/.style={draw=red, minimum width=\WidthMeasurement, minimum height=\WidthMeasurement},
    t/.style={font=\scriptsize},
    ghost/.style={minimum width=\WidthMeasurement, minimum height=\WidthMeasurement},
    decoration={snake, segment length=4mm, amplitude=0.5mm}
]
\pgfsetmatrixrowsep{0.5cm}
\pgfsetmatrixcolumnsep{0.3cm}
\pgfmatrix{rectangle}{center}{mymatrix}
{\pgfusepath{}}{\pgfpointorigin}{\let\&=\pgfmatrixnextcell}
{
\node[post](1){}; \\
\node[ghost, text=red](11){}; 
\node[post](2){}; \\
\node[ghost, text=red](21){};
\node[post](3){}; \\
\node[ghost, text=red](31){}; \\
\node[pre](4){}; \\
\\
\\
\\
\\
}

\draw[->, thick] (1.south) -- (11.north) node [midway,inner sep=0 cm] (1to11) {};
\draw[->, thick] (2.south) -- (21.north) node [midway,inner sep=0 cm] (2to21) {};
\draw[->, thick] (3.south) -- (31.north) node [midway, inner sep=0 cm] (3to31) {};
\draw[->, thick] (4.north) -- (31.south) node [midway, inner sep=0cm] (4to31) {};

\node[t, right=.5 mm of 4to31](S1){$S_1$};
\node[t, right=.5 mm of 3to31](S2){$S_2$};
\node[t, right=.5 mm of 2to21](S3){$S_3$};
\node[t, right=.5 mm of 1to11](S4){$S_4$};


\end{tikzpicture}

%% file: 4-Connection/Figs/2TO_int.tex
\begin{tikzpicture}[trace/.pic={\draw [thick](-0.4,0)--(0.4,0);\draw [thick](-0.3,0.1)--(0.3,0.1);\draw [thick](-0.2,0.2)--(0.2,0.2);\draw [thick](-0.1,0.3)--(0.1,0.3);}]

\draw (0, 0) -- (3, 0) -- (3, 2) -- (0, 2) -- cycle;
\node at (1.5, 1) {$C$};

\draw[->] (0.5, 0) -- (0.5, -0.5);
\draw[->] (1.5, 0) -- (1.5, -0.5);
\draw[->] (2.5, 0) -- (2.5, -0.5);
\draw[->] (0.5, 2) -- (0.5, 2.5);

\node[font=\scriptsize] at (0.8, -0.5) {$S_2$};
\node[font=\scriptsize] at (1.8, -0.5) {$S_4$};
\node[font=\scriptsize] at (2.8, -0.5) {$S_3$};
\node[font=\scriptsize] at (0.8, 2.5) {$S_1$};

\draw[-] (2.5, -0.6) -- (2.5, -1);
\draw[-] (1.5, -0.6) to[in=90, out=270] (0.5, -1);
\draw[->] (0.5, -1) -- (0.5, -1.5);
\draw[-] (2.5, -1) to[in=90, out=270] (0.5, -2);
\draw[->] (0.5,-2) -- (0.5, -2.5);

\draw[-] (0.5, -0.6) to[in=90, out=270] (0, -1);
\draw[-] (0,-1) -- (0,-2.5);
\draw[-] (0,-2.5) to[in=90, out=270] (0.5, -3);
\draw[->] (0.5, -3) -- (0.5, -3.5);

\draw (2.5,2.5) pic {trace};
\draw (1.5,2.5) pic {trace};
\draw[-] (1.5, 2) -- (1.5, 2.5);
\draw[-] (2.5, 2) -- (2.5, 2.5);

\end{tikzpicture}

%% file: 4-Connection/Figs/2TO_1F3B.tex
\begin{tikzpicture}[trace/.pic={\draw [thick](-0.4,0)--(0.4,0);\draw [thick](-0.3,0.1)--(0.3,0.1);\draw [thick](-0.2,0.2)--(0.2,0.2);\draw [thick](-0.1,0.3)--(0.1,0.3);}]

\draw (0, 0) -- (3, 0) -- (3, 2) -- (0, 2) -- cycle;
\node at (1.5, 1) {$C$};

\draw[->] (0.5, 0) -- (0.5, -0.5);
\draw[->] (1.5, 0) -- (1.5, -0.5);
\draw[->] (2.5, 0) -- (2.5, -0.5);
\draw[->] (0.5, 2) -- (0.5, 2.5);

\node[font=\scriptsize] at (0.8, -0.5) {$S_2$};
\node[font=\scriptsize] at (1.8, -0.5) {$S_4$};
\node[font=\scriptsize] at (2.8, -0.5) {$S_3$};
\node[font=\scriptsize] at (0.8, 2.5) {$S_1$};

\draw[-] (2.5, -0.6) -- (2.5, -1);
\draw[-] (1.5, -0.6) to[in=90, out=270] (0.5, -1);
\draw[->] (0.5, -1) -- (0.5, -1.5);
\draw[-] (2.5, -1) to[in=90, out=270] (0.5, -2);
\draw[->] (0.5,-2) -- (0.5, -2.5);

\draw[-] (-0.5, -4.5) -- (-0.5, 3);
\draw[-] (0.5, 2.6) -- (0.5, 3);
\draw[-] (0.5, -0.6) to[in=90, out=270] (0, -1);
\draw[-] (0,-1) -- (0,-2.5);
\draw[-] (0,-2.5) to[in=90, out=270] (0.5, -3);
\draw[->] (0.5, -3) -- (0.5, -3.5);
\draw[->] (0.5,-4.5) -- (0.5,-4);

\draw[-] (0.5, 3) arc (0:180: 0.5);
\draw[-] (-0.5, -4.5) arc (180:360:0.5);

\draw (2.5,2.5) pic {trace};
\draw (1.5,2.5) pic {trace};
\draw[-] (1.5, 2) -- (1.5, 2.5);
\draw[-] (2.5, 2) -- (2.5, 2.5);

\end{tikzpicture}

%% file: 6-Appendix/PartialOrder.tex
\section{Isomorphism, free operations and partial order on MTS: proof of \cref{thm:partialorder}}

\label{appendix: partialorder}

Here we provide the formal definitions of the concept of \emph{isomorphisms} and \emph{free operations} used in the main text, particular in \cref{thm:partialorder}, and provide a proof of this theorem.  We begin by formally defining the concept of isomorphism between multi-time objects motivated in the main text.

\begin{definition}[Isomorphisms between MTSs]
\label{def: time_ord_iso}
   Let $M_1$ and $M_2$ be two MTS which are associated to the same set $\{S_1,...,S_n\}$ of systems corresponding to tensor factors of their MT Hilbert space, where each system $S_i$ has the same direction of evolution in $M_1$ and $M_2$. Let each $S_i$ be associated to a time label $t_i$ in $M_1$ and $t'_i$ in $M_2$. Then we say that $M_1$ and $M_2$ are \emph{isomorphic} if $M_1$ and $M_2$ only differ by an isomorphism $$\mathcal{R}: \{t_1,...,t_n\} \mapsto \{t'_1,...,t'_n\}$$ on the time labels.

\end{definition}

Notice that when an MTS is expressed in the bra ket notation, such an isomorphism preserves the coefficients and  can be seen as permuting the order of the bras and kets, recovering the intuitions discussed in \cref{subsubsec:2TO-MTS}.
Since bras and kets are associated to different times in the MT framework, such isomorphisms between MTS when viewed more operationally can generally involve CTCs that move spaces forward or backwards in time. Free operations, defined below, correspond to those subset of isomorphisms which can be operationally implemented without CTCs, and preserve the direction of (time) evolution of systems in that MTS.

\begin{definition}[Free operations on an MTS]
\label{def: free_op}
Free operations on an MTS $M$, correspond to transformations that take one of the following two forms.

\begin{itemize}
    \item {\bf Stretching forward evolving systems into the future} Transformation of an MTS $M$ with a forward evolving system $S$ at time $t$ to any MTS $M'$ which is identical to $M$ except for the time label on $S$ being transformed to a later time $t'>t$. 
        \item {\bf Stretching backward evolving systems into the past} Transformation of an MTS $M$ with a backward evolving system $S$ at time $t$ to any MTS $M'$ which is identical to $M$ except for the time label on $S$ being transformed to an earlier time $t'<t$.  
\end{itemize}
\end{definition}
Notice that both the above free transformations can be achieved by composing the original MTS $M$ with a 2TO which models an identity channel between the two times $t$ and $t'$, with backward/forward moving systems $S$ contracted to the output/input of this channel.

\PartialOrder*

\begin{proof}
    By \cref{def: 2TO2TS}, any 2TS $M_{2TS}$ has all of its backward evolving systems occurring after all of its forward evolving systems. For any other isomorphic MTS $M$ which has a certain number of backward evolving systems ordered before some of the forward evolving systems, free operations can be used on $M_{2TS}$ to evolve those backward evolving systems to the past, to ensure the correct relative time order of any isomorphic MTS $M$.
    When all the backward evolving systems have been moved to the past, to occur before any forward evolving system, an isomorphic 2TO $M_{2TO}$ is obtained. 
    
    Going from $M_{2TO}$ to any $M$ is not possible using free operations alone as it requires chronology breaking operations which would move some of the backward evolving systems to the future, to be ordered later than at least one forward evolving system (otherwise it still remains a 2TO by \cref{def: 2TO2TS}). Following the same logic, going from any $M$ to an isomorphic $M_{2TS}$ solely through free operations is not possible. This yields the strict partial order of \cref{eq: MTS_partialorder}, using \cref{def:partialorder}.

    Moreover this is a partial order and not a total order, as there exist isomorphic MTS that cannot be mapped to one another (in either direction) solely using free operations as shown in the following explicit example.

Let $S_1$, $S_3$ be forward evolving systems, $S_2$ be a backward evolving system and consider the following two MTS $M_1$ and $M_2$ defined on these systems as depicted in \cref{fig:thm62}. $M_1$ has the time-order $S_1$, then $S_2$ and then $S_3$ and $M_2$ has the order $S_3$, $S_2$ and then $S_1$. Notice that $M_1$ has a 2TO structure between $S_2$ and $S_3$ but $M_2$ has a 2TS structure between these. It follows from by applying the same arguments as above, that it is not possible to transform $M_1$ to $M_2$ freely. Further, $M_2$ has a 2TO structure between $S_2$ and $S_1$ but $M_1$ has a 2TS structure between these. Similarly, this implies that $M_2$ cannot be transformed to $M_1$ freely. Together, this yields $M_1\not\preceq \not\succeq M_2$.

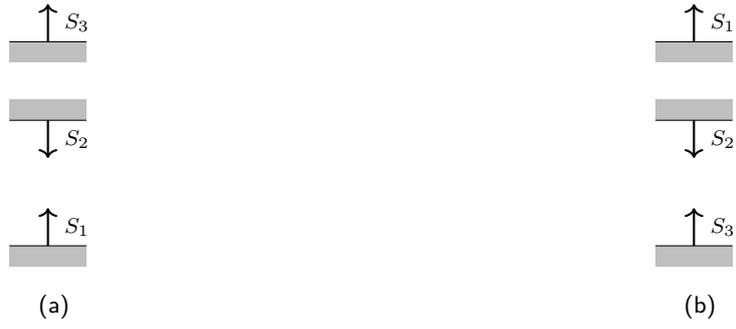
\begin{figure*}[ht!]
    \centering
    \begin{subfigure}{.5\linewidth}
        \centering
        \scalebox{1}{
            \input{6-Appendix/Figs/thm6_2}
        }
        \caption{}
        \label{fig:thm6_2}
    \end{subfigure}%
    \begin{subfigure}{.5\linewidth}
        \centering
        \scalebox{1}{
            \input{6-Appendix/Figs/thm6_2_1}
        }
        \caption{}
        \label{fig:thm6_2_1}
    \end{subfigure}
    \caption{The MTS (a) $M_1$ and (b) $M_2$ defined in the last paragraph, which are isomorphic to each other. It is not possible to transform $M_1$ to $M_2$ via free operations or vice-versa and hence $M_1\not\preceq \not\succeq M_2$.}
    \label{fig:thm62}
\end{figure*}

\end{proof}

%% file: 6-Appendix/Figs/thm6_2.tex
\begin{tikzpicture}[
    op/.style={shape= Op, minimum width = \WidthPrePost},
    pre/.style={shape = Pre, minimum width = \WidthPrePost},
    post/.style={shape = Post, minimum width = \WidthPrePost},
    Kop/.style={draw=red, minimum width=\WidthMeasurement, minimum height=\WidthMeasurement},
    t/.style={font=\scriptsize},
    ghost/.style={minimum width=\WidthMeasurement, minimum height=\WidthMeasurement},
    decoration={snake, segment length=4mm, amplitude=0.5mm}
]
\pgfsetmatrixrowsep{0.5cm}
\pgfsetmatrixcolumnsep{0.3cm}
\pgfmatrix{rectangle}{center}{mymatrix}
{\pgfusepath{}}{\pgfpointorigin}{\let\&=\pgfmatrixnextcell}
{
\node[ghost, text=red](11){}; \\
\node[pre](1){}; \\
\node[post](2){}; \\
\node[ghost, text=red](21){}; \\
\node[pre](3){}; \\
}

\draw[->, thick] (1.north) -- (11.south) node [midway,inner sep=0 cm] (1to11) {};
\draw[->, thick] (2.south) -- (21.north) node [midway,inner sep=0 cm] (2to21) {};
\draw[->, thick] (3.north) -- (21.south) node [midway, inner sep=0 cm] (3to31) {};

\node[t, right=.5 mm of 3to31](S1){$S_1$};
\node[t, right=.5 mm of 2to21](S2){$S_2$};
\node[t, right=.5 mm of 1to11](S3){$S_3$};


\end{tikzpicture}

%% file: 6-Appendix/Figs/thm6_2_1.tex
\begin{tikzpicture}[
    op/.style={shape= Op, minimum width = \WidthPrePost},
    pre/.style={shape = Pre, minimum width = \WidthPrePost},
    post/.style={shape = Post, minimum width = \WidthPrePost},
    Kop/.style={draw=red, minimum width=\WidthMeasurement, minimum height=\WidthMeasurement},
    t/.style={font=\scriptsize},
    ghost/.style={minimum width=\WidthMeasurement, minimum height=\WidthMeasurement},
    decoration={snake, segment length=4mm, amplitude=0.5mm}
]
\pgfsetmatrixrowsep{0.5cm}
\pgfsetmatrixcolumnsep{0.3cm}
\pgfmatrix{rectangle}{center}{mymatrix}
{\pgfusepath{}}{\pgfpointorigin}{\let\&=\pgfmatrixnextcell}
{
\node[ghost, text=red](11){}; \\
\node[pre](1){}; \\
\node[post](2){}; \\
\node[ghost, text=red](21){}; \\
\node[pre](3){}; \\
}

\draw[->, thick] (1.north) -- (11.south) node [midway,inner sep=0 cm] (1to11) {};
\draw[->, thick] (2.south) -- (21.north) node [midway,inner sep=0 cm] (2to21) {};
\draw[->, thick] (3.north) -- (21.south) node [midway, inner sep=0 cm] (3to31) {};

\node[t, right=.5 mm of 3to31](S1){$S_3$};
\node[t, right=.5 mm of 2to21](S2){$S_2$};
\node[t, right=.5 mm of 1to11](S3){$S_1$};


\end{tikzpicture}